\documentclass[acmsmall,screen,dvipsnames,
]{acmart}
\pdfoutput=1

\makeatletter
\def\@acmplainindent{0pt}
\def\@acmdefinitionindent{0pt}
\def\@proofindent{\noindent}

\usepackage{hyperref}
\usepackage{enumitem}
\usepackage{cancel}
\usepackage{cleveref}
\usepackage{graphicx}
\usepackage{listings}
\usepackage{mathpartir}
\usepackage{mathtools}
\usepackage{relsize}
\usepackage{xfrac}
\usepackage{rotating}
\usepackage{stmaryrd}
\usepackage{subcaption}
\usepackage{lscape}
\usepackage{xcolor}
 \hypersetup{
    colorlinks,
    linkcolor={red!50!black},
    citecolor={blue!50!black},
    urlcolor={blue!80!black}
}
\usepackage{wrapfig}
\usepackage{scalerel}
\usepackage{thmtools} 
\usepackage{thm-restate}
\usepackage[all, cmtip]{xy}
\usepackage{framed}
\usepackage{trimclip}
\usepackage{xr}

\usepackage{tikz}
\usetikzlibrary{shapes,arrows,positioning}

\def\extended{0}
\usepackage{olmacros}

\ifx\extended\undefined
\fi

\ifx\apponly\undefined\else
\usepackage[29-69]{pagesel}
\fi

\definecolor{codegreen}{rgb}{0,0.6,0}
\definecolor{codegray}{rgb}{0.5,0.5,0.5}
\definecolor{codepurple}{rgb}{0.58,0,0.82}
\definecolor{backcolour}{rgb}{0.95,0.95,0.92}

\lstdefinelanguage{gcl}{%
  keywords={def,let,in,while,do,repeat,if,then,else,skip,return},
  comment=[l]{//},
  commentstyle = {\color{codecomment}},
}[keywords,comments,strings]

\lstdefinestyle{mystyle}{
    commentstyle=\color{codegreen},
    keywordstyle=\color{magenta}\bfseries\ttfamily,
    numberstyle=\tiny\color{codegray},
    stringstyle=\color{codepurple},
    basicstyle=\itshape,
    breakatwhitespace=false,         
    breaklines=true,                 
    captionpos=b,                    
    keepspaces=true,                 
    showspaces=false,                
    showstringspaces=false,
    showtabs=false,                  
    tabsize=2
}
\setcopyright{cc}
\setcctype{by}
\acmDOI{10.1145/3776651}
\acmYear{2026}
\acmJournal{PACMPL}
\acmVolume{10}
\acmNumber{POPL}
\acmArticle{9}
\acmMonth{1}
\received{2025-07-02}
\received[accepted]{2025-11-06}


\begin{document}
\title{Probabilistic Concurrent Reasoning in Outcome Logic: Independence, Conditioning, and Invariants}
%
%
\author{Noam Zilberstein}
\email{noamz@cs.cornell.edu}
\orcid{0000-0001-6388-063X}
\affiliation{%
  \institution{Cornell University}
  \country{USA}
}

\author{Alexandra Silva}
\email{alexandra.silva@cornell.edu}
\orcid{0000-0001-5014-9784}
\affiliation{%
  \institution{Cornell University}
  \country{USA}
}

\author{Joseph Tassarotti}
\email{jt4767@nyu.edu}
\orcid{0000-0001-5692-3347}
\affiliation{%
  \institution{New York University}
  \country{USA}
}

\begin{abstract}
Although randomization has long been used in distributed computing, formal methods for reasoning about probabilistic concurrent programs have lagged behind. No existing program logics can express specifications about the full \emph{distributions of outcomes} resulting from programs that are both probabilistic and concurrent. To address this, we introduce \emph{Probabilistic Concurrent Outcome Logic} (\pcol), which incorporates ideas from concurrent and probabilistic separation logics into Outcome Logic to introduce new compositional reasoning principles. At its core, \pcol reinterprets the rules of Concurrent Separation Logic in a setting where separation models probabilistic independence, so as to compositionally describe joint distributions over variables in concurrent threads. Reasoning about outcomes also proves crucial, as case analysis is often necessary to derive precise information about threads that rely on randomized shared state. We demonstrate \pcol on a variety of examples, including to prove almost sure termination of unbounded loops.
\end{abstract}

\maketitle              

\section{Introduction}

Randomization is an important tool in concurrent and distributed computing.  Concurrent algorithms can be made more efficient using randomization \cite{rabin1980n-process,rabin1982choice,morris1978counting} and some distributed synchronization problems have no deterministic solution \cite{fischer1985impossibility,lehmann1981advantages}. But despite the prevalence of randomization in concurrent computing over the last several decades, formal methods for such programs are limited.
The mixture of \emph{computational effects} in probabilistic concurrent programs is a major source of difficulty in developing verification techniques; random choice is introduced by sampling operations and nondeterminism arises from scheduling the concurrent threads. These two computational effects do not compose in standard ways \cite{varacca_winskel_2006}, so even just describing the semantics of such programs requires specialized models \cite{jifeng1997probabilistic,mciver2005abstraction,zilberstein2025denotational,zilberstein2025demonic}.

In this paper, we introduce \emph{Probabilistic Concurrent Outcome Logic} (\pcol), a logic for reasoning about programs that are both probabilistic and concurrent.
In \pcol, preconditions and postconditions are not just assertions about a single program state.
Instead, they describe the \emph{distribution of possible outcomes} that can arise from executing the program.
A key challenge is that, in the concurrent setting, different orderings of threads can give rise to different distributions over program behaviors.
To address this, \pcol takes inspiration from the recently introduced Demonic Outcome Logic (\dol) \cite{zilberstein2025demonic}, which supports reasoning about sequential probabilistic programs that additionally have a nondeterministic choice operator resolved by an adversary.
Different nondeterministic choices can cause different distributions of behaviors in these programs, just as different thread interleavings can cause different distributions in the concurrent setting.

However, while \dol's approach to describing the space of possible distributions provides a basis for \pcol, reasoning about the nondeterminism that arises from concurrent scheduling is substantially more complicated than reasoning about a choice operator.
With a choice operator, nondeterminism is \emph{localized} to the points where the operator is used, whereas in a concurrent program, every step can potentially involve nondeterminism from thread interleaving.
Reasoning explicitly about nondeterminism at every step is intractable and non-compositional.

To recover compositional reasoning, \pcol{} incorporates ideas from various \emph{separation logics}. Concurrent Separation Logic (\csl) uses disjointness of resources to ensure that concurrent computations only interact in controlled ways, so that each thread can be analyzed on its own \cite{csl,brookes2004semantics}. Probabilistic Separation Logics (\psl) use the notions of independence and conditioning to reason about the interaction between randomness and control flow \cite{psl,li2023lilac,bao2025bluebell,bao2021bunched,bao2022separation,yan2025combining}. \csl and \psl achieve compositional reasoning in concurrent and probabilistic settings, respectively, so combining their reasoning principles appears to be a natural way to derive a compositional logic for the combination of both effects.
However, as we will see in \Cref{sec:overview}, such a combination is challenging to achieve because the metatheories of the two logics are highly specialized to their respective domains, and a direct combination of their rules would not be sound.
This paper develops the metatheory in a more complex semantic domain, where concurrency and probabilistic computation can coexist.
As a result, \pcol is the first logic to combine all of the following features:

\subsubsection*{Compositional Concurrency Reasoning}
\pcol supports \emph{compositional} concurrency reasoning, meaning that each thread in a concurrent program can be analyzed in isolation, without considering all the possible interleavings or behaviors of the scheduler. Similar to Concurrent Separation Logic (\csl), compositionality stems from \emph{separation}---as long as two threads operate on their own portions of memory, they cannot interfere with each other. In addition, shared state is handled via \emph{resource invariants}, properties about shared state that remain true at every step.

\subsubsection*{Compositional Probabilistic Reasoning}
In a probabilistic context, separation of memory footprints is not sufficient; while it can tell us how the local variables of each thread are distributed, it does not give us the \emph{joint distribution} over the entire global memory.
In response, we use a \psl-style model where separation additionally models \emph{probabilistic independence}  \cite{psl}. Going beyond prior work, our parallel composition rule guarantees that the scheduler cannot introduce any probabilistic correlation between the local states of each thread.

\subsubsection*{Compositional Outcome Reasoning}
When concurrent threads depend on randomized shared state, it is often necessary to do case analysis over the possible values of that shared state in order to capture the probabilistic correlation between the threads (\eg see \Cref{sec:overview-outcomes}).
In the style of Demonic Outcome Logic (\dol) \cite{zilberstein2025demonic}, \pcol supports compositional reasoning about the outcomes generated via both probabilistic branching and the nondeterministic behavior of the scheduler. 
But unlike \dol---which does not support separation---case analysis must be done with care, as it can invalidate the independence guarantees of the \ruleref{Frame} rule. Compared to prior outcome logics, the \emph{outcome conjunction} of \pcol has a new measure theoretic foundation.

\subsubsection*{Unbounded Looping and Almost Sure Termination}
Our logic includes rules for establishing \emph{almost sure termination}---termination with probability 1---for unbounded loops. This goes beyond the capabilities of all prior separation logics that model separation with probabilistic independence, which either have only bounded looping constructs (\eg for loops) or require loops to always terminate (which must be established externally to the logic). Unbounded looping is important in randomized concurrent programs, as such programs often only achieve the desired distribution of outcomes in the limit, and not after a bounded number of steps (\eg \Cref{sec:von-neumann}).

\medskip
\noindent We begin in \Cref{sec:overview} with an overview of the technical challenges and design of \pcol. Next, in \Cref{sec:semantics} we outline the programming language and semantic model that we will use. The logic and inference rules are defined in \Cref{sec:model,sec:logic}. We demonstrate the capabilities of \pcol on four case studies in \Cref{sec:examples}. Finally, we conclude by discussing related work and future directions in \Cref{sec:related,sec:discussion}. Omitted proofs and details are given in the appendix \cite{zilberstein2025probabilistic}.

\section{Overview: Familiar Reasoning Principles in a New Setting}
\label{sec:overview}

In this paper, we develop a logic for verifying the correctness of randomized concurrent imperative programs, written in
a language that includes control flow operations (if statements and while loops), parallel composition $C_1 \parallel C_2$, and random sampling $x \samp d$.
A major hurdle in concurrency analysis is that the semantics is \emph{non-compositional}; two programs can have completely different behavior when run in parallel than they do when run in isolation. Enumerating all possible interleavings of the threads is not a viable strategy, so abstractions must be introduced for sound compositional analysis. In this section, we explore the mechanisms that enable compositional reasoning about probabilistic concurrent programs in \pcol, including the challenges that arise with shared state.

\subsection{Concurrent Separation Logic meets Probabilistic Separation Logic}
\label{sec:overview-csl-psl}

Concurrent Separation Logic (\csl) achieves compositionality via \emph{separation} \cite{csl,brookes2004semantics}---if two threads act on disjoint memory regions, then their behavior will not change when run in parallel. This idea is encapsulated by the \hyperlink{rule:par1}{\textsc{Par}} rule, where the \emph{separating conjunction} $\varphi\sep\psi$ means that the machines's memory cells can be divided to satisfy $\varphi$ and $\psi$ individually. In addition, the \hyperlink{rule:frame1}{\textsc{Frame}} rule guarantees that threads cannot interfere with memory outside of their local state.
\begin{mathpar}
\hypertarget{rule:par1}{
  \inferrule{
    \triple{\varphi_1}{C_1}{\psi_1}
    \\
    \triple{\varphi_2}{C_2}{\psi_2}
  }{
    \triple{\varphi_1\sep\varphi_2}{C_1\parallel C_2}{\psi_1\sep \psi_2}
  }{\textsc{Par}}
}

\hypertarget{rule:frame1}{
\inferrule{
  \triple{\varphi}C{\psi}
}{
  \triple{\varphi\sep\vartheta}C{\psi\sep\vartheta}
}{\textsc{Frame}}
}
\end{mathpar}
Now suppose that $C_1$ and $C_2$ in the \hyperlink{rule:par1}{\textsc{Par}} rule are probabilistic programs. For example, if we flip two coins and store the results in the variables $x$ and $y$, then the respective variables will be distributed according to Bernoulli distributions with parameter $\frac12$, as shown in the following specifications, where $\sure P$ means that $P$ holds with probability 1 (almost surely),  ${x\mapsto -}$ means that the current thread has permission to read and write $x$, and $x \sim d$ means that $x$ is distributed according to $d$.
\[
  \triple{\sure{x\mapsto -}}{x \samp \bern{\tfrac12}}{x \sim \bern{\tfrac12}}
  \qquad\quad
  \triple{\sure{y\mapsto -}}{y \samp \bern{\tfrac12}}{y \sim \bern{\tfrac12}}
\]
Composing these programs in parallel, we would ideally want to derive a specification that dictates not only how $x$ and $y$ are distributed, but also their \emph{joint distribution}. In this case, $x$ and $y$ are \emph{probabilistically independent}, meaning that each outcome (\eg $x$ and $y$ are both 1) occurs with probability $\frac12\cdot\frac12=\frac14$.
Thus, it is natural to consider using the interpretation of separation proposed in Probabilistic Separation Logic (\psl) \cite{psl}, in which $\varphi\sep\psi$ states that the events described by $\varphi$ and $\psi$ are probabilistically independent.
By treating separation as \emph{both} disjointness of memory \emph{and} probabilistic independence, one might hope to validate a probabilistic interpretation of the \hyperlink{rule:par1}{\textsc{Par}} rule, with which we could derive the following specification.
\[
  \triple{\sure{x\mapsto -} \sep \sure{y\mapsto -}}{x \samp \bern{\tfrac12} \parallel y \samp \bern{\tfrac12}}{x \sim \bern{\tfrac12} \sep y \sim \bern{\tfrac12}}
\]
Intuitively, the probabilistic interpretation of \hyperlink{rule:par1}{\textsc{Par}} in this instance is justifiable because $C_1$ and $C_2$ execute without interaction, so there is no correlation between their random behaviors.
The lack of correlation is due to the fact that there is no shared state, and therefore the nondeterminism introduced by the interleaving behavior of the scheduler is not observable in any way.

Bigger challenges arise when the threads \emph{do} share state, making the nondeterministic behavior of the scheduler observable.
\csl handles shared state with \emph{resource invariants}---threads can interact with shared state as long as the invariant $I$ is preserved by every atomic step. More specifically, \csl provides the following inference rules: a more general \hyperlink{rule:par2}{\textsc{Par}} rule allows the invariant resources to be used in both threads; the \hyperlink{rule:atom2}{\textsc{Atom}} rule \emph{opens} the invariant, as long as the program is a single atomic command $a$; and \hyperlink{rule:share2}{\textsc{Share}} allocates an invariant that is true before and after the program execution.
\begin{mathpar}\small
\hypertarget{rule:par2}{
  \inferrule{
    I\vdash\triple{\varphi_1}{C_1}{\psi_1}
    \quad
    I\vdash\triple{\varphi_2}{C_2}{\psi_2}
  }{
    I\vdash\triple{\varphi_1\sep\varphi_2}{C_1\parallel C_2}{\psi_1\sep \psi_2}
  }{\textsc{Par}}
}

\hypertarget{rule:atom2}{
\inferrule{
  \vdash\triple{\varphi\sep\sure I}a{\psi\sep \sure I}
}{
  I\vdash\triple\varphi{a}\psi
}{\textsc{Atom}}
}

\hypertarget{rule:share2}{
\inferrule{
  I \vdash\triple{\varphi}C{\psi}
}{
  \vdash\triple{\varphi\sep\sure{I}}C{\psi\sep\sure I}
}{\textsc{Share}}
}
\end{mathpar}
Under a probabilistic interpretation of $\sep$, this stronger \hyperlink{rule:par2}{\textsc{Par}} rule is valid when the shared state described by the invariant $I$ is \emph{deterministic}, for example, in the following program, where $z$ has a fixed value, which is read from both threads.
\[
  z \mapsto 1 \vdash\triple{\sure{x\mapsto-}\sep\sure{y\mapsto-}}{x\samp \bern{\tfrac{z}2} \parallel y \samp \bern{\tfrac{z}2}}{(x\sim\bern{\tfrac12}) \sep (y \sim\bern{\tfrac12})}
  \]
However, with nondeterministic or randomized shared state, reasoning about programs becomes more complicated because correlations can be introduced in very subtle ways. In the remainder of this section, we will see a few such representative scenarios, and how they are handled in \pcol.

\subsection{Handling Randomized Shared State with Outcome Logic}
\label{sec:overview-outcomes}

When shared state is randomized, compositional reasoning about outcomes becomes essential.
As an example, consider the following program, where $x$ and $y$ read from the shared randomized variable $z$ in different threads.
\begin{equation}\label{eq:conditioning}
  z \samp \bern{\tfrac12} \fatsemi ( x\coloneqq z \parallel y \coloneqq 1-z )
\end{equation}
Here, $x$ and $y$ are clearly not independently distributed, since both are derived from $z$. So it seems dubious that \ruleref{Par} could be used to reason about this program.
The key observation is that $x$ and $y$ are \emph{conditionally} independent on $z$, so we can compositionally reason about the threads by first breaking down the outcomes of the sampling operation such that $z$ is deterministic in each case.
We draw inspiration from the logics \DIBI \cite{bao2021bunched}, Lilac~\cite{li2023lilac}, and Bluebell~\cite{bao2025bluebell}, which include constructs to reason about conditioning, but which are not sufficient for the kind of analysis needed here.
\DIBI's model of separation is too coarse (see \Cref{sec:ov-precise}), Lilac has no rule for case analysis over conditioning modalities, and Bluebell's case analysis rule (\textsc{c-wp-swap}) has side conditions which would preclude adapting it for use in parallel composition.

To support \emph{both} case analysis \emph{and} parallel composition, in \pcol we introduce an \emph{outcome conjunction} $\bigoplus_{X\sim d}\varphi$, which binds a new logical variable $X$ that can be referenced in $\varphi$, and is distributed according to $d$, \eg $z \sim \bern{\frac12}$ is syntactic sugar for $\bigoplus_{Z \sim\bern{\frac12}} \sure{z\mapsto Z}$.
Outcome conjunctions feature in prior logics \cite{outcome,zilberstein2025outcome,zilberstein2025demonic,zilberstein2024outcome,zhang2024quantitative}, but in this paper we use a new measure theoretic foundation based on \emph{direct sums} \cite{fremlin2001measure}, which interacts well with separation, and represents conditional probabilities via Bayes' Law.
As a result, the \ruleref{Split} rule---shown below---is admissible in \pcol, which is critical for concurrency reasoning.\footnote{Our approach of using direct sums has tradeoffs: the outcome conjunction $\bigoplus$ is fundamentally discrete, whereas other logics, notably Lilac~\cite{li2023lilac}, have continuous conditioning modalities. Despite this restriction, our examples will illustrate the expressive power of the logic, as many uses of randomization in concurrent and distributed programs only require discrete distributions. See \Cref{sec:related} for a deeper comparison with other logics.}
\[
  \ruledef{Split}{
    I\vdash\triple{\varphi}C{\psi}
  }{\textstyle
    I\vdash\triple{\bigoplus_{X\sim d}\varphi}C{\bigoplus_{X\sim d}\psi}
  }
\]
In the premise of \ruleref{Split}, $X$ is unbound, and therefore it is implicitly universally quantified. So, the rule allows us to partition the sample space according to $d$, reason about the program as if $X$ is a deterministic value, and then once again bind $X$ under an outcome conjunction in the conclusion.

Returning to Example (\ref{eq:conditioning}), by considering each possible outcome of the sampling operation, we end up in a situation where $z$ is deterministic, and therefore the \hyperlink{rule:par2}{\textsc{Par}} rule can apply.
Indeed, the \ruleref{Split} rule is the missing piece that we need. After executing the sampling operation, we get the assertion $(z \sim \bern{\frac12})\sep\sure{x\mapsto - \sep y\mapsto -}$, which is equivalent to $\bigoplus_{Z\sim\bern{\frac12}} \sure{z \mapsto Z \sep x\mapsto-\sep y\mapsto-}$. So, to analyze the parallel composition, all we have to do is apply \ruleref{Split} to make $z$ deterministic, and then allocate the invariant $z \mapsto Z$, stating that $z$ has some fixed---but universally quantified---value.
The derivation is sketched below and shown fully in \Aref{app:cond-ind}.
\[\small
\inferrule*[right=\rulereff{Split}]{
  \inferrule*[Right=\rulereff{Share}]{
    \inferrule*[Right=\rulereff{Par}]{
      \inferrule*{}{
        z\mapsto Z \vdash\triple{\sure{x\mapsto-}}{x \coloneqq z}{\sure{x\mapsto Z}}
      }
      \quad
      \inferrule*{}{
        z\mapsto Z \vdash\triple{\sure{y\mapsto-}}{y \coloneqq 1-z}{\sure{y\mapsto 1-Z}}
      }
    }{
      z\mapsto Z \vdash\triple{\sure{x\mapsto-\sep y\mapsto-}}{x \coloneqq z \parallel y \coloneqq 1-z}{\sure{x\mapsto Z\sep y\mapsto1-Z}}
    }
  }{
    \vdash\triple{\sure{z \mapsto Z\sep x\mapsto-\sep y\mapsto-}}{x \coloneqq z \parallel y \coloneqq 1-z}{\sure{z \mapsto Z\sep x\mapsto Z\sep y\mapsto1-Z}}
  }
}{
  \vdash\triple{\smashoperator{\bigoplus_{Z\sim\bern{\frac12}}} \sure{z \mapsto Z \sep x\mapsto-\sep y\mapsto-}}{x \coloneqq z \parallel y \coloneqq 1-z}{\smashoperator{\bigoplus_{Z\sim\bern{\frac12}}} \sure{z \mapsto Z\sep x\mapsto Z\sep y\mapsto1-Z}}
}
\]
Ultimately, we conclude that $x$ and $y$ are independent inside the scope of the $\bigoplus_{Z\sim\bern{\frac12}}$---where they are deterministic---while still recording exactly how they are probabilistically correlated. This pattern arises frequently, \eg in \Cref{sec:shuffle}, where we prove the correctness of a concurrent shuffling algorithm in \pcol.

\subsection{Taming Nondeterminism with Precise Assertions}
\label{sec:ov-precise}

In the previous section, we saw a new form of compositional reasoning, but it was limited to scenarios where the shared state could be made deterministic via case analysis. In other words, the nondeterministic scheduling order could not affect the outcome of the program.
The effects of scheduling become observable when threads \emph{mutate} shared state, because different interleavings can cause the shared variables to take on different values throughout the program execution.

In fact, treating the scheduler adversarially, shared state provides opportunities for the scheduler to introduce probabilistic correlations in unexpected ways. For example, in the following program, the scheduler can force $x$ and $y$ to be equal by first executing the sampling operation; then, with the value of $x$ fixed, choosing an order for the writes to $y$ so that $y \coloneqq 1-x$ is first, and $y \coloneqq x$ is last. 
\begin{equation}\label{eq:coinflip-game}
  \left( x \samp \bern{\tfrac12} \fatsemi y \coloneqq 0 \right) \;\; \mathlarger{\mathlarger\parallel} \;\;  y\coloneqq 1
\end{equation}
In that case, $x$ and $y$ are both distributed according to $\bern{\frac12}$, but are certainly not \emph{independently} distributed.
To ensure that the correlations between $x$ and $y$ do not invalidate the \ruleref{Par} rule in \pcol, we additionally require the postconditions of each thread to be \emph{precise} (\Cref{def:precise}), essentially meaning that they exactly determine the probability of each event. Any pure assertion $\sure P$ is precise---$P$ occurs with probability exactly $1$---as is $x\sim d$, since $d$ dictates the probability of $\sure{x \mapsto v}$ for all $v$.
Given that, $(x\sim \bern{\frac12})\sep \sure{y \in \{0, 1\}}$ is a valid---and precise---postcondition for (\ref{eq:coinflip-game}).

It may be surprising that $(x\sim \bern{\frac12})\sep \sure{y \in \{0, 1\}}$ holds; as we just saw, $x$ and $y$ are not necessarily independently distributed. However, \pcol uses a \emph{probability-space-as-a-resource} model of separation, first explored by \citet{li2023lilac}, to make separation more flexible. That is, we care only about independence of \emph{measurable events} and not \emph{samples}. In the above case, we need only measure the probability of $\sure{y\in\{0,1\}}$, which occurs with probability 1, and nothing finer such as $\sure{y\mapsto 0}$ and $\sure{y\mapsto 1}$. Independence is therefore trivial, \eg we get that $x \mapsto 1$ and $y \in \{0, 1\}$ with probability $\frac12\cdot1 = \frac12$.
While it is difficult to compositionally determine the exact set of possible joint distributions resulting from different scheduling behaviors, \pcol neatly abstracts away those semantic considerations using simple syntactic checks on the postcondition.

\subsection{Weak Separation and Case Analysis over Shared State}

Examples (\ref{eq:conditioning}) and (\ref{eq:coinflip-game}) illustrate how \pcol combines both the disjointness and independence interpretations of separation to obtain a probabilistic \hyperlink{rule:par1}{\textsc{Par}} rule. 
However, in some cases, the combination of disjointness and independence is too strong and does not hold, because the scheduler \emph{can} induce correlations between concurrent threads through shared state.
While we can sometimes \emph{coarsen} what is measurable to recover independence---as we did with $y$ in (\ref{eq:coinflip-game})---it is not always possible.

In response, \pcol uses a second form of separating conjunction, which we call \emph{weak separation}, written $\varphi \osep \psi$, which only requires that $\varphi$ and $\psi$ hold for disjoint state, and need not be probabilistically independent.
Weak separation allows us to still recover some of the benefits of separation logic for reasoning about disjointness, even when we cannot expect independence to hold.
To see an example, consider a thread running the following command, in which control flow depends on a variable $x$ that is part of shared state and may be written by another thread:
\[
  C \quad \triangleq \quad x'\coloneqq x \fatsemi {\iftf{x'}{z \samp \bern{\tfrac12}}{z \coloneqq 1}}
\]
To analyze the program above, we need to do case analysis on the value of $x$ at the moment that it is read. Then, we conclude that regardless of $x$'s value, $z$ is distributed according to a Bernoulli distribution with parameter \emph{at least} $\frac12$, motivating the following triple.
\begin{equation}\label{eq:weak-triple}
  x\mapsto 0 \vee x\mapsto 1 \vdash\triple{\sure{z\mapsto- \sep x' \mapsto -}}{C}{ \exists X\ge \tfrac12.\ z \sim \bern{X}}
\end{equation}
Triple (\ref{eq:weak-triple}) is indeed valid on its own, but it is not compatible with the \hyperlink{rule:frame1}{\textsc{Frame}} rule that we saw in \Cref{sec:overview-csl-psl}. To see why, consider the extended program below in a case where $x$ is initially $0$. 
\[
  (y \samp \bern{\tfrac12} \fatsemi C) \parallel (x \coloneqq 1)
\]
In the left thread, after the $y \samp \bern{\tfrac12}$, if we apply \hyperlink{rule:frame1}{\textsc{Frame}} with $\vartheta = y \sim \bern{\frac12}$ and use the triple for $C$ in (\ref{eq:weak-triple}), we would get the postcondition $(\exists X\ge \tfrac12.\ z\sim \bern{X}) \sep (y\sim\bern{\frac12})$. But that postcondition is not valid because the scheduler can force $y$ and $z$ to be correlated by using the value of $y$ to influence the value of $x$.\footnote{For the interested reader:
if the scheduler makes $x$ equal to $y$, then $x$ must also be uniformly distributed, and so $z \sim \bern{\frac12}$ with probability $\frac12$ and $z\mapsto 1$ with probability $\frac12$, which implies that $z \sim \bern{\frac34}$. But clearly, $(z \sim\bern{\frac34})\sep(y\sim\bern{\frac12})$ does not hold, since whenever $y = 0$, then $z=1$, and so the probability of $\sure{y\mapsto 0\sep z\mapsto 1}$ is $\frac12$, which is not equal to $\frac12\cdot\frac34$.}
As we show in \Cref{sec:structural}, case analysis on nondeterministic state causes triples to only be \emph{weakly} frame preserving, meaning that only a variant of the \ruleref{Frame} rule using weak separation $\varphi \osep \psi$ can be used.
Under the right conditions, strong frame preservation can be restored, as shown in \Cref{sec:ast-rules,sec:entropy-mixer,sec:von-neumann}.


\medskip 

\noindent Having now given an overview of \pcol's features, we begin the technical development in \Cref{sec:semantics} by outlining a denotational model for probabilistic concurrent programs. We then give a measure theoretic model of assertions---including the separating and outcome conjunctions--- in \Cref{sec:model}. In \Cref{sec:logic}, we define \pcol triples and provide a proof system. Four examples are shown in \Cref{sec:examples} before we conclude by discussing related work in \Cref{sec:related} and future directions in \Cref{sec:discussion}.

\section{A Probabilistic and Concurrent Programming Language}
\label{sec:semantics}

\begin{figure}
\begin{align*}
\mathsf{Cmd} \ni C &\Coloneqq \skp 
\mid C_1 \fatsemi C_2
\mid C_1 \parallel C_2 
\mid \iftf b{C_1}{C_2} 
\mid \whl{b}{C}
\mid a
\\
\act \ni a &\Coloneqq
  x \coloneqq e 
  \mid x \samp d(e)
\\
\mathsf{Dist} \ni d &\Coloneqq \bern{-} \mid \geo{-} \mid \unif{-}
\\
\test \ni b & \Coloneqq
  \tru \mid \fls \mid b_1 \land b_2 \mid b_1 \lor b_2 \mid \lnot b \mid e_1 \asymp e_2
\\
\mathsf{Exp} \ni e &\Coloneqq x \mid v \mid b \mid e[e'] \mid [e_1, \ldots, e_n] \mid e_1 + e_2 \mid e_1 \cdot e_2 \mid \cdots
\end{align*}
\caption{Syntax of a probabilistic concurrent language, where $x\in\mathsf{Var}$, $v \in \mathsf{Val}$, and $\mathord{\asymp} \in \{ =, \le, <, \ldots \}$.}
\label{fig:syntax}
\end{figure}

We begin by describing the syntax and semantics of a probabilistic concurrent programming language, shown in \Cref{fig:syntax}. Program commands $C\in\mathsf{Cmd}$ consist of no-ops ($\skp$), sequential composition ($C_1\fatsemi C_2$), parallel composition ($C_1\parallel C_2$), if statements, while loops, and actions $a\in\act$. 

Actions can perform probabilistic sampling operations $x \samp d(e)$, where $d\in\mathsf{Dist}$ is a discrete probability distribution with the expression $e$ as a parameter. We include three types of distributions---Bernoulli distributions $\bern p$, assigning probability $p$ to $1$ and probability $1-p$ to $0$; geometric distributions $\geo p$, assigning probability $(1-p)^n p$ to each $n \in\mathbb N$; and uniform distributions $\unif e$ where $e$ evaluates to a finite set or list of values $[v_1, \ldots, v_n]$, each having probability $1/n$.

Our restriction to discrete distributions was motivated by the applications that we are targeting, including synchronization protocols, which only require fair coin flips \cite{lehmann1981advantages,ben-or1983another}; cryptography, where keys are uniformly sampled fixed length bit-strings; and randomized sketching data structures, where hashes are modeled as uniform random samples over a finite set \cite{flajolet1985approximate}. Semantic domains for combining nondeterminism with continuous probability have been explored \cite{tix2009semantic,keimel2017mixed}, but our approach would need modifications to exploit that. See \Cref{sec:related} for a discussion.

The language also has deterministic assignments $x\coloneqq e$, where $e$ is an expression. Expressions consist of variables $x$, values $v$, tests $b$, list literals $[e_1, \ldots, e_n]$, list accesses $e[e']$ (where $e$ is a list and $e'$ is an index), and standard arithmetic operations.
Many more actions could be added to this semantics, including nondeterministic assignment and atomic concurrency primitives such as compare-and-swap, but we do not explore them in this paper.

We use the recently introduced \emph{Pomsets with Formulae} model \cite{zilberstein2025denotational}, which augments typical  denotational techniques for concurrency semantics \cite{gischer1988equational,pratt1986modeling} in order to properly capture probabilistic behavior. We use a denotational model because \pcol describes \emph{distributions} over program states, which differs from the usual \csl setting where assertions are predicates on a single state. Thus, we cannot adapt the \citet{vafeiadis2011concurrent} style operational soundness argument of quantifying over all finite executions. Instead, we must use domain-theoretic techniques to construct the full set of distributions that can occur after an infinite amount of time. While there are several Iris-based separation logics for randomized programs that use an operational proof, all of those logics have predicates over single program states, which they then lift to a probabilistic interpretation using either couplings (for relational proofs)~\citep{polaris, gregersen2024asynchronous} or using resources to describe \emph{one} property of the randomized program (\eg error bounds or expected costs)~\citep{haselwarter2024tachishigherorderseparationlogic, aguirre2024error}. In contrast, \pcol assertions allow for rich specifications over distributions of states.
%

\subsection{Preliminaries: Memories and the Convex Powerset}
\label{sec:convex}

Programs must be interpreted in a domain that supports both probabilistic and nondeterministic computation. Although none of our actions are explicitly nondeterministic, nondeterminism arises due to the interleaving of concurrent threads. The difficulty is that typical representations of probabilistic computation (distributions) do not compose well with nondeterminism (powersets) \cite{varacca_winskel_2006,zwart2019no}. We instead use the \emph{convex powerset} $\C$ in our denotational semantics, which we describe in this section.

\subsubsection*{Memories, Expressions, and Tests}
A memory $\sigma \in \mem S \triangleq S\to \mathsf{Val}$ is a mapping from a finite set of variables $S \subseteq_\fin \mathsf{Var}$ to values, where $\mathsf{Val}$ consist of integers, rationals, and lists. The disjoint union $\mathord{\uplus}\colon \mem S \to \mem T \to \mem{S\cup T}$ combines two memories as long as $S\cap T=\emptyset$, and a similar operation $A\sep B$ is defined on sets of memories $A\subseteq \mem S$ and $B\subseteq \mem T$.
\begin{mathpar}
  (\sigma\uplus\tau)(x) \triangleq \left\{
    \begin{array}{ll}
      \sigma(x) & \text{if}~ x\in S
      \\
      \tau(x) & \text{if}~ x\in T
    \end{array}
  \right.

  A\sep B \triangleq \{ \sigma \uplus \tau \mid \sigma\in A, \tau \in B \}
\end{mathpar}
The notation $A\sep B$ is reminiscent of the \emph{separating conjunction} \cite{bi}, and indeed we will use it in \Cref{sec:logic} to define the separating conjunction.
We define projections $\pi_S \colon \mem T \to \mem{S\cap T}$ as $\pi_S(\sigma)(x) \triangleq \sigma(x)$ if $x\in S$.
Expressions are interpreted in the usual way with $\de{e}_\expr \colon \mem S \to \mathsf{Val}$ as long as $\free(e) \subseteq S$, if not then $\de{e}_\expr(\sigma)$ is undefined. The same is true for tests and $\de{b}_\test\colon \mem S\to \mathbb B$ (where $\mathbb{B} = \{0,1\}$) is defined if $\free(b) \subseteq S$.

\subsubsection*{Discrete Probability Distributions}
A discrete probability distribution $\mu \in \D(X)$ over a countable set $X$ is a mapping from elements of $X$ to $[0,1]$ such that $\sum_{x\in X}\mu(x) = 1$. The support of a distribution is the set of elements to which it assigns nonzero probability $\supp(\mu) \triangleq \{ x \in X \mid \mu(x) \neq 0 \}$. The Dirac, or point-mass, distribution $\delta_x$ assigns probability 1 to $x$ and $0$ to everything else. The previously defined projections extend to distributions $\pi_S \colon \D(\mem T) \to \D(\mem{S\cap T})$ by marginalizing as follows $\pi_S(\mu)(\sigma) \triangleq {\sum_{\tau \in \mem{T \setminus S}}} \mu(\sigma\uplus\tau)$. Distributions $\mu,\nu \in \D(X \cup \{\bot\})$ are ordered as follows: $\mu \sqsubseteq_\D \nu$ iff $\mu(x) \le \nu(x)$ for all $x\in X$ and $\mu(\bot) \ge \nu(\bot)$. This makes $\tuple{\D(X \cup \{\bot\}), \sqsubseteq_\D}$ a pointed poset with bottom $\bot_\D = \delta_\bot$.

\subsubsection*{The Convex Powerset}
A convex powerset is a set of all the possible distributions of outcomes that could result from (nondeterministic) scheduling.
For a more complete explanation, refer to \citet{jifeng1997probabilistic} and \citet{zilberstein2025demonic}.
Distributions can be added and scaled pointwise: $(\mu + \nu)(x) = \mu(x) + \nu(x)$ and $(p \cdot \mu)(x) = p\cdot\mu(x)$. The convex combination of two distributions is defined as $\mu \oplus_p\nu = p\cdot\mu + (1-p)\cdot \nu$.
A set of distributions $S\subseteq \D(X \cup \{\bot\})$ is \emph{convex} if it is closed under convex combinations: $(\mu \oplus_p \nu) \in S$ for every $\mu,\nu\in S$ and $p\in[0,1]$.

Additional requirements ensure that the domain is a DCPO, and therefore suitable for representing iterated computations and fixed points.
A set $S$ is \emph{up-closed} if for all $\mu, \nu \in \D(X \cup \{\bot\})$, if $\mu \in S$ and $\mu \sqsubseteq_\D \nu$ then $\nu \in S$.
Finally, $S$ is \emph{Cauchy closed} if it is closed in the product of Euclidean topologies \cite{mciver2005abstraction}, \ie it is a finite union of closed regions of $X\cup\{\bot\}$-dimensional Euclidean space. The convex powerset is now defined as follows:
\[
  \C(X) \triangleq
  \{ S \subseteq \D(X \cup \{\bot\})
    \mid
    S ~\text{is nonempty, up-closed, convex, and Cauchy closed}
  \}
\]
We include $\bot$ to represent nontermination and undefined behavior such as accessing a variable that is not scope. Nonemptiness ensures that the semantics is not vacuous, since undefined behavior is represented by $\{ \delta_\bot\}$ rather than $\emptyset$.
Up-closure ensures that $\C$ is a partial order in the \citet{smyth1978power} powerdomain, \ie $S \sqsubseteq_\C T$ iff $\forall \nu \in T.\ \exists \mu\in S.\ \mu \led \nu$. In fact, due to up-closure, $S \lec T$ iff $S \supseteq T$, so suprema are given by set intersections. Cauchy closure ensures that the intersections of directed sets are nonempty, and therefore $\tuple{\C(X), \lec}$ is a pointed DCPO with bottom $\bot_\C = \D(X \cup \{\bot\})$.

Convexity ensures that $\C$ carries a \emph{monad} structure \cite{jacobs2008coalgebraic}, making the sequencing of actions compositional. More precisely, there is a unit $\eta \colon X \to \C(X)$ and Kleisli extension $(-)^\dagger \colon (X \to \C(Y)) \to \C(X) \to \C(Y)$, which obey the monad laws: $\eta^\dagger = \mathsf{id}$, $f^\dagger \circ \eta = f$, and $f^\dagger \circ g^\dagger = (f^\dagger \circ g)^\dagger$. These operations are defined as follows:
\begin{mathpar}
  \eta(x) \triangleq \{ \delta_x \}

  f^\dagger(S) \triangleq \Big\{
    \smashoperator[r]{\sum_{x\in\supp(\mu)}} \mu(x)\cdot \nu_x
    ~ \Big| ~
    \mu\in S, 
    \forall x \in \supp(\mu).\
    \nu_x \in f_\bot(x)
  \Big\}
%
\end{mathpar}
where $f_\bot(x) = f(x)$ for $x\in X$ and $f_\bot(\bot) = \bot_\C$.
As an overloading of notation, we will occasionally write $f^\dagger(\mu)$ to mean $f^\dagger(\{\mu\})$ for any $\mu \in \D(X)$.
Convex combinations in $\C$ are defined as $S \oplus_p T \triangleq \{ \mu \oplus_p \nu \mid \mu \in S, \nu\in T \}$, and  convex union is $S \nd T \triangleq \bigcup_{p \in [0,1]} S\oplus_p T$. For some finite index set $I = \{ i_1, \ldots, i_n \}$, we let $\bignd_{i\in I} S_i \triangleq S_{i_1} \nd \cdots \nd S_{i_n}$. Finally, we extend projections to convex sets as $\pi_S(T) \triangleq \{ \pi_S(\mu) \mid \mu \in T \}$ where we let $\sigma\uplus\bot = \bot$.

In the next section, we will see how `$\nd$' will be used to represent the choices of the \emph{scheduler} when interleaving concurrent threads. The fact that $S\nd T$ is represented as a set of convex combinations operationally corresponds to the idea that the scheduler can use randomness to choose between $S$ and $T$, rather than making the choice deterministically \cite[\S 6.5]{varacca2002powerdomain}.

\subsection{Actions and Invariants}
\label{sec:actions}

We can now use the convex powerset to give semantics to actions. The basic action evaluation is defined below $\de{-}_\act \colon \act\to \mem S \to \C(\mem S)$.
\begin{align*}
  \de{x \coloneqq e}_\act(\sigma) &\triangleq \left\{
    \begin{array}{ll}
      \eta(\sigma[x \coloneqq \de{e}_\expr(\sigma)]) & \text{if} ~ \free(e) \cup\{x\} \subseteq S
      \\
      \bot_\C & \text{otherwise}
    \end{array}
  \right.
\\
  \de{x \samp d(e)}_\act(\sigma) &\triangleq \left\{
    \begin{array}{ll}
       \left\{ \sum_{v\in \supp(\mu)} \mu(v) \cdot \delta_{\sigma[x \coloneqq v]} \right\} & \text{if} ~ \free(e) \cup\{x\} \subseteq S, \mu = d(\de{e}_\expr(\sigma))
      \\
      \bot_\C & \text{otherwise}
    \end{array}
  \right.
\end{align*}
As we alluded to in \Cref{sec:overview}, we will reason about shared state via \emph{invariants}---assertions about shared state that must be preserved by every atomic action. To model the ways in which shared state may be modified by other threads, we define an \emph{invariant sensitive semantics}, in which the scheduler may alter shared state before executing each atomic action.
This is based on \emph{semantic invariants}, finite sets of memories $\I \subseteq_\mathsf{fin} \mem T$ that represent the legal values of shared state.

We limit invariants to be finite sets in order to avoid issues arising from unbounded nondeterminism.
Stemming from the impossibility result of \citet{apt1986countable},
the semantics of loops cannot be constructed in standard ways using least fixed points in the presence of unbounded nondeterminism. In $\C$ specifically, unbounded nondeterminism breaks Cauchy closure \cite[Appendix B.4.2]{mciver2005abstraction}. Finite invariants are sufficient for a wide variety of verification tasks, such as the examples in \Cref{sec:examples}. Looking forward, there are additional algorithms where shared state only takes on finitely many values such as the randomized Dining Philosophers \cite{lehmann1981advantages} and other synchronization protocols \cite{rabin1980n-process,rabin1982choice}.

The invariant-sensitive model is a family of semantic functions for actions, indexed by a semantic invariant: $\de{a}_\act^\I \colon \mem S \to \C(\mem S)$, where $\I \subseteq_\fin \mem T$ and $T\subseteq S$.
\[
  \de{a}^\I_\act \triangleq (\mathsf{check}^\I)^\dagger \circ \de{a}_\act^\dagger \circ (\mathsf{replace}^\I)^\dagger \circ \mathsf{check}^\I
\]
\[
  \mathsf{check}^\I(\sigma) \triangleq \left\{
    \begin{array}{ll}
      \eta(\sigma) & \text{if}~ \pi_T(\sigma) \in \I
      \\
      \bot_\C & \text{if}~ \pi_T(\sigma) \notin \I
    \end{array}
  \right.
  \qquad
  \mathsf{replace}^\I(\sigma) = \bignd_{\tau \in \I} \eta(\pi_{S\setminus T}(\sigma) \uplus \tau)
\]
In the invariant sensitive semantics, $\I$ is first checked to ensure that the current state satisfies the invariant. Next, a new valid state (or distribution thereof) is chosen to replace the current one, simulating a parallel thread which may alter the shared state at any step. The standard action semantics is then executed, followed by another check to ensure that the invariant still holds.
If the invariant is ever violated, then $\bot_\C$ is returned to indicate that the execution is faulty. Letting $\emp \in \mem\emptyset$ be the empty memory, we remark that $\de{a}_\act = \de{a}_\act^{\{\emp\}}$, meaning that invariant sensitive execution using the empty invariant is equal to normal execution. In \Aref{lem:act-mono}, we prove that the invariant sensitive semantics is monotonic---\ie $\de{a}_\act(\sigma) \subseteq \de{a}^\I_\act(\sigma)$---adding an invariant can only add behaviors, making it an over-approximation of the program's behavior.
So, safety properties about $\de{a}^\I_\act(\sigma)$ immediately apply to $\de{a}_\act(\sigma)$ too.

\subsection{Semantics of Randomized Concurrent Programs}
\label{sec:pomset}

To give semantics to commands, we use \emph{Partially Ordered Multisets (Pomsets) with Formulae} \cite{zilberstein2025denotational}, where a partial order represents the causality between actions in the program. We write $a_1 \to a_2$ to mean that the action $a_1$ must be scheduled before $a_2$. Pomsets with formulae are constructed using three composition operators, shown below, which mirror the program syntax.
\newcommand{\guardarrows}{
  \ar@<.05mm>@[NavyBlue]^{\Tarr}[ul]
  \ar@<-.05mm>@[NavyBlue][ul]
  \ar@<-.05mm>@[Maroon]_{\Farr}[ur]  
  \ar@<.05mm>@[Maroon][ur]
}
\begin{mathpar}
\de{a_1 \fatsemi a_2} = 
\vcenter{\vbox{\xymatrix@R=6pt@C=3pt{
  a_2 \\ a_1 \ar[u]
}}}

\de{a_1 \parallel a_2} = 
\vcenter{\vbox{\xymatrix@R=6pt@C=3pt{
  a_1 && a_2 \\& \fork \ar[ul]\ar[ur]
}}}

\de{\iftf b{a_1}{a_2}} =
\vcenter{\vbox{\xymatrix@R=6pt@C=3pt{
  a_1 && a_2
  \\
  & b \guardarrows
}}}
\end{mathpar}
From left to right, the sequential composition $a_1\fatsemi a_2$ results in a totally ordered structure, where $a_1$ must occur before $a_2$. In a parallel composition $a_1 \parallel a_2$, the actions $a_1$ and $a_2$ are not ordered with respect to each other, so they can be scheduled in any order. Finally, if statements result in a guarded branch, where the two successors of the test $b$ both must be scheduled after $b$, but also will only be scheduled if the outcome of the test matches the label on the arrow.

For the purposes of our program logic, we are interested in the \emph{linearized} version of the model, $\lin(\de-) \colon \mathsf{Cmd} \to \mem S \to \C(\mem S)$, which maps input memories to convex sets of output memories, representing the set of possible distributions that can arise due to different interleavings of the parallel threads chosen by the scheduler.
Linearization is defined in terms of the semantics for actions and tests from \Cref{sec:actions}. The semantics of actions $\de{-}^\I_\act$ is indexed by an invariant $\I$, and so linearization $\lin^\I$ is also indexed by an invariant, indicating which semantic function to use for actions.
We provide the definition of $\lin$---due to \citet{zilberstein2025denotational}---in \Aref{app:pomset}.
We omit the indexing invariant $\I$ and just write $\lin(\de-)$ when $\I = \{\emp\}$.

When the invariant $\I$ is $\{\emp\}$ in $\lin$, the scheduler does not simulate any mutations performed by other threads, since all actions are evaluated according to $\de{a}_\act^{\{\emp\}} = \de{a}_\act$, so $\lin(\de{C})(\sigma)$ can be viewed as the \emph{true} semantics of the program.
Similar to action evaluation, linearization is also monotonic with respect to invariants (\Cref{lem:inv-mono}),
meaning that
$
  \lin(\de{C})(\sigma)
  \subseteq
  \lin^\I(\de{C})(\sigma)
$.
This guarantees that adding an invariant will only add new behaviors, so safety properties about $\lin^\I(\de{C})(\sigma)$ will automatically carry over to $\lin(\de{C})(\sigma)$. 

The linearized \emph{state transformer} is ideal for modeling a program logic, where programs are specified in terms of preconditions and postconditions. As shown by \citet[Lemma 5.2]{zilberstein2025denotational}, linearization of non-parallel programming constructs is well-behaved; it is compositional with respect to sequencing, if statements, and while loops, as shown by the following equational rules:
\begin{mathpar}
  \lin^\I(\de{\skp}) = \eta
  
  \lin^\I(\de{a}) = \de{a}_\act^\I

  \lin^\I(\de{C_1 \fatsemi C_2}) = \lin^\I(\de{C_2})^\dagger \circ \lin^\I(\de{C_1})

  \lin^\I(\de{\iftf b{C_1}{C_2}})(\sigma) = \left\{
    \begin{array}{ll}
      \lin^\I(\de{C_1})(\sigma) & \text{if}~ \de{b}_\test(\sigma) = \tru
      \\
      \lin^\I(\de{C_2})(\sigma) & \text{if}~ \de{b}_\test(\sigma) = \fls
    \end{array}
  \right.

  \lin^\I(\de{\whl bC}) = \mathsf{lfp}\left(\Psi_{\tuple{b, C, \I}}\right)
  \quad
  \Psi_{\tuple{b, C, \I}}(f)(\tau) = \left\{
  \arraycolsep=2pt
    \begin{array}{ll}
      f^\dagger(\lin^\I(\de{C})(\tau)) & \text{if}~ \de{b}_\test(\tau) = \tru
      \\
      \eta(\tau) & \text{if}~ \de{b}_\test(\tau) = \fls
    \end{array}
  \right.
\end{mathpar}
In fact, the parallel-free fragment of the linearized model is equivalent to the model of Demonic Outcome Logic \cite[Theorem 5.3]{zilberstein2025denotational}, and so some of the metatheory for standard commands carries over from \dol to \pcol.
However, there is no straightforward compositional property for parallel programs, since the input-output behavior of two threads can completely change if they are run in parallel. 
Building on the insights of Concurrent Separation Logic \cite{csl} and Probabilistic Separation Logics \cite{psl,bao2025bluebell,li2023lilac},  we develop compositional reasoning techniques for parallel programs in \Cref{sec:model,sec:logic}.

\section{The Model of Probabilistic Assertions}
\label{sec:model}

Preconditions and postconditions in Probabilistic Concurrent Outcome Logic (\pcol) are inspired by both Demonic Outcome Logic \cite{zilberstein2025demonic} and also probabilistic separation logics \cite{psl,li2023lilac,bao2025bluebell,bao2021bunched,bao2022separation}. We begin by giving the syntax and semantics for basic assertions about memories in \Cref{sec:basic-assertions}.
Next, we discuss background on measure theory and probability spaces in \Cref{sec:measure}.
Like Lilac \cite{li2023lilac} and Bluebell \cite{bao2025bluebell}, the model of resources uses probability spaces that only assign probabilities to certain \emph{measurable} sets of memories. 
Based on this, we define probabilistic assertions in \Cref{sec:prob-assert}.

\subsection{Pure Assertions}
\label{sec:basic-assertions}

We begin by describing pure (non-probabilistic) assertions, which are inspired by standard separation logic \cite{localreasoning,sl}, but where memories range over variables rather than heap cells, as we discussed in \Cref{sec:actions}. The syntax for these assertions are shown below.
\begin{align*}
  P &\Coloneqq
    \tru \mid \fls\mid P\land Q \mid P\lor Q\mid P \sep Q \mid \exists X.\ P \mid e \mapsto E \mid E_1 \asymp E_2
    \qquad (\mathord{\asymp} \in \{ =, \le, <, \in, \ldots \})
  \\
  E &\Coloneqq X \mid v \mid E_1 + E_2 \mid E_1 \cdot E_2 \mid \cdots
\end{align*}
In addition to expressions and variables from \Cref{sec:semantics}, assertions also depend on logical variables $X,Y,Z\in\mathsf{LVar}$, which cannot be modified by programs. Logical expressions $E \in \mathsf{LExp}$ mirror standard ones, but operate over logical variables $X\in\mathsf{LVar}$ rather than $x\in\mathsf{Var}$. Logical expression evaluation under a context $\Gamma \colon \mathsf{LVar} \to \mathsf{Val}$ is written $\de{E}_{\mathsf{LExp}}(\Gamma)$ and is defined in a standard way.

Pure assertions are modelled by both a context $\Gamma$, and a memory $\sigma\in\mem S$, the satisfaction relation is shown in \Cref{fig:pure-assertions}. The meaning of $\tru$, $\fls$, conjunction, and disjunction are standard. The separating conjunction $P\sep Q$ means that the memory $\sigma \in \mem S$ can be divided into two smaller memories $\sigma_1 \in \mem{S_1}$ and $\sigma_2\in\mem{S_2}$ to satisfy $P$ and $Q$ individually. By the definition of $\uplus$, $S_1$ and $S_2$ must be disjoint. 
Our logic is an \emph{intuitionistic} \cite{simon} or \emph{affine} interpretation of separation logic, meaning that information about variables can be discarded; if $\Gamma,\sigma\vDash P$, then $P$ need not describe the entire memory $\sigma$. As such, we only require that $\sigma_1 \uplus \sigma_2 \sqsubseteq \sigma$, which we define as $\sigma \sqsubseteq \tau$ iff $\sigma \uplus \sigma' = \tau$ for some $\sigma'$, so $\tau$ could contain more variables than $\sigma$.

We also include a points-to predicate $e\mapsto E$, although it has a slightly different meaning than points-to predicates in the heap model. Here, $e$ does not describe a pointer, but can rather be any concrete expression, and $e\mapsto E$ simply means that the $e$ evaluates to the same value under $\sigma$ as $E$ does under $\Gamma$, allowing us to connect the concrete and logical state. Finally, $E_1\asymp E_2$ allows us to make assertions about logical state, where $\mathord{\asymp} \in \{ =, \le, \in, \cdots \}$ ranges over similar comparators to the ones in \Cref{sec:semantics}.
We define the following notation to obtain the set of all memories $\sigma\in \mem S$ that satisfy an assertion $P$, we omit the superscript when we wish to minimize $S$, so that the memories contain only the free variables of $P$:
\[
  \sem{P}_\Gamma^S \triangleq \{
    \sigma \in \mem S \mid \Gamma, \sigma\vDash P
  \}
  \qquad\qquad
  \sem{P}_\Gamma \triangleq \sem{P}_\Gamma^{\free(P)}
\]
Finally, we provide syntactic sugar for restricting the domain of existential quantifiers, asserting membership in a set, and asserting that the resources of $e$ are owned by the current thread.
\begin{mathpar}
    \exists X\in E.\ P \triangleq \exists X.\ P \sep X \in E

    {e\in E} \triangleq \exists X\in E.\ e \mapsto X

    \own(e_1, \ldots, e_n) \triangleq \bigsep_{i=1}^n \exists X_i.\ e_i \mapsto X_i
\end{mathpar}

\begin{figure}
\[
\begin{array}{lll}
  \Gamma, \sigma \vDash \tru & \multicolumn{2}{l}{\text{always}} \\
  \Gamma, \sigma \vDash \fls & \multicolumn{2}{l}{\text{never}} \\
  \Gamma,\sigma\vDash P\land Q & \text{iff} & \Gamma,\sigma\vDash P\quad\text{and}\quad \Gamma,\sigma\vDash Q \\
  \Gamma,\sigma\vDash P\lor Q & \text{iff} & \Gamma,\sigma\vDash P\quad\text{or}\quad \Gamma,\sigma\vDash Q \\
  \Gamma, \sigma \vDash P\sep Q & \text{iff} &
      \exists \sigma_1,\sigma_2.\quad
      \sigma_1 \uplus \sigma_2 \sqsubseteq \sigma
      \quad\text{and}\quad \Gamma,\sigma_1\vDash P
      \quad\text{and}\quad \Gamma,\sigma_2\vDash Q
    \\
  \Gamma, \sigma \vDash \exists X.\ P & \text{iff} & \Gamma[X \coloneqq v], \sigma \vDash P \ \text{for some}\ v\in \mathsf{Val} \\
  \Gamma, \sigma \vDash e \mapsto E & \text{iff} &
      \de{e}_\expr(\sigma) = \de{E}_{\mathsf{LExp}}(\Gamma)
    \\
  \Gamma, \sigma \vDash E_1 \asymp E_2 & \text{iff} &
      \de{E_1}(\Gamma) \asymp \de{E_2}(\Gamma)
\end{array}
\]
\caption{Satisfaction relation for pure assertions.}
\label{fig:pure-assertions}
\end{figure}

\subsection{Measure Theory and Probability Spaces}
\label{sec:measure}

We now introduce basic definitions from measure theory needed to define probabilistic separation. For a more thorough background, refer to \citet{royden1968real} or \citet{fremlin2001measure}.
A probability space $\P=\langle\Omega, \mathcal F, \mu\rangle$ consists of a sample space $\Omega$, an event space $\F$, and a probability measure $\mu$. For our purposes, the sample space $\Omega\subseteq \mem S$ will consist of memories over a particular set of variables $S$. The event space $\F \subseteq\bb{2}^\Omega$ gives the events---\ie sets of memories---which are measurable. It must be a $\sigma$-algebra, meaning that it contains $\emptyset$ and $\Omega$, and it is closed under complementation and countable unions and intersections. The probability measure $\mu \colon \F\to [0,1]$ assigns probabilities to the events in $\F$, and must obey $\mu(\emptyset) = 0$, $\mu(\Omega) = 1$, and countable additivity: $\mu( \biguplus_{i\in I}A_i) = \sum_{i\in I} \mu(A_i)$ where $I$ is a countable index set and all the $A_i$ sets are pairwise disjoint. For a probability space $\P$, we use $\Omega_\P$, $\F_\P$, and $\mu_\P$ to refer to its respective parts.

We require probability spaces to be \emph{complete}, meaning that they contain all events of measure zero. More formally, $\P$ is complete if for any $A \in \F_\P$ such that $\mu_\P(A) = 0$, then $B\in\F_\P$ for all $B\subseteq A$ \cite{royden1968real}. We will often also require the sample space to be the full set of memories $\mem S$ for some $S$. A probability space $\P$ with $\Omega_\P \subseteq \mem S$ can be \emph{extended} as follows: $\comp(\P) \triangleq \tuple{\mem S, \F, \mu}$ where $\F \triangleq \{ A \subseteq \mem{S} \mid A \cap \Omega_\P \in \F_\P \}$ and $\mu(A) \triangleq \mu_\P(A \cap \Omega_\P)$.

 
We now define a preorder on probability spaces. As is typical in intuitionistic logic, this preorder $\P \preceq \Q$ will indicate when $\Q$ contains more information than $\P$. The information can be gained across two dimensions: by expanding the memory footprint, or by making the event space more granular. Formally, for $\P$ and $\Q$ such that $\Omega_\P \subseteq \mem S$, we define $\P \preceq \Q$ as follows:
\begin{align*}
  \P \preceq \Q \qquad\text{iff}&\quad
  \Omega_\P \subseteq \pi_S(\Omega_\Q)
  \\
  \text{and}&\quad
  \F_\P \subseteq \{ \pi_S(A) \mid A \in \F_\Q \}
  \\\text{and}&\quad
  \forall A \in \F_\P.~ \mu_\P(A) = \mu_\Q\left({\bigcup \{B\in \F_\Q\mid \pi_S(B) = A \}}\right)
\end{align*}
So, $\P\preceq\Q$ iff $\P$ contains smaller sample and event spaces, but $\P$ and $\Q$ agree on the probability of events whose projections are measurable in $\P$. Any proper distribution $\mu \in \D(\mem U)$, can be used as a probability space where $\Omega_\mu = \mem U$, $\F_\mu = \bb{2}^{\mem U}$ is the greatest $\sigma$-algebra on $\mem U$, and $\mu_\mu(A) = \sum_{\sigma\in A} \mu(\sigma)$.
Projections of probability spaces are defined as $\pi_U(\P) = \tuple{\Omega,\F,\mu}$, where $\Omega = \pi_U(\Omega_\P)$, $\F = \{ \pi_U(A) \mid A \in \F_\P \}$, and $\mu(A) = \mu_\P(A \sep \pi_{\mathsf{Var}\setminus U}(\Omega_\P))$.

We define two more operations on probability spaces, which will help us to give semantics to the separating conjunction and outcome conjunction in \Cref{sec:prob-assert}. The first operation is the product space $\P\otimes \Q$, which is defined when $\Omega_\P \subseteq \mem S$, $\Omega_\Q\subseteq\mem T$, and $S\cap T=\emptyset$. The sample space $\Omega_{\P\otimes\Q} = \Omega_\P\sep\Omega_\Q$ is the set of all joined memories in the two spaces, the event space $\F_{\P\otimes\Q}$ is the smallest $\sigma$-algebra containing $\{ A \sep B \mid A \in \F_\P, B\in \F_\Q \}$, and the measure has the property that $\mu_{\P\otimes\Q}(A \sep B) = \mu_\P(A)\cdot \mu_\Q(B)$ for any $A\in \F_\P$ and $B\in\F_\Q$. The full construction uses Carath\'{e}odory's method, and is given in Chapter 25 of \citet{fremlin2001measure}.

Note that this definition is more similar to the initial formulation of \psl \cite{psl} (albeit, in a probability space), rather than Lilac and Bluebell, which rely on a theorem stating that independent products are unique \cite[Lemma 2.3]{li2023lilac}.
We use the explicit product construction in order to guarantee that each variable can only occur on one side of the $\sep$, making mutation rules simpler. Lilac does not allow mutable state, and so mutation is not a factor. On the other hand, Bluebell handles mutation by explicitly tracking permissions, but we found the product construction to be simpler to use than the permission approach.

The next operation is a \emph{direct sum} for combining disjoint probability spaces \cite[214L]{fremlin2001measure}. More precisely, for some countable index set $I$, discrete distribution $\nu \in \D(I)$, and probability spaces $\P_i = \tuple{\Omega_i, \F_i, \mu_i}$ such that the $\Omega_i$ are pairwise disjoint, we define the direct sum as:
\begin{mathpar}
  \bigoplus_{i \sim \nu} \P_i \triangleq \tuple{\textstyle\biguplus\limits_{i\in I} \Omega_i, \F,\mu}

  \F \triangleq \{ A \subseteq \Omega \mid \forall i\in I.\ A\cap \Omega_i \in \F_i \}

  \mu(A) \triangleq \smashoperator{\sum_{i\in I}} \nu(i)\cdotp \mu_i(A \cap \Omega_i)
\end{mathpar}
The sample space is the union of all the individual sample spaces, the measurable events are those events whose projections into each $\Omega_i$ are measurable according to $\F_i$, and the probability measure is given by a convex sum. The direct sum will be used to give semantics to our outcome conjunction.  Finally, we remark that independent products distribute over direct sums (\Aref{lem:otimes-oplus-dist}):
\[\textstyle
  \left(\bigoplus_{i\sim\nu} \P_i \right) \otimes\Q
  =
  \bigoplus_{i\sim \nu} (\P_i \otimes \Q)
\]

\subsection{Probabilistic Assertions}
\label{sec:prob-assert}

We now define probabilistic assertions, which will serve as pre- and postconditions in \pcol triples. The syntax is shown below and the semantics is in \Cref{fig:prob-sem}.
\[
  \varphi \Coloneqq
    \top \mid \bot
    \mid \varphi\land\psi
    \mid \varphi\lor\psi
    \mid \exists X.\ \varphi
    \mid  \smashoperator{\bigoplus_{X\sim d( E)}} \varphi
    \mid \bignd_{X\in E} \varphi
    \mid \varphi\ast_m\psi
    \mid \sure P
\qquad (m\in\{\st,\wk\})
\]
The semantics for probabilistic assertions is based on a context $\Gamma\colon \mathsf{LVar}\to\mathsf{Val}$, and a complete probability space $\P = \tuple{\mem S, \F_\P, \mu_\P}$. The $\top$, $\bot$, conjunction, disjunction, and existential quantification assertions have the usual semantics.

\begin{figure}\small
\[\def\arraystretch{1.5}
\begin{array}{lll}
\Gamma,\P \vDash\top & \multicolumn{2}{l}{\text{always}}
\\
\Gamma,\P \vDash\bot & \multicolumn{2}{l}{\text{never}}
\\
\Gamma,\P \vDash \varphi\land\psi &\text{iff}&
    \Gamma,\P\vDash \varphi
    \quad\text{and}\quad
    \Gamma,\P\vDash \psi
\\
\Gamma,\P \vDash \varphi\lor\psi &\text{iff}&
    \Gamma,\P\vDash \varphi
    \quad\text{or}\quad
    \Gamma,\P\vDash \psi
\\
\Gamma,\P \vDash \exists X.\ \varphi &\text{iff}&
    \Gamma[X \coloneqq v],\P\vDash \varphi
    \quad\text{for some}\quad v\in\mathsf{Val}
\\
\Gamma, \P \vDash \bigoplus_{X\sim d( E)} \varphi
  &\text{iff}&  
    \forall v.~ \Gamma[X \coloneqq v],\comp(\P_v) \vDash \varphi
    \quad\text{and}\quad
    \bigoplus_{v \sim\mu} \P_v \preceq \P
    \quad\text{for some}\quad
    (\P_v)_{v\in\supp(\mu)}
    \\&& \quad\text{where}\quad
    \mu = d(\de{E}_\mathsf{LExp}(\Gamma))
\\
\Gamma, \P \vDash \bignd_{X\in E} \varphi
  &\text{iff}&  \Gamma,\P \vDash \bigoplus_{X\sim\mu} \varphi \quad\text{for some}\quad \mu \in \D(\de{E}_\mathsf{LExp}(\Gamma))
\\
\Gamma, \P \vDash \varphi \ast_m \psi
  &\text{iff}&
  \P' \preceq \P
  \;\;\text{and}\;\;
  \Gamma, \P_1 \vDash \varphi
   \;\;\text{and}\;\;
  \Gamma, \P_2 \vDash \psi
  \;\;\text{for some}\;\; \P_1,\P_2, \text{and}\ \P' \in \P_1 \diamond_m \P_2
\\
\Gamma, \P \vDash \sure P
  &\text{iff}& 
  \sem{P}_\Gamma^S \in \F_\P
  \quad\text{and}\quad
  \mu_\P\left(\sem{P}_\Gamma^S \right) = 1
\end{array}
\]
\caption{The satisfaction relation, where $\Gamma \colon \mathsf{LVar} \to\mathsf{Var}$ is a logical context and $\P = \tuple{\mem S, \F_\P, \mu_\P}$ is a complete probability space. All the existentially quantified probability spaces are also complete.}
\label{fig:prob-sem}
\end{figure}

Next, we have two kinds of \emph{outcome conjunctions}, adapted from Demonic OL (\dol) \cite{zilberstein2025demonic}, but with a new measure-theoretic semantics based on direct sums.
The standard outcome conjunction $\bigoplus_{X\sim d( E)} \varphi$ allocates a new logical variable $X$, which is distributed according to $\mu = d(\de{ E}_\mathsf{LExp}(\Gamma))$, and can be referenced in $\varphi$. The probability space $\P$ must be a refinement of the direct sum of $(\P_v)_{v\in\supp(\mu)}$. For every $v$, we then also require that $\Gamma[X\coloneqq v], \comp(\P_v) \vDash\varphi$, so $\varphi$ holds in the sub-probability space $\P_v$ with the value of $X$ in $\Gamma$ updated accordingly.
Essentially, the outcome conjunction splits the sample space $\mem S$ according to the support of $d(E)$.

The nondeterministic outcome conjunction $\bignd_{X\in E}$ is similar, but here only the support of the distribution is specified (as $E$). This connective is used when disjunctions or existential quantification would be used in a purely nondeterministic logic. For example, as a result of running the concurrent program $x \coloneqq 1 \parallel x \coloneqq 2$, it is not correct to say that $\sure{x\mapsto 1} \vee \sure{x \mapsto 2}$ since the probabilistic scheduler could choose to make $x$ equal to 1 with some probability $0<p<1$. On the other hand $\bignd_{X \in \{1,2\}} \sure{x\mapsto X}$ means that $x$ takes on value 1 with some (existentially quantified) probability, matching the convex powerset interpretation of nondeterminism.


Next, we have two variations of the separating conjunction, parameterized by a mode $m\in\{\st, \wk\}$. Strong separation (mode $\st$) is the interpretation of separation that requires \emph{probabilistic independence} and separation of variables, whereas weak separation (mode $\wk$) does not require independence, and only separates the variables.
As we explained in \Cref{sec:overview}, weak separation is needed when case analysis over nondeterministic shared state leaves us in a scenario where strong frame preservation does not hold.
Semantically, the difference is captured by the combinator operation $\diamond_m$, with strong separation using an independent product and weak separation simply requiring that the marginal probability spaces are correct.
\begin{align*}
  \P_1 \diamond_\st \P_2 &\triangleq \{ \P_1 \otimes \P_2 \}
  &
  \P_1 \diamond_\wk \P_2 &\triangleq \{ \P \mid \P_1 = \pi_U(\P), \P_2 = \pi_V(\P) \}  
\end{align*}
where $\Omega_{\P_1} = \mem{U}$ and $\Omega_{\P_2} = \mem V$ and $U\cap V = \emptyset$.
Clearly, $\P_1 \diamond_\st\P_2 \subseteq \P_1\diamond_\wk\P_2$, which immediately gives us that $\varphi \sep_\st \psi \Rightarrow \varphi \osep \psi$. Strong separation will be used more commonly, so we will drop the subscript there and write $\sep$ to mean $\sep_\st$.


Finally, the almost sure assertion $\sure P$ states that the pure assertion $P$, as described in \Cref{sec:basic-assertions}, occurs with probability 1. We also define syntactic sugar below for a binary outcome conjunction $\oplus_E$, a bounded binary outcome conjunction $\oplus_{\ge E}$, and expressions distributed according to some distribution $e \sim d(E)$.
\[\small
    \varphi \oplus_E \psi \triangleq {\bigoplus_{X \sim \bern E}}
    (\sure{X = 1} \sep \varphi) \vee (\sure{X = 0} \sep \psi)
    \qquad\quad
    e\sim d(E) \triangleq {\bigoplus_{X\sim d(E)}} \sure{e \mapsto X}
\]
\[
  \varphi \oplus_{\ge E} \psi \triangleq \exists X.\ \sure{X \ge E} \sep (\varphi \oplus_X \psi)
\]

\subsection{Convex and Precise Assertions and Entailment Laws}

As we mentioned in \Cref{sec:overview}, the parallel composition rule of \pcol requires the postcondition from each thread to be \emph{precise}, so that any correlations introduced via concurrent scheduling are not measurable. We now define precision formally in terms of probability spaces.

\begin{definition}[Precision]\label{def:precise}
An assertion $\varphi$ is \emph{precise} if for any $\Gamma$ under which $\varphi$ is satisfiable there is a unique smallest probability space $\P$ such that $\Gamma,\P \vDash \varphi$ and if $\Gamma,\P' \vDash \varphi$, then $\P \preceq \P'$.
We write $\precise{\varphi_1, \ldots, \varphi_n}$ to mean $\precise{\varphi_1} \land \cdots \land \precise{\varphi_n}$.
\end{definition}
Below, we give a few rules to determine that assertions are precise.
\begin{mathpar}
\inferrule{\;}{\precise{\sure P}}

\inferrule{
  \precise{\varphi,\psi}
}{
  \precise{\varphi\sep\psi}
}

\inferrule{
  \precise\varphi
  \\
  \varphi\Rightarrow\sure{e \mapsto X}
}{
  \precise{\textstyle\bigoplus_{X\sim d(E)}\varphi}
}
\end{mathpar}
Almost sure assertions are always precise, since the smallest model is the one where $\sem{P}_\Gamma$ occurs with probability 1, and is the smallest measurable set with nonzero probability. Separating conjunctions are precise if their subcomponents are, which follows from monotonicity of the independent product (\Aref{lem:otimes-mono}). Outcome conjunctions are precise if the inner assertion is precise, and implies that $\sure{e\mapsto X}$ for some program expression $e$, which witnesses how to partition the sample space for the direct sum.
Without this partitioning side condition, the result is not necessarily precise.
For example, $\sure{x \mapsto 1} \oplus_\frac12 \sure{x \in \{0,1\}}$ is not precise, since it is not possible to determine the probability of the event $x=1$, despite it being measurable in one of the sub-probability spaces.
However, $\sure{x \mapsto 1 \sep y\mapsto0} \oplus_\frac12 \sure{x \in \{0,1\} \sep y\mapsto 1}$ is precise, since $y$ witnesses the partition.

Weak separating conjunctions $\osep$ and nondeterministic outcome conjunctions $\bignd$ are \emph{not} precise, as there are generally many different minimal probability spaces that satisfy them, assigning different probabilities to each event. However, those assertions do obey \emph{convexity}, a weaker condition, which intuitively means that $\varphi \oplus_p\varphi \Rightarrow \varphi$. We give the formal definition below.

\begin{definition}[Convex Assertions]
$\convex\varphi$ iff for all $\Gamma$ under which $\varphi$ is satisfiable, there exist $\Omega$, $\F$, and a convex set $S$ of probability measures on $\F$ such that:
\[
  \forall \P. \quad
 \Gamma,\P\vDash\varphi
 \qquad\text{iff}\qquad
 \exists \mu\in S.\quad
 \tuple{\Omega, \F, \mu} \preceq \P
\]
\end{definition}

Clearly any precise assertion is convex, since the set of measures $S$ is just a singleton in that case. So, in addition to analogues of the precision rules above, we also have the following:
\begin{mathpar}
\inferrule{
  \convex{\varphi,\psi}
}{
  \convex{\varphi\osep\psi}
}

\inferrule{
    \convex{\varphi,\psi}
    \\
    \substack{\varphi\Rightarrow\sure{e\mapsto 1} \\ \psi\Rightarrow \sure{e\mapsto 0}}
  }{
    \convex{\varphi \oplus_{\ge p}\psi}
  }
  
\inferrule{
  \convex{\varphi}
  \\
  \varphi\Rightarrow\sure{e\mapsto X}
}{
  \convex{\textstyle\bignd_{X\in E}\varphi}
}
\end{mathpar}
Precision and convexity are useful for formulating entailment laws, which we provide in \Cref{fig:entailment}.
The first row consists of rules for the interaction between separating conjunctions and pure assertions. Weakening can be performed underneath both pure assertions and separating conjunctions. In addition, strong separation implies weak separation, and weak and strong separation are equivalent when one of the conjuncts is a pure assertion (since independence is trivial in that case).

\begin{figure}
\small
\begin{mathpar}
  \inferrule{P \vdash Q}{\sure P\vdash \sure{Q}}

  \inferrule{
    \varphi \vdash \varphi'
    \\
    \psi\vdash\psi'
  }{
    \varphi\sep_m\psi \vdash\varphi'\sep_m\psi'
  }

  \varphi\sep\psi \vdash\varphi\osep\psi

  \sure{P\sep Q} \dashv\vdash \sure P \sep_m \sure Q
  
  \varphi \sep \sure P \dashv\vdash \varphi\osep \sure{P}

\textstyle\bigoplus_{X\sim d(E)}\varphi\vdash\bignd_{X\in \supp(d(E))} \varphi
%

\varphi[E/X] \dashv\vdash\bignd_{X \in \{E\}}\varphi

\sure{E \subseteq E'} \sep \bignd_{X \in E}\varphi \vdash\bignd_{X\in E'}\varphi

  \inferrule{\varphi \vdash \psi}{\textstyle
    \bigoplus_{X\sim d( E)} \varphi \vdash \bigoplus_{X\sim d( E)}\psi
   }
   \qquad
\inferrule{
  Y \notin\mathsf{fv}(\varphi)
}{\textstyle
  \bigoplus_{X\sim d(E)} \varphi \vdash \bigoplus_{Y\sim d(E)} \varphi[Y/X]
}
\qquad
\inferrule{X\notin\mathsf{fv}(\psi)}{\textstyle
  (\bigoplus_{X\sim d(E)} \varphi) \sep \psi \vdash \bigoplus_{X\sim d(E)} (\varphi\sep \psi)
}

\inferrule{
  X\notin\mathsf{fv}(\psi)
  \\
  \precise\psi
}{\textstyle
  \bigoplus_{X\sim d(E)} (\varphi\sep \psi) \vdash (\bigoplus_{X\sim d(E)} \varphi) \sep \psi
} 
\quad
\inferrule{
  X\notin\mathsf{fv}(\psi)
  \\
  \convex\psi
}{\textstyle
  \bigoplus_{X\sim d(E)} (\varphi\osep \psi) \vdash (\bigoplus_{X\sim d(E)} \varphi) \osep \psi
} 
\quad
\inferrule{
  X\notin\mathsf{fv}(\varphi)
  \quad
  \mathsf{convex}(\varphi)
}{\textstyle
  \bigoplus _{X \sim d(E)} \varphi \vdash \varphi
}

\end{mathpar}
\caption{Selected entailment laws, where $\varphi[E/X]$ denotes a syntactic substitution of $E$ for $X$ in $\varphi$. Recall that $m\in \{\st, \wk\}$ and when $m$ is omitted, $\sep = \sep_\st$.}
\label{fig:entailment}
\end{figure}


In the second row, we give some rules pertaining to $\bignd$. An outcome conjunction $\bigoplus_{X\sim d(E)}$ can be weakened to a $\bignd$ over $\supp(d(E))$, where $\supp(\bern{E}) \triangleq \{0, 1\}$ and $\supp(\unif{E}) \triangleq E$. A $\bignd_{X\in \{E\}}\varphi$ over a singleton set is the same as substituting $E$ for $X$ in $\varphi$. Finally, the bounds of a $\bignd$ can always be expanded, similar to how $P \Rightarrow P\vee Q$ in classical logic.

The final two rows pertain to outcome conjunctions, and each rule has a corresponding one (not shown) with $\bignd$ instead of $\bigoplus$.
Weakening can be performed inside of an outcome conjunction, and bound variables can be $\alpha$-renamed as long as the new variable name is fresh. As in Bluebell, the (strong) separating conjunction distributes over the outcome conjunction, so that assertions can be moved inside of an outcome conjunction, but this rule is invalid if $\sep$ is replaced by $\osep$, due to the possibility of correlations between $\psi$ and $d(E)$.

Factoring assertions out of an outcome conjunction is only supported in Bluebell for almost-sure assertions $\sure P$, whereas in \pcol, it can be performed for any \emph{precise} assertion, due to our semantics based on a direct sum.
We saw at the end of \Cref{sec:measure} that independent products (which model separating conjunctions) distribute over direct sums (which model outcome conjunctions), however the corresponding entailment $\bigoplus_{X\sim d( E)} (\varphi\sep \psi) \vdash(\bigoplus_{X\sim d( E)} \varphi)\sep \psi$ requires that $\psi$ is satisfied by the same model in each case, which can be guaranteed by forcing $\psi$ to be precise. In fact, this exact scenario also arose in the \textsc{RCond} and \textsc{RCase} rules of \psl, where an analogous concept called \emph{supported} was used to ensure soundness \cite{psl}. Replacing strong separation with weak separation, we only need $\psi$ to be convex---not precise---to factor it out of the outcome conjunction, since independence is not implied. Finally, outcome conjunctions over convex assertions that do not depend on the bound variable $X$ can be collapsed.

\section{Probabilistic Concurrent Outcome Logic}
\label{sec:logic}

Probabilistic Concurrent Outcome Logic (\pcol) specifications are given as triples of the form $I\vDash_m\triple\varphi{C}\psi$, where $\varphi$ and $\psi$ are probabilistic assertions (\Cref{sec:prob-assert}), $C\in\mathsf{Cmd}$ (\Cref{fig:syntax}), $I$ is a basic assertion, and $m\in \{\st,\wk\}$. Roughly speaking, the meaning of these triples is that if the states are initially distributed according to $\varphi$, then any invariant sensitive execution of $C$ with invariant $I$ will satisfy $\psi$. The mode $m$ dictates what kind of \emph{frame preservation} property the triple has.

Recall from \Cref{sec:pomset} that invariant sensitive execution requires the invariant states to be drawn from a finite set. For this reason, $I$ must be a \emph{finitary} basic assertion; formally, $\mathsf{finitary}(I)$ iff $\sem{I}_\Gamma$ is a finite set for any context $\Gamma$.
The formal validity definition of \pcol triples is below.

%

\begin{definition}[\pcol Triples]\label{def:triple}
The \pcol triple $I\vDash_m\triple\varphi{C}\psi$ is valid iff for all $\Gamma\colon\mathsf{LVar}\to\mathsf{Val}$, $\mu$,
and probability spaces $\P$, $\P_F$, and $\P' \in \P\diamond_m \P_F$ such that $\P' \preceq \mu$ and $\Gamma,\P\vDash\varphi\sep\sure{I}$, then:
\[
  \forall \nu \in \lin^{\sem{I}_\Gamma}\left(\de{C}\right)^\dagger(\mu).\quad
  \exists \Q.\ \exists \Q' \in \Q \diamond_m \P_F.\quad
    \Q' \preceq \nu
    \quad\text{and}\quad
    \Gamma,\Q\vDash \psi\sep\sure{I}
\]
\end{definition}
As in many separation logics, frame preservation is built into the semantics of the triples \cite{birkedal2007relational,iris}; in addition to quantifying over a probability space $\P$ to satisfy $\varphi$, we also quantify over a probability space $\P_F$, which describes unused resources and is preserved by the program execution. As with the separating conjunction, we will omit the $m$ when $m=\st$. In addition, since we defined probability spaces to operate over memories $\mem S$, without $\bot$, our triples are \emph{fault avoiding}, which is also a standard choice for separation logics \cite{yang2002semantic}. That is, if $I\vDash\triple\varphi{C}\psi$ is valid, then we know that $C$ will not encounter a memory fault starting from a distribution satisfying $\varphi$.

These triples also imply almost sure termination, or \emph{total correctness}.
We chose to pursue total correctness, as it aligns with probabilistic liveness properties that we are interested in (\eg see \Cref{sec:von-neumann}). In the probabilistic context, there is no single natural notion of partial correctness, and nontermination breaks parallel composition; composing a thread with a nonterminating thread alters the behavior of the first thread even without shared state.


In the remainder of this section, we will present inference rules for deriving \pcol triples. We write $I\vdash_m\triple\varphi{C}\psi$ to mean that a triple is derivable using these rules. All of the rules are sound with respect to \Cref{def:triple}.

\begin{restatable}[Soundness]{theorem}{soundnessthm}
For all of the rules in \Cref{fig:cmd-rules,fig:csl-rules,fig:structural-rules,fig:while-rule}, if $I\vdash_m\triple\varphi{C}\psi$ then $I\vDash_m\triple\varphi{C}\psi$.
\end{restatable}

\subsection{Rules for Sequential Commands}

The rules for sequential commands are given in \Cref{fig:cmd-rules}. Although the rules appear like the standard ones for Hoare-like logics \cite{hoarelogic} and separation logic \cite{sl,localreasoning}, they rely on the properties of linearization shown in \Cref{sec:pomset}.
In \ruleref{Skip}, the precondition is preserved by a no-op, and \ruleref{Seq} is the standard rule for sequential composition.

The rules for if statements are split into two cases, for when the precondition implies that the true or false branch will be taken, respectively, similar to standard Outcome Logic \cite{zilberstein2025outcome,zilberstein2025demonic}. These rules can be combined into a single rule for analyzing both branches using the various split rules, which we will introduce in \Cref{sec:structural}.

\begin{figure}
\small
\begin{mathpar}
  \ruledef{Skip}{\;}{
    I\vdash_m\triple\varphi\skp\varphi
  }

  \ruledef{Seq}{
    I\vdash_m\triple{\varphi}{C_1}{\vartheta}
    \quad
    I\vdash_m\triple{\vartheta}{C_2}{\psi}
  }{
    I\vdash_m\triple\varphi{C_1\fatsemi C_2}\psi
  }
\\
  \ruledef{IfT}{
    \varphi \Rightarrow \sure{b\mapsto \tru}
    \\
    I\vdash_m\triple\varphi{C_1}\psi
  }{
    I\vdash_m\triple\varphi{\iftf b{C_1}{C_2}}{\psi}
  }

  \ruledef{IfF}{
    \varphi \Rightarrow \sure{b\mapsto \fls}
    \\
    I\vdash_m\triple\varphi{C_2}\psi
  }{
    I\vdash_m\triple\varphi{\iftf b{C_1}{C_2}}{\psi}
  }

  \ruledef{Assign}{
    \varphi \Rightarrow \sure{e\mapsto E} \land (\psi \sep \sure{\own(x)})
  }{
    I\vdash_m\triple{\varphi}{x \coloneqq e}{\psi\sep \sure{x \mapsto E}}
  }

  \ruledef{Samp}{
    \varphi \Rightarrow \sure{ e \mapsto E} \land (\psi\sep\sure{\own(x)})
  }{
    I\vdash_m\triple{\varphi}{x\samp d( e)}{\psi \sep (x \sim d( E)}
  }
\end{mathpar}

\caption{Rules for Sequential Commands}
\label{fig:cmd-rules}
\end{figure}

Finally, we give rules for atomic actions. \ruleref{Assign} requires the precondition to determine that the program expression $e$ evaluates to the logical expression $E$, and that the variable $x$, which is being assigned, is owned by the current thread and is disjoint from the assertion $\psi$. This structure gives the flexibility to apply the rule both when $\varphi = \psi \sep \sure{x \mapsto E'}$ and $\varphi = \psi \sep\sure{\own(x)}$. \ruleref{Samp} has a similar requirement, but ultimately concludes that $x$ is distributed according to $d(E)$ rather than having a deterministic value.

\subsection{Concurrent Separation Logic Rules}
\label{sec:concur-rules}

Next, in \Cref{fig:csl-rules}, we have a variety of rules inspired by Concurrent Separation Logic (\csl) \cite{csl,brookes2004semantics,vafeiadis2011concurrent}.
First is the \ruleref{Par} rule for parallel composition. Although \ruleref{Par} looks like the analogous rule from \csl---aside from the condition about precision---the soundness of the rule is substantially more complicated due to the probabilistic interpretation of the separating conjunction. It is not hard to imagine situations where the scheduler can introduce correlation between variables. For example, in the following program (which we previously saw in \Cref{sec:ov-precise}), the scheduler could choose to schedule the $y \coloneqq 1$ action \emph{after} the sampling operation is resolved, meaning that it could make $x=y$ with probability 1, a clear correlation.
\[
  x \samp \bern{\tfrac12} \fatsemi y \coloneqq 0 \quad\mathlarger{\mathlarger\parallel}\quad y \coloneqq 1
\]
As such, the outcomes of the two threads will not be independent after being run concurrently, but rather only \emph{observably} independent in some restricted event space. By requiring the postconditions of each thread to be precise, we know that the probability of each measurable event must be specified exactly, so that the nondeterministic behavior of the scheduler will not be measurable (\Areft{lem:par-ind}{lem:par}). In the case of the program above, the strongest precise assertion about $y$ is $\sure{ y \in \{0,1\}}$, that $y$ is always either 0 or 1, and $\left(x \sim \bern{\frac12}\right) \sep \sure{ y \in \{0,1\}}$ is a valid postcondition for the program, since almost sure assertions $\sure P$ are trivially independent from all other assertions.

More formally, the soundness proof uses the fact that $\psi_1$ and $\psi_2$ are precise to obtain unique minimal probability spaces $\Q_1$ and $\Q_2$ satisfying them. We then show that for any distribution $\nu$ resulting from running $C_1\parallel C_2$, and for any events $B_1\in\F_{\Q_1}$ and $B_2 \in\F_{\Q_2}$, it must be the case that $\nu(B_1 \sep B_2) = \mu_{\Q_1}(B_1)\cdot\mu_{\Q_2}(B_2)$. Since $C_1$ and $C_2$ may not terminate in a bounded amount of time, this probability only converges to the desired product in the limit.

\begin{figure}
\small
\begin{mathpar}
 \ruledef{Par}{
    I\vdash\triple{\varphi_1}{C_1}{\psi_1}
    \\
    I\vdash\triple{\varphi_2}{C_2}{\psi_2}
    \\
    \precise{\psi_1,\psi_2}
}{
    I\vdash\triple{\varphi_1\sep\varphi_2}{C_1 \parallel C_2}{\psi_1 \sep\psi_2}
}
\\
\ruledef{Atom}{
  J \vdash_m \triple{\varphi \sep \sure I}a{\psi\sep \sure I}
}{
  {I \sep J}\vdash_m\triple\varphi{a}\psi
}

\ruledef{Share}{
  I\sep J \vdash_m \triple{\varphi}C\psi
  \quad
  \mathsf{finitary}(I)
}{
  J \vdash_m \triple{\varphi\sep \sure I}C{\psi\sep \sure I}
}

\ruledef{Frame}{
  I\vdash_m\triple{\varphi}C{\psi}
}{
  I\vdash_m\triple{\varphi\ast_m \vartheta}C{\psi\ast_m\vartheta}
}

\ruledef{Weaken}{
  I\vdash\triple{\varphi}C\psi
}{
  I\vdash_\wk\triple{\varphi}C\psi
}

\ruledef{Strengthen}{
  I\vdash_\wk\triple{\varphi}C\psi
  \;\;
  \precise{\psi}
}{
  I\vdash\triple{\varphi}C\psi
}
\end{mathpar}
\caption{Concurrent Separation Logic Rules}
\label{fig:csl-rules}
\end{figure}

The next two rules are for interacting with invariants. The \ruleref{Atom} rule opens the invariant by moving it into the triple as an almost sure assertion, as long as the program is a single atomic action $a$. The fact that that the program executes atomically, and that $I$ is true before and after execution, means that $I$ is true at every step.
Next, the \ruleref{Share} rule allows a finitary almost sure assertion $I$ to be moved into the invariant. The soundness of this rule relies on the invariant monotonicity property that we discussed in \Cref{sec:pomset}. 
\begin{restatable}[Invariant Monotonicity]{lemma}{invMono}\label{lem:inv-mono}
For any $U,V,W \subseteq \mathsf{Var}$ and $\sigma \in \mem{W}$ such that $U\cap V = \emptyset$, $\I \subseteq \mem U$, $\J \subseteq \mem V$, and $U\cup V \subseteq W$:
\[
  \lin^{\I \sep \J}(\Alpha)(\sigma) \lec \lin^\I(\Alpha)(\sigma)
\]
\end{restatable}
Recall that $\lec$ is equivalent to $\supseteq$, so invariant monotonicity states that expanding the invariant (via $\sep$) can only add new behaviors to the set of outcomes. \Cref{lem:inv-mono} follows from the more general monotonicity property of linearization \cite{zilberstein2025denotational}.

The \ruleref{Frame} rule allows a local specification to be lifted into a larger memory footprint \cite{yang2002semantic}. The type of separation used depends on the mode $m$ of the triple. If $m=\st$, then the frame $\vartheta$ not only represents a disjoint set of physical \emph{resources}, but also that those resources are distributed independently from the information about the present program. A strong triple can always be weakened to a weak triple using the \ruleref{Weaken} rule. A weak triple can be strengthened---via \ruleref{Strengthen}---as long as the postcondition is precise. Just as with the \ruleref{Par} rule, precision here ensures that the scheduler cannot force any correlation between the postcondition and the frame.

\subsection{Structural and Outcome Splitting Rules}
\label{sec:structural}

Additional structural rules are given in \Cref{fig:structural-rules}. 
The first four splitting rules enable pointwise reasoning over outcome conjunctions, similar to those of Demonic Outcome Logic \cite{zilberstein2025demonic}. All these rules require that the logical variable $X$, bound by the outcome conjunction, does not appear free in the invariant $I$, since $X$ is unbound in the premise of the rule. If $X$ is free in $I$, then the rule can be applied after $\alpha$-renaming $X$ in the outcome conjunction (see \Cref{fig:entailment}).

The first rule, \ruleref{Split1}, requires that $\psi$ dictates the partition of the probability spaces in order to construct a final direct sum. This is done in a similar fashion to rules for establishing precision that we saw in \Cref{sec:prob-assert}---by requiring that $\psi\Rightarrow\sure{e\mapsto X}$ for some expression $e$. Since $X$ takes on distinct values in each case of the direct sum, then $e\mapsto X$ witnesses that the sample space can be partitioned. If $\psi$ does not witness a partition, then the \ruleref{Split2} rule can instead be used, which requires $\psi$ to be convex and not dependent on $X$.

\begin{figure}
\small
\begin{mathpar}
\ruledef{Split1}{
  I\vdash_m\triple{\varphi}C{\psi}
  \quad
  \psi \Rightarrow \sure{e \mapsto X}
  \quad
  X\notin\mathsf{fv}(I)
}{
  I\vdash_m\triple{\smashoperator{\bigoplus_{X\sim d( E)}}\varphi}C{\smashoperator{\bigoplus_{X\sim d( E)}}\psi}
}

\ruledef{NSplit1}{
  I\vdash_m\triple{\varphi}C{\psi}
  \quad
  \psi \Rightarrow \sure{e \mapsto X}
  \quad
  X\notin\mathsf{fv}(I)
}{
  I\vdash_m\triple{\smashoperator{\bignd_{X\in E}}\varphi}C{\smashoperator{\bignd_{X \in E}}\psi}
}

\ruledef{Split2}{
  I\vdash_m\triple{\varphi}C{\psi}
  \quad
  \convex{\psi}
  \quad
  X\notin\mathsf{fv}(I, \psi)
}{
  I\vdash_m\triple{\smashoperator{\bigoplus_{X\sim d( E)}}\varphi}C{\psi}
}

\ruledef{NSplit2}{
  I\vdash_m\triple{\varphi}C{\psi}
  \quad
  \convex{\psi}
  \quad
  X\notin\mathsf{fv}(I, \psi)
}{
  I\vdash_m\triple{\smashoperator{\bignd_{X\in E}}\varphi}C{\psi}
}
  
\ruledef{Exists}{
  I\vdash_\wk\triple{\textstyle\bignd_{X\in E}\sure{P}}C{\psi}
  \quad
  \sure{P}\Rightarrow\sure{e\mapsto X}
}{
  I\vdash_{\wk}\triple{\sure{\exists X\in E.\ P}}C{\psi}
}

\ruledef{Consequence}{
    \varphi'\Rightarrow\varphi
    \quad
    I\vdash_m\triple\varphi{C}\psi
    \quad
    \psi\Rightarrow\psi'
  }{
    I\vdash_m\triple{\varphi'}C{\psi'}
  }
\end{mathpar}
\caption{Structural and Outcome Splitting Rules}
\label{fig:structural-rules}
\end{figure}

We now demonstrate these two modes of use. Below, $x$ is initially distributed according to some distribution $d(E)$, and then it is incremented. The sample space is still partitioned after the increment, which is witnessed by $\sure{x - 1 \mapsto X}$, so \ruleref{Split1} can be used.
\begin{mathpar}
\inferrule*[right=\rulereff{Split1}]{
  \vdash\triple{\sure{x\mapsto X}}{x\coloneqq x+1}{\sure{x\mapsto X+1}}
  \\
  \sure{x\mapsto X+1} \Rightarrow \sure{x-1\mapsto X}
}{
  \vdash\triple{\smashoperator{\bigoplus_{X\sim d(E)}}\sure{x\mapsto X}}{x\coloneqq x+1}{\smashoperator{\bigoplus_{X\sim d(E)}}\sure{x\mapsto X+1}}
}
\end{mathpar}
On the other hand, the next program samples $y$, and the resulting sample space is not partitioned according to $X$. Instead, \ruleref{Split2} is used since the postcondition is convex.
\begin{mathpar}
\inferrule*[right=\rulereff{Split2}]{
  \vdash\triple{\sure{x\mapsto X\sep X\in\{0,1\}}}{y\samp \bern{{x}/2}}{\exists Y.\ \sure{Y\le 1/2} \sep y\sim\bern{Y}}
}{
  \vdash\triple{\smashoperator{\bigoplus_{X\sim \bern{1/2}}}\sure{x\mapsto X \sep X\in\{0,1\}}}{y\samp \bern{{x}/2}}{\exists Y.\ \sure{Y\le 1/2} \sep y\sim\bern{Y}}
}
\end{mathpar}
\ruleref{NSplit1} and \ruleref{NSplit2} are the nondeterministic analogues of the two aforementioned rules. They operate in exactly the same way when the distribution over $X$ is not known.

The \ruleref{Exists} rule allows for a more complex form of case analysis, where the logical variable being scrutinized, $X$, is bound by an existential quantifier inside of a pure assertion. Since the precondition is a pure assertion, it may be satisfied by a probability space that cannot measure the probability that $X$ takes on each value of $E$. Nonetheless, \ruleref{Exists} allows us to use a stronger precondition where $X$ is instead bound by a $\bignd$, allowing us to do case analysis using \ruleref{NSplit1} or \ruleref{NSplit2}. The tradeoff is that this form of reasoning is incompatible with strong frame preservation because turning a pure assertion into a nondeterministic outcome conjunction has a bad interaction with the strong separating conjunction, as shown below.
{\small\begin{align*}
  \sure{ y \in \{0, 1\} } \sep (x \sim\bern{p})
  \centernot\Rightarrow
  \Big(\bignd_{Y\in\{0,1\}} \sure{y\mapsto Y}\Big) \sep (x \sim\bern{p})
  \Rightarrow
  \bignd_{Y\in\{0,1\}} \Big( \sure{y\mapsto Y} \sep (x \sim\bern{p}) \Big)
\end{align*}}%
On the left hand side, it is possible that $x$ and $y$ are always equal. However, the first implication---which is unsound---implies that information about $x$'s distribution can be distributed into every outcome of $y$, \ie $x$ is distributed according to $\bern{p}$ \emph{for every} value of $y$. So, \ruleref{Exists} is compatible only with weak frame preservation, but it nevertheless provides an important capability to do case analysis over shared state, which we will see more concretely in \Cref{sec:ast-rules,sec:entropy-mixer,sec:von-neumann}. As long as the postcondition becomes precise at some later point, the \ruleref{Strengthen} rule can be used to regain strong frame preservation.
Finally, the rule of \ruleref{Consequence} allows pre- and postconditions to be manipulated in the standard way. The invariant can be neither strengthened nor weakened, since doing so would break the assumptions of other threads.

\subsection{Loops and Almost Sure Termination}
\label{sec:ast-rules}

The final proof rule is for analyzing while loops, and proving that they \emph{almost surely} terminate---that is, they terminate with probability 1. The \ruleref{BoundedRank} rule, shown in \Cref{fig:while-rule}, is based on rules for sequential programs due to \citet{mciver2005abstraction} and \citet{zilberstein2025demonic}.
It revolves around a loop invariant $\varphi$ and a rank. The rank $R$ must be integer-valued and bounded between $\ell$ and $h$, \ie $\varphi\Rightarrow\sure{\ell\le R\le h}$. As long as $R > \ell$, the loop continues to iterate ($\varphi\sep\sure{R>\ell}\Rightarrow\sure{b\mapsto\tru}$) and once $R$ reaches $\ell$, the loop must terminate ($\varphi\sep\sure{R=\ell} \Rightarrow\sure{b\mapsto\fls}$).

The premise of \ruleref{BoundedRank} guarantees that the rank strictly decreases each iteration with probability at least $p > 0$. Since the rank is bounded between $\ell$ and $h$, this means that from any start state, the loop is guaranteed to terminate with probability at least $p^{h-\ell} > 0$, which allows us to conclude that the program must almost surely terminate by the Zero-One law of \citet[Lemma 2.6.1]{mciver2005abstraction}.
Finally, since the terminating outcome $\varphi[\ell/R]$ is precise, there must be a unique minimal probability space $\Q$ that satisfies it. Similar to the soundness proof of the \ruleref{Par} rule, we complete the soundness proof of the \ruleref{BoundedRank} rule by showing that for all $B \in \F_\Q$ and all $\nu$ resulting from finite approximations of the loop, $\nu(B)$ converges to $\mu_\Q(B)$.

\begin{figure}
\small
\begin{mathpar}
\ruledef{BoundedRank}{
  \substack{\varphi \Rightarrow \sure{\ell \le R \le h} \\ \varphi\sep\sure{R=\ell}\Rightarrow\sure{b\mapsto\fls} \\ \varphi\sep\sure{R>\ell}\Rightarrow\sure{b\mapsto\tru}}
  \quad
  I\vdash_m\triple{\varphi \sep \sure{R = N > \ell}}{C}{\textstyle (\bignd_{R = \ell}^{N-1}\varphi) \oplus_{\ge p} (\bignd_{R = N}^h \varphi)}
  \quad
  \substack{
    0 < p \le 1 \\
    N \notin\free(\varphi) \\
    \precise{\varphi[\ell/R]}
  }
}{
  I\vdash_m\triple{\textstyle\bignd_{R=\ell}^h\varphi}{\whl bC}{\varphi[\ell/R]}
}
\end{mathpar}
\caption{The {\textrm{\textsc{BoundedRank}}} rule for almost sure termination.
}
\label{fig:while-rule}
\end{figure}

As an example of how to apply this rule, consider the example program below. The program implements a sort of random walk where $x$ moves towards the origin with probability $\frac12$, otherwise it is updated to $y$, which may be altered by a parallel thread.
\[
\code{RandWalk} \quad=\quad \whl{x > 0}{\;\;
\underbrace{
 b \samp \textstyle\bern{\frac12}\fatsemi
 \iftf{b}{
  x \coloneqq x-1
}{
  x \coloneqq y
}}_{C_\mathsf{body}}}
\]
As long as the value of $y$ is bounded, this loop almost surely terminates, which we can prove using \ruleref{BoundedRank}, subject to the resource invariant $I = y \in \{0, \ldots, 5\}$.
We will sketch the proof of almost sure termination here, the full derivation is shown in \Aref{app:ast-example}.
The loop invariant is $\varphi\triangleq{\sure{x \mapsto R \sep 0\le R\le 5 \sep \own(b)}}$, essentially just stating that $x$ is between 0 and 5. The rank $R$ is the value of $x$, and so $R$ decreases each iteration with probability at least $\frac12$. However, $x$ may also be updated to $y$, which is nondeterministic, requiring a use of the \ruleref{Exists} rule to conclude that $\bignd_{R=1}^5\sure{x \mapsto R}$ after the command $x \coloneqq y$. Clearly $\bignd_{R=1}^5\sure{x \mapsto R}$ is not precise, so the resulting triple is not strongly frame preserving. Ultimately, we get the following weak triple for the loop body, which states that the rank strictly decreases with probability at least $\frac12$.
\begin{equation}\label{eq:ast-body}\small
  y\in\{0, \ldots, 5\} \vdash_\wk\triple{\sure{x \mapsto N}}{
    C_\mathsf{body}}{\textstyle
      \big( \bignd_{R=0}^{N-1}\sure{x \mapsto R} \big) \oplus_{\ge\frac12} \big( \bignd_{R=N}^{5}\sure{x \mapsto R} \big)
    }
\end{equation}
At this stage, the weak triple signifies that we do not know exactly how likely the program is to terminate; we can only \emph{bound} the probability. Indeed, the scheduler can influence the likelihood that $x$ takes on particular values, and can force $x$ to be correlated with other state. For example, suppose we wanted to apply the \ruleref{Frame} rule with $z \sim \unif{\{0,\ldots, 5\}}$. The scheduler could choose to always make $y$ and $z$ equal, in which case $z \sim \unif{\{0,\ldots, 5\}}$ would not be independent of the postcondition of (\ref{eq:ast-body}).
However, the point of \ruleref{BoundedRank} is that the scheduler \emph{cannot} affect the probability of eventual termination. As such, the postcondition at the end of the execution is $\varphi[0/R]$, which is equivalent to the precise assertion $\sure{x\mapsto 0}$, allowing us to apply \ruleref{Strengthen} to get the following strong triple for the entire program.
\[
  y \in \{0, \ldots, 5\} \vdash \triple{\sure{x \in \{0, \ldots, 5\}}}{\code{RandWalk}}{\sure{x\mapsto 0}}
\]
One remarkable aspect of \pcol is that the triple above also implies that $\code{RandWalk}$ almost surely terminates when run in parallel with any other almost surely terminating program that obeys the resource invariant $I$. We get this property for free from the \ruleref{Par} rule, which guarantees that the probability of any infinite interleaving of both programs converges to 0.

This holds without any fairness assumption; our semantics guarantees that each thread almost surely terminates regardless of interference from any other threads, therefore each enabled action is almost surely scheduled within a finite amount of time. Of course, there are many programs that only almost surely terminate subject to a fair scheduler, including probabilistic consensus and synchronization protocols \cite{ben-or1983another,lehmann1981advantages,rabin1980n-process}. As we discuss in \Cref{sec:discussion}, we plan to augment \pcol with capabilities to reason about fair termination in the future.

\section{Examples}
\label{sec:examples}

In this section, we present four examples to demonstrate how the proof rules of \pcol come together into more complex derivations. Proofs are sketched here, and shown fully in \Aref{app:examples}.

\subsection{Entropy Mixer}
\label{sec:entropy-mixer}

There are many scenarios where several potential sources of entropy or randomness are available, which must be mixed together with the guarantee that if at least one of the sources of entropy is high quality, then the output will be at least that good. A simplified example of a such scenario is modeled in the following program, where $x_2$ is a reliable source of entropy, but $x_1$ is unreliable, because it is derived from $y$ in a way that can be controlled adversarially by the scheduler. Despite that, $z$, which is derived from $x_1$ and $x_2$ is a high quality source of randomness.
\[
  \code{EntropyMixer} \quad\triangleq\quad
  y \coloneqq 0 \fatsemi \left(\;\; x_1 \coloneqq y \fatsemi x_2 \samp \bern{\tfrac12} \fatsemi z \coloneqq \xor(x_1, x_2)  \;\;\mathlarger{\mathlarger\parallel}\;\;
  y \coloneqq 1 \;\;\right)
\]
We will analyze this program using the invariant $I = (y \in \{0,1\})$, and conclude in the end that $z \sim \bern{\frac12}$. It is easy to see that the second thread satisfies the invariant, so we will focus on the first thread.
We first show how information about $y$ can be extracted from the invariant in order to give a specification for the assignment to $x_1$. The use of the \ruleref{Exists} rule results in a weak triple.
\[\footnotesize
\inferrule*[right=\rulereff{Atom}]{
  \inferrule*[Right={\rulereff{Exists}}]{
    \inferrule*[Right=\rulereff{Consequence}]{
      \inferrule*[Right=\rulereff{NSplit1}]{
        \inferrule*[Right=\rulereff{Assign}]{\;}{
          \vdash_\wk\triple{\sure{y\mapsto Y \sep \own(x_1)}}{x_1 \coloneqq y}{\sure{x_1 \mapsto Y \sep y \mapsto Y }}
        } 
      }{
        \vdash_\wk\triple{\textstyle\bignd_{Y\in\{0,1\}} \sure{y\mapsto Y \sep {\own(x_1)}}}{x_1 \coloneqq y}{\textstyle\bignd_{Y\in\{0,1\}} \sure{x_1 \mapsto Y \sep y \mapsto Y}}
      } 
    }{
      \vdash_\wk\triple{\textstyle\bignd_{Y\in\{0,1\}} \sure{y\mapsto Y \sep {\own(x_1)}}}{x_1 \coloneqq y}{(\textstyle\bignd_{Y\in\{0,1\}} \sure{x_1 \mapsto Y}) \sep \sure{y \in \{0,1\}}}
    } 
  }{
    \vdash_\wk\triple{\sure{\own(x_1)} \sep \sure{y \in \{0,1\}}}{x_1 \coloneqq y}{(\textstyle\bignd_{Y\in\{0,1\}} \sure{x_1 \mapsto Y}) \sep \sure{y \in \{0,1\}}}
  } 
}{
  y\in\{0,1\}\vdash_\wk\triple{\sure{\own(x_1)}}{x_1 \coloneqq y}{\textstyle\bignd_{Y\in\{0,1\}} \sure{x_1 \mapsto Y}}
} 
\]
These derivations are best read moving up from the lowermost precondition and then down from the topmost postcondition.
First, \ruleref{Atom} is applied to open the invariant. Next, we use \ruleref{Exists} and \ruleref{NSplit1} to gain access to the value of $y$, so that we can apply the \ruleref{Assign} rule. After the assignment, we use \ruleref{Consequence} to weaken the information about $y$ and move it outside the scope of the $\bignd$ so that we can close the invariant. Now, we move on to the derivation for the remainder of the thread.
\[\footnotesize
    \inferrule*[right=\rulereff{NSplit2}]{
      \inferrule*[Right=\rulereff{Consequence}]{
        \inferrule*{\vdots}{
          I\vdash_\wk\triple{ \sure{x_1 \mapsto Y \sep \own(x_2, z)}}{x_2 \samp \bern{1/2} \fatsemi z \coloneqq \xor(x_1, x_2)}{\textstyle\bigoplus_{X \sim \bern{1/2}} \sure{z\mapsto \xor(Y, X)}}
        }
      }{
        I\vdash_\wk\triple{ \sure{x_1 \mapsto Y \sep \own(x_2, z)}}{x_2 \samp \bern{1/2} \fatsemi z \coloneqq \xor(x_1, x_2)}{z \sim \bern{1/2}}
      } 
    }{
      I\vdash_\wk\triple{\textstyle\bignd_{Y\in\{0,1\}} \sure{x_1 \mapsto Y \sep \own(x_2, z)}}{x_2 \samp \bern{1/2} \fatsemi z \coloneqq \xor(x_1, x_2)}{z \sim \bern{1/2}}
    } 
\]
First, we must again do case analysis on the logical variable $Y$, but this time we use \ruleref{NSplit2}, which requires the postcondition to be convex and not depend on $Y$. In the premise of \ruleref{NSplit2}, we we show that $z$ is distributed properly regardless of the value of $Y$. For any fixed $Y$, after executing the writes to $x_2$ and $z$, we know that $\bigoplus_{X\sim\bern{1/2}} \sure{z \mapsto \xor(Y,X)}$ (the proof is quite mechanical, and shown in \Aref{app:entropy-mixer}).
Since $Y$ is constant, then $\xor(Y,X)$ is a bijection from $\{0,1\}$ to $\{0,1\}$, and we can therefore use the rule of \ruleref{Consequence} to conclude that $z$ is uniformly distributed. Since that postcondition is precise, we can strengthen the triple. After combining the two threads with \ruleref{Par}, we get the following specification for the whole program.
\[
  \vdash\triple{\sure{\own(x_1, x_2, y, z)}}{\code{EntropyMixer}}{z \sim \bern{\tfrac12}}
\]

\subsection{Concurrent Shuffle}
\label{sec:shuffle}
\newcommand{\app}{\mathbin{+\!+}}

\begin{figure}
\[
\begin{array}{ll}
\begin{array}{l}
\code{ConcurrentShuffle}:\\
\quad a_1 \coloneqq [] \fatsemi  a_2 \coloneqq [] \fatsemi i\coloneqq 0 \fse \\
\quad \whl{i < \code{len}(a)}{} ( \\
\qquad b \samp \bern{\frac12} \fse \\
\qquad \iftf b{a_1 \coloneqq a_1 \app [a[i]]}{a_2 \coloneqq a_2 \app [a[i]]} \fse \\
\qquad i\coloneqq i+1 \\
\quad) \fse \\
\quad \code{shuffle}_1 \parallel \code{shuffle}_2 \fse \\
\quad a \coloneqq a_1 \app a_2
\end{array}
&
\begin{array}{l}
\code{shuffle}_k: \\
\quad i_k \coloneqq \code{len}(a_k) \fse \\
\quad \whl{i_k > 1}{} \\
\qquad j_k \samp \unif{[0, \ldots, i_k]} \fse \\
\qquad a_k \coloneqq \code{swap}(a_k, i_k, j_k) \fse \\
\qquad i_k \coloneqq i_k -1
\\\; \\\; \\\;
\end{array}
\end{array}
\]
\caption{A concurrent shuffling algorithm, where $\ell_1\app\ell_2$ concatenates two lists and $\code{swap}(\ell, i, j)$ returns $\ell$ with the elements $\ell[i]$ and $\ell[j]$ swapped.}
\label{fig:con-shuff}
\end{figure}

\citet{bacher2015mergeshuffle} showed that shuffling algorithms can be made up to seven times faster through parallelization. They introduced a divide-and-conquer algorithm in which sub-arrays are shuffled concurrently and then merged. In this example, we prove the correctness of a simple concurrent shuffle algorithm using \pcol. The program is shown in \Cref{fig:con-shuff}.
First, the elements in the array are randomly assigned to two buckets, $a_1$ and $a_2$, the buckets are then shuffled in parallel using the standard \citet{fisher1938statistical} algorithm, and then the two shuffled sub-lists are concatenated together. For some list $\ell$, let $\Pi(\ell)$ be the set of all permutations of $\ell$. For the purposes of this example, we will presume that all lists do not contain duplicate values. The specification of the Fisher-Yates shuffle is shown below (and proven in \Aref{app:shuffle}). That is, if a list $A$ is stored in $a_k$, then $a_k$ will hold a uniformly chosen permutation of that list after execution of $\code{shuffle}_k$.
\[
  \vdash\triple{\sure{a_k \mapsto A \sep \own(i_k, j_k)}}{\code{shuffle}_k}{a_k \sim \unif{\Pi(A)}}
\]
For some list $\ell$ and bit-string $x$, let $\ell[x]$ be the list obtained by filtering $\ell$ to only contain the indices $i$ such that $x[i] = 1$, \eg $[1, 2, 3, 4][ 1001 ] = [1, 4]$.
After the execution of the bucketing loop, we get the  postcondition
$
  \bigoplus_{X \sim \unif{\{0,1\}^{\code{len}(A)}}} \sure{a_1 \mapsto A[X] \sep a_2 \mapsto A[\neg X] }
$, where $X$ is a uniformly chosen bit-string that dictates into which bucket each element of $A$ is placed, and $\neg X$ is bitwise logical negation. Next, to analyze the concurrent calls to $\code{shuffle}_k$, we first use \ruleref{Split1} so that we can separate $a_1$ and $a_2$. Using \ruleref{Par} with the triple for $\code{shuffle}_k$, we have that $a_1$ and $a_2$ are uniformly and independently distributed permutations of the initial lists after the concurrent shuffles.
\[\footnotesize
  \inferrule*[right=\rulereff{Par}]{
      \vdash\triple{\sure{a_1 \mapsto A[X]}}{\code{shuffle}_1}{a_1 \sim \unif{\Pi(A[X])}}
      \quad
      \vdash\triple{\sure{a_2 \mapsto A[\neg X]}}{\code{shuffle}_2}{a_2 \sim \unif{\Pi(A[\neg X])}}
  }{
    \vdash\triple{\sure{a_1 \mapsto A[X] \sep a_2 \mapsto A[\neg X]}}{\code{shuffle}_1 \parallel \code{shuffle}_2}{(a_1 \sim \unif{\Pi(A[X])}) \sep (a_2 \sim \unif{\Pi(A[\neg X]})}
  }
\]
Reapplying the outcome conjunction over $X$, we can rewrite the postcondition to be:
\[\textstyle
  \bigoplus_{X \sim \unif{\{0,1\}^{\code{len}(A)}}} \bigoplus_{A_1 \sim \unif{\Pi(A[X])}} \bigoplus_{A_2\sim\unif{\Pi(A[\neg X])}} \sure{a_1 \mapsto A_1 \sep a_2 \mapsto A_2}
\]
Using \ruleref{Split1} three times and \ruleref{Assign}, we get the following postcondition after the assignment to $a$.
\[\textstyle
  \bigoplus_{X \sim \unif{\{0,1\}^{\code{len}(A)}}} \bigoplus_{A_1 \sim \unif{\Pi(A[X])}} \bigoplus_{A_2\sim\unif{\Pi(A[\neg X])}} \sure{a \mapsto A_1 \app A_2}
\]
We now prove the assertion above implies that $a$ is a uniform permutation. Take any $\ell \in \Pi(A)$ and let $n = \code{len}(A)$. For every $0\le k\le n$, there is exactly one split $X$ such that $X$ assigns the elements $\ell[0], \ldots, \ell[k-1]$ to $a_1$ and the rest to $a_2$. This split occurs with probability $\frac1{2^n}$.
Further, given that this correct split has occurred, $a_1$ and $a_2$ are shuffled so that $a_1 \app a_2 = \ell$ with probability $\frac{1}{k!}\frac{1}{(n-k)!}$. Thus the probability of getting the permutation $\ell$ is $\sum_{k=0}^n \frac1{2^n}\frac{1}{k!}\frac{1}{(n-k)!} = \frac1{n!}$.
Since there are exactly $|\Pi(A)| = n!$ permutations, this means that these permutations are produced uniformly, therefore:
\[
  \vdash\triple{\sure{a\mapsto A \sep \own(\cdots)}}{\code{ConcurrentShuffle}}{a \sim \unif{\Pi(A)}}
\]

\subsection{Private Information Retrieval}
\label{sec:pir}

\begin{figure}
\[
\begin{array}{l}
  \code{PrivFetch}: \\
  \quad q_1 \samp \unif{ \{0, 1\}^n} \fse \\
  \quad q_2 \coloneqq \xor(q_1, x) \fse \\
  \quad \code{fetch}_1 \parallel \code{fetch}_2 \fse \\
  \quad r \coloneqq \xor(r_1, r_2)
  \\\;
\end{array}
\qquad\qquad
\begin{array}{l}
\code{fetch}_k: \\
  \quad i_k \coloneqq 0 \fatsemi
  r_k \coloneqq 0 \fse \\
  \quad \whl{i_k < \code{len}(q_k)}{\;} \\
  \qquad \ift{q_k[i_k] = 1}{\;} \\
  \qquad\quad r_k \coloneqq \xor(r_k, d[i_k]) \fse \\
  \qquad i_k\coloneqq i_k+1
\end{array}
\]
\caption{A concurrent private information retrieval protocol.}
\label{fig:pir}
\end{figure}

Private information retrieval allows a user to fetch data without the database operator learning what data was requested \cite{chor1998private}.
A simple form of private information retrieval is modeled in the program shown in \Cref{fig:pir}. The $\code{fetch}_k$ programs process a bit string query $q_k$, indicating which entries of the database $d$ to return. Those entries are then bitwise xor'ed together.
Private retrieval is implemented in \code{PrivFetch}. The input $x$ is a \emph{one-hot} bit string $\code{onehot}(K)$, with a 1 in position $K$---indicating the index of the data to retrieve---and zeros everywhere else. Two queries are then made concurrently. The first one uses a randomly chosen bit string, and the second uses the same random string xor'ed with $\code{onehot}(K)$. The final data is an xor of the two responses, which reveals the data at position $K$.
\citet{psl} proved a similar example in \psl, but their version was sequential; both fetches happened within a single for loop. Our version better models a distributed system where the computation does not occur in lockstep. We first present a specification for $\code{fetch}_k$ (proven in \Aref{app:pir}), which states that $r$ is an xor of data entries $i$, such that $q_k[i] = 1$, subject to the invariant that $d\mapsto D$. 
\[
  d\mapsto D \vdash \triple{
    \sure{q_k\mapsto Q} \sep \sure{\own(i_k, r_k)}
  }{
    \code{fetch}_k
  }{
    \sure{r_k \mapsto \xor_{0 \le i < n : Q[i] = 1} D[i]}
  }
\]
We now sketch the derivation of the main procedure, the full proof is shown in \Aref{app:pir}. After sampling into $q_1$, we use \ruleref{Split2} to make the outcome of that query deterministic. After assigning $q_2$, we get $\sure{q_1 \mapsto Q} \sep \sure{q_2 \mapsto \xor(Q, \code{onehot}(K))}$. We can then apply the \ruleref{Par} rule to analyze the concurrent fetches to get the following postcondition:
\[
  {\sure{r_1 \mapsto \xor_{0 \le i < n : Q[i] = 1} D[i]} \sep \sure{r_2 \mapsto \xor_{0 \le i < n : \xor(Q, \code{onehot}(K))[i] = 1} D[i]}}
\]
The value of $r$ is the xor of $r_1$ and $r_2$, which differ only at index $K$, therefore we can conclude that $\sure{r\mapsto D[K]}$. Since this assertion is convex and does not depend on $Q$, we meet the side conditions of the $\ruleref{Cond2}$ rule, and therefore the final postcondition is simply $\sure{r\mapsto D[K]}$.
\[
\vdash\triple{\sure{ x \mapsto \code{onehot}(K) \sep d\mapsto D \sep \own(q_1,q_1,r_1,r_2,r,i_1,i_2)}}{
    \code{PrivFetch}
  }{\sure{r \mapsto D[K]}}
\]

\subsection{The von Neumann Trick}
\label{sec:von-neumann}

The \citet{neumann1951various} trick is a protocol for simulating a fair coin given a coin with unknown bias. The biased coin is flipped twice, if the outcome is 1-0 or 0-1, then the output is the first coin flip, otherwise the experiment is repeated. Below, we have a variant of the von Neumann trick where the coin's bias is stored in shared memory, and its value is not remembered across rounds.
\[
\code{vonNeumann} \quad\triangleq\quad
x \coloneqq 0 \fatsemi y \coloneqq 0 \fse 
\whl{x=y}{\;\;
p' \coloneqq p \fatsemi
x \samp\bern{p'} \fatsemi
y \samp\bern{p'}}
\]
Every round, a concurrent thread could change the bias, altering the probability of terminating in that round. Despite this, we show that the program almost surely terminates, and that $x$ is distributed according to a fair coin, subject to the invariant $I = ( p \in [\varepsilon, 1-\varepsilon]_\Delta)$ where $0 <\varepsilon \le \frac12$ is an arbitrarily small probability and $\Delta$ is a nonzero step size, making the interval finite (if $p$ could be 0 or 1, then the program may not terminate). A variant of this example appeared in \citet{zilberstein2025demonic}, where the bias was explicitly altered by an adversary. This concurrent version of the program introduces new challenges from a program analysis perspective.
The bulk of the derivation involves analyzing the while loop. To do so, we need a loop invariant, which is shown below.
\[
  \varphi \triangleq \varphi_0 \vee \varphi_1
   \quad
   \varphi_0\triangleq \smashoperator{\bigoplus_{X \sim \bern{1/2}}} \sure{x \mapsto X \sep y \mapsto \lnot X \sep R=0}
   \quad
    \varphi_1\triangleq\sure{x=y\mapsto \tru \sep R=1 \sep \own(p')}
\]
The rank $R$ is either 0 or 1. When $R=0$, $x \neq y$ and $x$ is uniformly distributed, so the loop terminates. When $R=1$, $x=y$, so the loop keeps iterating. Each iteration, the loop terminates with probability at least $2\varepsilon(1-\varepsilon)$. Due to the reliance on shared state, this probability is not exact, but only a bound.
The structure of the proof is similar to that of \Cref{sec:entropy-mixer}. First, we use \ruleref{Atom} and \ruleref{NSplit1} to open the invariant and conclude that $p'$ holds some probability in $[\varepsilon, 1-\varepsilon]_\Delta$. Since the postcondition is imprecise and based on shared state, we are only able to obtain a weak triple at this point.
\[
  I\vdash_\wk\triple{\sure{\own(p')}}{p' \coloneqq p}{\textstyle\bignd_{X \in [\varepsilon,1-\varepsilon]_\Delta} \sure{p' \mapsto X}}
\]
We then sequence this with the two sampling operations, shown below:
\[\footnotesize
\hspace{-2em}
\inferrule*[right=\rulereff{NSplit2}]{
  \inferrule*[Right=\rulereff{Consequence}]{
    \inferrule*{\vdots}{
      I\vdash_\wk\triple{\sure{p'\mapsto X \sep \own(x, y)}}{x \samp \bern{p'}\fatsemi y\samp\bern{p'}}{\sure{p'\mapsto X} \sep x\sim\bern{X}\sep y\sim\bern{X}}
    }
  }{
    I\vdash_\wk\triple{\sure{p'\mapsto X \sep \own(x, y)}}{x \samp \bern{p'}\fatsemi y\samp\bern{p'}}{\varphi_0 \oplus_{\ge 2\varepsilon(1-\varepsilon)} \varphi_1}
  }
}{
  I\vdash_\wk\triple{\textstyle\bignd_{X \in [\varepsilon, 1-\varepsilon]_\Delta} \sure{p'\mapsto X \sep \own(x, y)}}{x \samp \bern{p'}\fatsemi y\samp\bern{p'}}{\varphi_0 \oplus_{\ge 2\varepsilon(1-\varepsilon)} \varphi_1}
}
\]
We first use \ruleref{NSplit2} to show that for any fixed probability $X$, the loop terminates with probability at least $2\varepsilon(1-\varepsilon)$. This is done using straightforward combinatorial reasoning. Given that $p'\mapsto X$, we know that $x$ and $y$ will be independently and identically distributed according to $\bern{X}$. That means that $\sure{x\mapsto 1 \sep y\mapsto 0}$ and $\sure{x\mapsto 0 \sep y\mapsto 1}$ both occur with probability $X(1-X)$, so $\varphi_0$ occurs with probability $2X(1-X)$, and otherwise $x=y$, so $\varphi_1$ holds. Since we know that $X \in [\varepsilon, 1-\varepsilon]$, then clearly $2\varepsilon(1-\varepsilon) \le 2X(1-X)$, so the consequence above is valid.
After applying \ruleref{BoundedRank}, we get the postcondition $\varphi_0 \Rightarrow (x \sim \bern{\frac12})$. Since this postcondition is precise, we can use \ruleref{Strengthen} to get the following strong triple for the whole program, indicating both that the program almost surely terminates, and that $x$ is distributed like a fair coin. In fact, it also terminates when composed in parallel with other terminating threads that alter $p$ in arbitrary ways!
\[
  p \in [\varepsilon, 1-\varepsilon]_\Delta \vdash\triple{\sure{\own(x, y, p')}}{\code{vonNeumann}}{x \sim \bern{\tfrac12}}
\]

\section{Related Work}
\label{sec:related}

\subsubsection*{Logics for Probabilistic Concurrency}
Polaris is a relational separation logic built on Iris \cite{iris} for reasoning about concurrent probabilistic programs \cite{tassarotti2018verifying,polaris}. 
Compared to Polaris, \pcol has two main advantages: it supports unbounded looping and it supports direct probabilistic reasoning about the distribution of outcomes. Polaris cannot handle the von Neumann trick, which involves unbounded looping. Because it is a relational logic, Polaris works by relating a randomized concurrent program to a specification program, which is randomized and nondeterministic (but not concurrent). This requires the random choices between the two programs to be coupled in lockstep.

In our concurrent shuffling example, we prove with \pcol that a sequence of many random choices results in a completely different uniform distribution of outcomes. It would be impossible to recreate this proof in Polaris, since the natural specification program for a randomized shuffle would only make a single random choice (a uniform sample over permutations), whereas the actual program makes a large sequence of random choices.
It might be possible to do an alternate relational proof in Polaris by instead picking a specification program that makes a similar sequence of random choices, but then one would be left with the challenge of proving that such a specification program actually generated a uniform shuffle, for which Polaris cannot help.

In \pcol, we combine the two steps, avoiding the need to write down a specification program by building quantitative reasoning tools into the logic itself. However, the approaches do have similarities---our invariant sensitive semantics implicitly captures the behavior of the specification program by converting parallel manipulation of shared state into nondeterminism.
This is slightly easier, since the user must only write down an invariant rather than an entire specification program.
Then, rather than externally analyzing the specification program, we directly prove a quantitative postcondition. While \pcol does not support all the capabilities of Polaris (\eg ghost state, higher order state, etc.), proofs are carried out in fewer steps using a single, self-contained logic.
In the future, it may be fruitful to combine the two approaches. Indeed, \citet{bao2025bluebell} showed the advantages of supporting both relational and unary reasoning in a single probabilistic logic.

\citet{fesefeldt2022towards} pursued an alternative technique for reasoning about probabilistic concurrent programs, based on a \emph{quantitative} interpretation of separation logic \cite{qsl}. 
Their logic can be used to lower bound the probability of a single outcome, making all but the private information retrieval example in this paper out of reach. While \citeauthor{fesefeldt2022towards} do support unbounded looping, the probability of nontermination is always added to the final expected values, and so it cannot be used to prove almost sure termination.

\citet{fan2025program} developed another logic based on Probabilistic Rely-Guarantee \cite{mciver2016probabilistic}. While it supports outcome splitting, the program must be re-instrumented to explicitly declare where splitting occurs, whereas in \pcol splitting is purely logical and can be used anywhere. \citeauthor{fan2025program} require postconditions after splitting to be convex;  they have a rule similar to \ruleref{Split2}, but not \ruleref{Split1}. In addition, their parallel composition rules do not give independence guarantees, and accordingly can only make conclusions about the local distributions of all threads, and not the \emph{joint} distributions of all the variables in global memory, and programs must be proven to almost surely terminate externally to the logic.
Their logic is based on an oblivious adversary model---a weaker model than our unrestricted adversary. More programs are verifiable in this model, so it is often preferred, however the logic of \citeauthor{fan2025program} can only reason about obliviousness in limited ways, \ie by treating coin flip actions as atomic, so many programs are still out of reach.

\subsubsection*{Probabilistic Separation Logics}
Capturing probabilistic independence in separation logic was first explored by \citet{psl}, however the resulting Probabilistic Separation Logic (\psl) was limited in its ability to reason about control flow, and the frame rule had stringent side conditions. \DIBI later extended the \psl model to include conditioning \cite{bao2021bunched}. Lilac built on the two aforementioned logics and used conditioning to improve on \psl's handling of control flow, although without mutable state \cite{li2023lilac}. Lilac also added support for continuous distributions, and reformulated the notion of separation using probability spaces, making it more expressive. 

Lilac's lack of mutable state means that information about variables can be duplicated; for example, $x \sim \bern{\frac12} \sep \sure{ x \in \{0, 1\}}$ is satisfiable. 
Bluebell uses a similar model to Lilac, with the ability to reason about mutable state, which requires the logic to track permissions too \cite{bao2025bluebell}. We chose to use a more restrictive form of separation where information about variables cannot be shared, as it simplified the logical rules by eliminating the need for tracking permissions.


As we mentioned in \Cref{sec:overview}, the direct sum semantics of $\bigoplus$ differs from conditioning modalities $\mathsf{C}_{X \sim d}\varphi$, which are based on disintegration. However, we believe that in the future the two modalities can work together to do discrete case analysis over continuous programs. For example, consider the following program, which samples from a continuous distribution and then branches on the result.
\[
  x \samp \unif{[0,1]}\fatsemi \iftf{x \le \tfrac12}{C_1}{C_2}
\]
Lilac's rule for analyzing if statements is incompatible with mutable state, so we would instead need to use \ruleref{Split1} or \ruleref{Split2} to do case analysis on the guard $x\le\frac12$ in order to derive a specification for this program. However, those rules require an outcome conjunction, which cannot express the fact that $x$ is distributed according to a continuous distribution. Instead, we can partition the continuous distribution into two parts, joined with an outcome conjunction, which would then allow us to analyze the two branches of the if statement.
\[
  \mathsf{C}_{X \sim\unif{[0,1]}} \sure{x\mapsto X}
  \quad\implies\quad
  \left(\mathsf{C}_{X \sim\unif{[0,1/2]}} \sure{x\mapsto X}\right) \oplus_{\frac12} \left( \mathsf{C}_{X \sim\unif{(1/2, 1]}} \sure{x\mapsto X}\right)
\]
Although the direct sum's partitioning of the sample space imposes some limitations on the splitting rules, the benefit is that they are fully compositional; the use of splitting does not impede the application of other rules later in the derivation. This is in contrast with Bluebell's \textsc{c-wp-swap} \cite[\S 5.1]{bao2025bluebell}, which requires ownership over all program variables (denoted $\mathsf{own}_{\mathbb X}$), thereby precluding most later applications of the frame rule. Although \citet{bao2025bluebell} show fruitful uses of \textsc{c-wp-swap}, the restriction is not acceptable for a concurrency logic, since it would preclude use of the \ruleref{Par} rule.

This is similar to \ruleref{Exists} disabling strong frame preservation, but \ruleref{Exists} is used in specific scenarios for case analysis on nondeterministic shared state, which only arises in the concurrency setting.
In any case, strong frame preservation can always be reenabled if the postcondition is precise.
Non-frame-preserving operations arise in non-probabilistic separation logics too \cite{vindum2025nextgen,spies2025destabilizing}.
The idea of having two types of probabilistic separation was also explored in \textsf{LINA}, where weak separation corresponds to \emph{negative dependence} \cite{bao2022separation}.

The notion of \emph{precision} has been previously studied in separation logics, in part to explain when the separating conjunction distributes over other logical connectives, such as the regular conjunction \cite{semanticsep,vafeiadis2011concurrent}.
In \psl, precision (under the name \emph{supported}) was used to ensure that the guard of an if statement remains independent of the postcondition of the two branches. Conditioning \`a la Lilac and Bluebell provides a more flexible way to reason about control flow without forcing the guard to be independent of the states in the two branches.

Another category of probabilistic separation logics build on Iris \cite{iris}, from which they inherit expressive features, including ghost state and impredicative invariant reasoning.
In these logics, separating conjunctions have the usual meaning from \csl, and do not capture probabilistic independence.
\citet{lohse2024irisexpectedcostanalysis} and \citet{haselwarter2024tachishigherorderseparationlogic} develop logics for proving bounds on the expected runtime of a randomized program.
\citet{aguirre2024error} apply a similar approach for upper bounding the probability that a postcondition will fail to hold in sequential programs, and \citet{li2025modular} extended this work to the concurrent setting.
Additional logics have also been developed for relational reasoning and refinement \cite{gregersen2024almost-sure,ghaselwarter2024approximate,gregersen2024asynchronous}.
The tradeoff is that they focus on a narrow property about programs' probabilistic behaviors, \eg only capturing a bound on an expected cost or probability of a single event.
Outcome Separation Logic uses a more primitive form of heap separation, but is backed by a denotational model that supports specifications about the distribution of outcomes \cite{zilberstein2024outcome}.

%

\section{Conclusion}
\label{sec:discussion}

This paper brings together ideas from concurrent, probabilistic separation logics, and Demonic Outcome Logic in developing \pcol, a new expressive logic for analysis of probabilistic concurrent programs. Although \pcol represents a significant step in reasoning for randomized concurrent programs, more work remains to be done. Fine-grained concurrency analysis is notoriously complex, and we plan to augment \pcol in the following ways to support more expressive verification.

\subsubsection*{Fair Termination and Synchronization.}
Our \ruleref{BoundedRank} rule makes no assumptions about fairness (meaning that no thread can be indefinitely starved); indeed it applies only to programs for which the probability of eventual termination does not depend on the scheduler. 
However, many probabilistic synchronization protocols \cite{hart1983termination} such as the Dining Philosophers problem \cite{lehmann1981advantages} only terminate under a fair scheduler. We would like to extend \pcol to reason about these programs, but it will present significant challenges. Fairness is not a compositional property, so many of the properties of $\lin(\de{-})$ that we rely on for soundness would not hold.


\subsubsection*{Dynamic Allocation and Resource Algebras.}
As with \psl and all logics that build on it, our resource model uses variables rather than pointers.
However, most concurrent programs use pointers.
Modern \csl implementations such as Iris \cite{iris,iris1} use resource algebras, so that additional types of physical and logical state can be added to govern the ways in which concurrent threads modify shared resources. Bluebell already includes permissions, which help to duplicate knowledge about read-only variables \cite{bao2025bluebell}, however many other resources are used in practice and it is not yet understood how those resources interact with the independence model of separation. We plan to augment the model of \pcol to support Iris style resource algebras.

\subsubsection*{Mechanization}
As we saw in \Cref{sec:examples} and \Aref{app:examples}, \pcol derivations are quite involved---even for small programs---due to the handling of invariants, conditioning, and case analysis on shared state. Verification of larger programs would be infeasible with pen-and-paper proofs, therefore we plan to mechanize \pcol in the Lean proof assistant \cite{moura2015lean}. Lean is the ideal choice because much the underlying probability theory, domain theory, and topology needed to support \pcol are already formalized in mathlib \cite{community2020lean}, although more work would still be needed to formalize all the foundational theories of \pcol, such as the convex powerdomain.

\section*{Acknowledgements}

This work was supported by the National Science Foundation under awards 2504142 and 2504143 and ARIA’s Safeguarded AI programme.

 \bibliography{refs}
  
\ifx\extended\undefined\else
\allowdisplaybreaks
\appendix
\clearpage

\ifx\apponly\undefined\else
\setcounter{page}{1}
\fi

{\noindent \huge\bfseries\sffamily Appendix}

\footnotesize

\section{Definition of Program Semantics}
\label{app:pomset}

In this section, we give details on the \emph{Pomsets with Formulae} semantic model due to \citet{zilberstein2025denotational}.

\subsection{Pomsets with Formulae}

Whereas standard pomsets (partially ordered multisets) use a partial order to record the causality between atomic actions in the program (elements of a multiset) \cite{gischer1988equational,pratt1986modeling}, pomsets with formulae add a Boolean formula to each node, which records which tests must succeed or fail to reach that point in the execution. This is necessary to capture probabilistic concurrency, since tests may not always pass or fail, but rather they may only have some probability of passing, and so both paths must be represented in a single semantic structure.
We briefly introduce pomsets with formulae here, refer to \citet{zilberstein2025denotational} for a more complete treatment.

Let $\lab \triangleq \act \cup \test \cup \{ \bullet, \bot \}$, which can be either an action or test (from \Cref{fig:syntax}), a no-op ($\bullet$), or undetermined ($\bot$). Although $\bot$ labels will never appear in the denotation of a \emph{complete} program, they are used to indicate that a finite structure is an approximation of an infinite one, \ie in the semantics of while loops. Also, let $\Nodes$ be a countable universe of identifiers and $\Form$ be the set of Boolean formulae over $\Nodes$. For some $\psi \in \Form$, we write $\sat(\psi)$ to indicate that $\psi$ is satisfiable, and $\free(\psi) \subseteq \Nodes$ is the set of variables referenced in $\psi$. We now define the underlying structure.

\begin{definition}[Labelled Partial Order with Formulae (LPOFs)]\label{def:lpof}
  An LPOF is a 4-tuple $\tuple{N, <, \lambda, \varphi} \in \lpo$ where:
  \begin{enumerate}
 \item  $N \subseteq \Nodes$ is a countable set of nodes; 
\item $\tuple{N,\mathord{<}}$ is a strict poset with a single minimal element such that finitely many nodes appear at every finite distance from the root;
\item $\lambda \colon N \to \lab$ is a labelling function such that $x$ has no successors whenever $\lambda(x) = \bot$;
\item $\varphi \colon N \to \Form$ is a formula function such that: $\varphi(y) \Rightarrow \varphi(x)$, for all $x < y$ and for all $x\in N$, $\sat(\varphi(x))$ and $y < x$ for all $y \in \free(\varphi(x))$.
\end{enumerate}
For some $\alpha \in \lpo$, we will often use $N_\alpha$, $<_\alpha$, $\lambda_\alpha$, and $\varphi_\alpha$ to refer to its parts.
\end{definition}

The first three components are standard in pomset semantics. The order denotes \emph{causality}, so $x < y$ means that the action $\lambda(x)$ must be scheduled before $\lambda(y)$. The order is partial (not total) because actions that occur in parallel are not related. Formulae are a new addition to this structure, allowing us to encode guarded branching in the structure too. For any node $x\in N$, $\varphi(x)$ is a satisfiable formula comprised of nodes appearing earlier in the trace, since the corresponding tests must be resolved before executing later actions. In addition, formulae can only become stronger as the trace goes on, since dependencies on more and more tests are accumulated over time.

Two LPOFs are isomorphic, denoted $\alpha \equiv \beta$ iff there is a bijection $f \colon N_\alpha \to N_\beta$ such that $x <_\alpha y$ iff $f(x) <_\beta f(y)$, $\lambda_\alpha = \lambda_\beta \circ f$, and $\varphi_\alpha = f^{-1} \circ \varphi_\beta \circ f$, where $f^{-1}(\psi)$ renames the variables of $\psi$ in the obvious way. Given that, a pomset with formulae is an equivalence class of LPOFs.

\begin{definition}[Pomsets with Formulae]
We denote the isomorphism class of $\alpha$ as $[\alpha] \triangleq \{ \beta \in \lpo \mid \alpha\equiv\beta \}$.
A pomset with formulae (or, simply, a pomset) $\Alpha \in \pom$ is an isomorphism class of LPOFs.
\[
  \pom \triangleq \{ [\alpha] \mid \alpha \in \lpo \}
\]
\end{definition}
Pomsets with formulae have many nice domain-theoretic properties. For example, they have a DCPO (directed complete partial order) structure, meaning that Scott continuous operations over them have least fixed points. Roughly speaking $\Alpha \lepom \Beta$ iff $\Beta$ is obtained by replacing nodes labelled $\bot$ in $\Alpha$ with a new, larger structure. We will see a concrete example of this shortly when we discuss the semantics of loops. Infinite pomsets with formulae can also be represented as the supremum of their finite approximation, and monotone operations on those finite approximations can be extended to continuous operations on infinite pomsets. This is useful, since it is not possible to define operations inductively on infinite structures, as we will see in \Cref{sec:linearization}.

\begin{figure}
\begin{align*}
  \de{\skp} &\triangleq \singleton\bullet
  \\
  \de{
    {C_1 \fatsemi C_2}
  } & \triangleq \de{C_1} \fatsemi \de{C_2}
  \\
  \de{C_1 \parallel C_2} & \triangleq \de{C_1} \parallel \de{C_2}
  \\
  \de{\iftf b{C_1}{C_2}} &\triangleq \textsf{guard}(b, \de{C_1}, \de{C_2})
  \\
  \de{\whl bC} &\triangleq \mathsf{lfp}\left(\Phi_{\tuple{b, C}}\right)
  \quad\text{where}\   \Phi_{\tuple{b, C}}(\Alpha) \triangleq \guard\left(b, \de{C}\fatsemi \Alpha, \singleton{\bullet}\right)
  \\
  \de{a} &\triangleq \singleton{a}
\end{align*}
\caption{Trace semantics for commands $\de{-} \colon \mathsf{Cmd} \to \pom$.}
\label{fig:trace-semantics}
\end{figure}

The trace semantics $\de{-} \colon \mathsf{Cmd} \to \pom$, shown in \Cref{fig:trace-semantics}, interprets every command as a pomset with formulae. We use several operators to construct pomsets. First, $\singleton- \colon \lab \to \pom$ constructs a singleton pomset with the given label. This is used for the semantics of $\skp$ and actions $a \in \act$. Next, we have three combinators for combining pomsets, which we demonstrate pictorially below. The first is sequential composition $\Alpha\fatsemi\Beta$, which constructs a pomset where all the actions from $\Alpha$ occur before those in $\Beta$. In the diagram, $a_1\to a_2$ indicates causality ($a_1$ occurs before $a_2)$

 The second combinator is guarded choice $\guard(b, \Alpha, \Beta)$, where the test $b$ becomes the new root and all formulae of the nodes in $\Alpha$ and $\Beta$ are updated to indicate that the new root must pass or fail, respectively, as indicated by the arrows labelled $\mathsf{T}$ and $\mathsf{F}$. Finally, parallel composition $\Alpha \parallel \Beta$ joins the two pomsets with a new no-op root, so that the actions of $\Alpha$ and $\Beta$ can be interleaved without restriction.
\begin{mathpar}
\singleton{a_1} \fatsemi \singleton{a_2} = 
\vcenter{\vbox{\xymatrix@R=6pt@C=3pt{
  a_2 \\ a_1 \ar[u]
}}}

\guard(b, \singleton{a_1}, \singleton{a_2}) = 
\vcenter{\vbox{\xymatrix@R=6pt@C=3pt{
  a_1 && a_2
  \\
  & b \guardarrows
}}}

\singleton{a_1} \parallel \singleton{a_2} = 
\vcenter{\vbox{\xymatrix@R=6pt@C=3pt{
  a_1 && a_2 \\& \fork \ar[ul]\ar[ur]
}}}
\end{mathpar}
The semantics of loops is given by a least fixed point of the characteristic function $\Phi_{\tuple{b, C}}$, which is Scott continuous since it is defined using $\guard$ and $\fatsemi$, both of which are Scott continuous \cite{zilberstein2025denotational}.
 This means that the least fixed point is equal to the supremum of the finite unrollings of $\Phi_{\tuple{b, C}}$, as shown in the following diagram.
\[
\arraycolsep=.25em
\begin{array}{cccccccc}
\bot
&\lepom&
\vcenter{\vbox{\xymatrix@R=3pt@C=4pt{
  \bot \\
  a \ar[u] && \fork
 \\
 & b \guardarrows
}}}
&\lepom&
\vcenter{\vbox{\xymatrix@R=3pt@C=4pt{
  \bot \\
  a \ar[u] && \fork
 \\
 & b \guardarrows
 \\
 &a \ar[u] && \fork
 \\
 && b \guardarrows
}}}
& \lepom\cdots\lepom&
\vcenter{\vbox{\xymatrix@R=3pt@C=4pt{
  \vdots \\
  a \ar[u] && \fork
 \\
 & b \guardarrows
 \\
 &a \ar[u] && \fork
 \\
 && b \guardarrows
}}}
\\
\Phi_{\tuple{a, b}}^0(\singleton\bot)
&&
\Phi_{\tuple{a, b}}^1(\singleton\bot)
&&
\Phi_{\tuple{a, b}}^2(\singleton\bot)
&&
\sup_{n\in\mathbb N}\Phi_{\tuple{a, b}}^n(\singleton\bot)
\end{array}
\]
In each unrolling, $\bot$ appears further and further from the root until, in the supremum, it is pushed infinitely far away. That infinite spine represents the execution in which the guard $b$ never becomes false. In the context of probabilistic programs, it is important to include the infinite spine, as the probability of continuing to execute may only converge to 0 in the limit.

\subsection{Linearization}
\label{sec:linearization}

The benefit of the pomset model is that it is fully compositional; the parallel composition of two commands is interpreted as a straightforward combination of their syntactic structures. However, for the purposes of designing a program logic, we need a \emph{state transformer} model, which maps each input state to a convex set of distributions over output states. For this purpose, \citet{zilberstein2025denotational} also define linearization $\lin \colon \pom \to \mem{S} \to \C(\mem{S})$. The definition is repeated below.
\begin{align*}
  \nextt(\alpha, \psi, S) &= \left\{ x\in N_\alpha  \setminus S\mid \downarrow{x} \subseteq S, \psi \Rightarrow \varphi_\alpha(x) \right\}
  \\
  \nextt^\star(\alpha, \psi, S) &= \left\{ x\in N_\alpha  \setminus S \mid  \sat(\psi \land \varphi_\alpha(x)) \right\}
  \\
  \linnode^\I(\alpha, \psi, S, x)(\sigma) &= \left\{
    \begin{array}{ll}
      \linlpo^\I(\alpha, \psi, S\cup\{x\})^\dagger\left( \de{\lambda_\alpha(x)}^\I_\act(\sigma) \right) & \text{if} ~ \lambda_\alpha(x) \in \act
      \\
      \linlpo^\I(\alpha, \psi \land \sem{x = \de{\lambda_\alpha(x)}_\test(\sigma)}, S\cup\{x\})(\sigma) & \text{if} ~ \lambda_\alpha(x) \in \test
      \\
      \bot_\C & \text{if}~ \lambda_\alpha(x) = \bot
      \\
      \linlpo^\I(\alpha, \psi, S\cup\{x\})(\sigma) & \text{if} ~ \lambda_\alpha(x) =\fork
    \end{array}
  \right.
  \\
  \linlpo^\I(\alpha, \psi, S)(\sigma) &= \left\{
    \begin{array}{ll}
      \eta(\sigma) & \text{if}~ \nextt(\alpha, \psi, S) = \emptyset \\
      \bignd_{x \in \nextt(\alpha, \psi, S)} \linnode^\I(\alpha, \psi, S, x)(\sigma) & \text{if}~ \nextt(\alpha, \psi, S) \neq \emptyset 
    \end{array}
  \right.
  \\
  \linfin^\I([\alpha]) &= \linlpo^\I(\alpha, \tru, \emptyset)
  \\
  \lin^\I(\Alpha) &= \sup_{\Alpha'\ll\Alpha} \linfin^\I(\Alpha')
\end{align*}
The function $\nextt(\alpha, \psi, S)$ gives the set of nodes that are ready to be scheduled, given a finite LPOF $\alpha$, a path condition $\psi$, and the set of nodes $S$ that have already been processed. Linearizing a node  $\linnode$ does case analysis on the label of the node to decide how to proceed. If the node is an action, then the action is evaluated and composed with the linearization of the remainder of the LPOF. If the node is a test, then the result of the test is added to the path condition. If the node is $\bot$, then the execution is halted. If it is a no-op, then the execution simply continues.

To linearize an entire finite LPOF, we simply perform a nondeterministic choice over all the next nodes, linearize those nodes, and then proceed processing the rest of the structure recursively. To make the recursion well-founded, the structure must be finite, but we extend to the infinite case later. To linearize a finite pomset $\Alpha$, we simply linearize any representative LPOF $\alpha \in \Alpha$. Finally, to linearize an infinite pomset, we take the supremum over the linearization of all the finite approximations, where $\Alpha' \ll \Alpha$ indicates that $\Alpha'$ is a finite approximation of $\Alpha$.

Since most operations are continuous, we get, \eg that the supremum of finite approximations of a sequential composition is equal to the sequential composition over the supremum of finite approximations:
\[
  \sup_{\Alpha'\ll\Alpha}\sup_{\Beta'\ll\Beta} \Alpha' \fatsemi \Beta'
  = 
  \left(\sup_{\Alpha'\ll\Alpha} \Alpha'\right) \fatsemi  \left( \sup_{\Beta'\ll\Beta} \Beta'\right)
  =
  \Alpha \fatsemi \Beta
\]
The exception to this is for parallel composition, which is not continuous for technical reasons. However, if we define $\ll_1$ as follows:
\[
\GGamma \ll_1 \GGamma'
\qquad \text{iff} \qquad
\GGamma \ll \GGamma'
\quad\text{and}\quad
\mathsf{root}(\GGamma) = \mathsf{root}(\GGamma') = \fork 
\quad\text{or}\quad
\mathsf{root}(\GGamma') \neq \fork
\]
where $\mathsf{root}(\GGamma)$ is the label of the root node of the pomset $\GGamma$, then we get the following weaker pseudo-continuity property:
\[
  \sup_{\Alpha'\ll_1\Alpha}\sup_{\Beta'\ll_1\Beta} \Alpha' \parallel \Beta'
  = 
  \left(\sup_{\Alpha'\ll_1\Alpha} \Alpha'\right) \parallel  \left( \sup_{\Beta'\ll_1\Beta} \Beta'\right)
  =
  \Alpha \parallel \Beta
\]

\section{General Lemmas}
 
\subsection{Measure Theory Lemmas}

We begin by proving some general properties for characterizing discrete product spaces. By Lemma C.1 and C.2 of \citet{bao2025bluebell}, for any discrete probability space $\P$, there exists a countable partition $\{ A_i \mid i\in I \}$ of $\Omega_\P$ such that $\F_\P = \{ \biguplus_{i \in I'} A_i \mid I' \subseteq I \}$. Let $\ev(\P)$ denote this partition. It follows that for every $A \in \F_\P$, $\mu_\P(A) = \sum_{B \in S} \mu_\P(B)$ for some $S \subseteq \ev(\P)$.

\begin{lemma}\label{lem:prod-event-space}
For any probability spaces $\P_1$ and $\P_2$
\[
  \F_{\P_1\otimes\P_2} =
  \left\{ \biguplus_{(A, B) \in S} A\sep B ~\middle|~ S \subseteq \ev(\P_1)\times \ev(\P_2) \right\}
\]
\end{lemma}
\begin{proof}
The forward inclusion follows immediately from Lemma C.5 of \citet{bao2025bluebell}, so it suffices to only show the reverse inclusion. Take any $S \subseteq \ev(\P_1)\times\ev(\P_2)$.
For each $(A,B)\in S \subseteq \ev(\P_1)\times\ev(\P_2)$, clearly $A \in \ev(\P_1)$ and $B\in\ev(\P_2)$, so $A\sep B \in \F_{\P_1\otimes\P_2}$ since by definition $\F_{\P_1\otimes\P_2}$ is the smallest sigma algebra containing $\{ A\sep B \mid A\in \F_{\P_1}, B\in\F_{\P_2} \}$. Since $\F_{\P_1\otimes\P_2}$ is closed under countable unions, then $\biguplus_{(A,B) \in S} A\sep B \in \F_{\P_1\otimes\P_2}$.
\end{proof}

\begin{lemma}\label{lem:prod-equiv}
$\Q = \P_1 \otimes \P_2$ iff:
\begin{enumerate}
\item $\Omega_\Q = \Omega_{\P_1} \sep \Omega_{\P_2}$
\item $\F_\Q = \{ \biguplus_{(A, B) \in S} A\sep B \mid S \subseteq \ev(\P_1)\times \ev(\P_2) \}$
\item $\mu_\Q(A \sep B) = \mu_{\P_1}(A) \cdot \mu_{\P_2}(B)$ for all $A \in \ev(\P_1)$ and $B\in\ev(\P_2)$
\end{enumerate}
\end{lemma}
\begin{proof}
The forward direction follows immediately from the definition of product spaces and \Cref{lem:prod-event-space}. We now show the reverse direction. By (1) and the definition of product spaces, $\Omega_\Q = \Omega_{\P_1}\sep \Omega_{\P_2} = \Omega_{\P_1 \otimes \P_2}$. By (2) and \Cref{lem:prod-event-space}, $\F_\Q = \F_{\P_1 \otimes \P_2}$. Now, take any event $A \in \F_\Q = \F_{\P_1 \otimes \P_2}$. We already know that $A = \biguplus_{(A_1, A_2) \in S}A_1\sep A_2$ for some $S \subseteq \ev(\P_1)\times\ev(\P_2)$, so:
\begin{align*}
  \mu_\Q(A)
  &= \mu_\Q\left( \biguplus_{(A_1, A_2) \in S}A_1\sep A_2 \right)
  \\
  &= \sum_{(A_1, A_2) \in S} \mu_\Q(A_1 \sep A_2)
  \\
  &= \sum_{(A_1, A_2) \in S} \mu_{\P_1}(A_1)\cdot \mu_{\P_2}(A_2)
  \\
  &= \sum_{(A_1, A_2) \in S} \mu_{\P_1\otimes\P_2}(A_1\sep A_2)
  \\
  &= \mu_{\P_1\otimes\P_2}\left( \biguplus_{(A_1, A_2) \in S} A_1\sep A_2 \right)
  \ =\ \mu_{\P_1\otimes\P_2}(A)
\end{align*}

\end{proof}
 
\begin{lemma}\label{lem:otimes-oplus-dist}
If both sides of the following equality are well-defined, then:
\[
  \left(\bigoplus_{i\sim \nu} \P_i\right) \otimes \Q = \bigoplus_{i\sim\nu} (\P_i\otimes \Q)
\]
\end{lemma}
\begin{proof}
Let $\P_i = \tuple{ \Omega_i, \F_i, \mu_i}$ for each $i\in I = \supp(\nu)$. If the above are well defined, then there must be $S$ and $T$ such that $\Omega_i \subseteq \mem S$ and $\Omega_\Q \subseteq \mem T$.
We first show that both probability spaces have the same sample space, since:
\begin{align*}
  \Omega_{(\bigoplus_{i\sim \nu} \P_i) \otimes \Q}
  & = ( \smashoperator{\bigcup_{i\in\supp(\nu)}} \Omega_i ) \sep \Omega_\Q
  \\
  &= \left\{ \sigma \uplus \sigma' ~\middle|~ \sigma \in \smashoperator{\bigcup_{i\in I}} \Omega_iv, \sigma' \in \Omega_\Q \right\}
  \\
  &= \smashoperator{\bigcup_{i\in I}} \{ \sigma \uplus \sigma' \mid \sigma \in \Omega_i, \sigma'\in \Omega_\Q \}
  \\
  &= \smashoperator{\bigcup_{i\in I}} \Omega_i\sep\Omega_\Q
  \ = \ \Omega_{\bigoplus_{i\sim\nu}(\P_i \otimes \Q)}
\end{align*}
To show that the $\sigma$-algebras are the same, it will suffice to show that they are generated from the same disjoint partitions \cite[Lemma C.1]{bao2025bluebell}. We start by using \Cref{lem:prod-event-space} to conclude that:
\begin{align*}
  \ev\left(\left(\bigoplus_{i\sim \nu} \P_i\right) \otimes \Q\right)
  &= \ev\left(\bigoplus_{i\sim\nu} \P_i\right) \times\ev(\Q)
  \intertext{Since all the $\P_i$ have disjoint sample spaces, their partitions must also be disjoint:}
  &= \left(\biguplus_{i\in I} \ev(\P_i)\right) \times\ev(\Q)
  \\
  &= \biguplus_{i\in I} \ev(\P_i) \times\ev(\Q)
  \\
  &= \biguplus_{i\in I} \ev(\P_i \otimes \Q)
  = \ev\left( \bigoplus_{i\sim\nu} \P_i \otimes \Q \right)
\end{align*}
Finally, by \Cref{lem:prod-equiv}, it suffices to show that the product measures agree on events of the form $A\sep B$, where $A \in \F_{\bigoplus_{i\sim\nu}}\P_i$ and $B\in\F_\Q$:
\begin{align*}
  \mu_{(\bigoplus_{i\sim \nu} \P_i) \otimes \Q}(A \sep B)
  &= \mu_{\bigoplus_{i\sim \nu} \P_i}(A) \cdot \mu_\Q(B)
  \\
  &= \left(\sum_{i\in I} \nu(i)\cdot \mu_i(A \cap \Omega_i) \right) \cdot \mu_\Q(B)
  \\
  &= \sum_{i\in I} \nu(i)\cdot \mu_i(A \cap \Omega_i) \cdot \mu_\Q(B)
  \\
  &= \sum_{i\in I} \nu(i)\cdot \mu_{\P_i\otimes\Q}((A \cap \Omega_i)\sep B)
  \\
  &= \sum_{i\in I} \nu(i)\cdot \mu_{\P_i\otimes\Q}((A \sep B)  \cap (\Omega_i \sep \Omega_\Q))
  \\
  &= \mu_{\bigoplus_{i\sim\nu}\P_i\otimes\Q}(A \sep B)
\end{align*}
\end{proof}

\begin{lemma}[Monotonicity of $\otimes$]\label{lem:otimes-mono}
If $\P \preceq \P'$ and $\Q \preceq \Q'$, then $\P\otimes \Q \preceq \P'\otimes\Q'$.
\end{lemma}
\begin{proof}
Let $S$, $S'$, $T$, and $T'$ be sets such that $\Omega_{\P} \subseteq S$, $\Omega_{\P'}\subseteq \mem{S'}$, $\Omega_{\Q}\subseteq\mem T$, and $\Omega_{\Q'} \subseteq \mem{T'}$.
First we show the condition on sample spaces:
\begin{align*}
  \Omega_{\P \otimes\Q}
  &= \Omega_{\P} \sep \Omega_{\Q}
  \\
  &= \{ \sigma\uplus \tau \mid \sigma\in \Omega_{\P}, \tau\in\Omega_{\Q} \}
  \\
  &\subseteq \{ \sigma\uplus \tau \mid \sigma\in \pi_S(\Omega_{\P'}), \tau\in \pi_T(\Omega_{\Q'}) \}
  \\
  &= \{ \pi_S(\sigma)\uplus \pi_T(\tau) \mid \sigma\in \Omega_{\P'}, \tau\in \Omega_{\Q'} \}
  \\
  &= \{ \pi_{S\cup T}(\sigma \uplus\tau) \mid \sigma\in \Omega_{\P'}, \tau\in \Omega_{\Q'} \}
  = \pi_{S\cup T}(\Omega_{\P'\otimes\Q})
\end{align*}
Now we show the condition on $\sigma$-algebras:
\begin{align*}
  \F_{\P \otimes \Q}
  &= \sigma(\{ A\sep B \mid A \in \F_\P, B \in \F_\Q \})
  \intertext{Since $\sigma$-closure is monotonic:}
  &\subseteq \sigma(\{ A\sep B \mid A \in \pi_S(\F_{\P'}), B \in \pi_T(\F_{\Q'}) \})
  \\
  &= \sigma(\{ \pi_S(A)\sep \pi_T(B) \mid A \in \F_{\P'}, B \in \F_{\Q'} \})
  \\
  &= \sigma(\pi_{S\cup T}(\{ A\sep B \mid A \in \F_{\P'}, B \in \F_{\Q'} \}))
  \\
  &= \pi_{S\cup T}(\sigma(\{ A\sep B \mid A \in \F_{\P'}, B \in \F_{\Q'} \}))
  \\
  &= \pi_{S\cup T}(\F_{\P' \otimes \Q'})
\end{align*}
Finally, by \Cref{lem:prod-equiv}, it suffices to show that the product measure is equal only for events of the form $A\sep B$ where $A \in \F_\P$ and $B\in \F_\Q$:
\begin{align*}
  \mu_{\P\otimes \Q}(A\sep B)
  &= \mu_\P(A) \cdot \mu_\Q(B)
  \intertext{Since $\P \preceq\P'$ and $\Q\preceq\Q'$:}
  &= \mu_{\P'}\left(\bigcup \{A' \in \F_{\P'} \mid \pi_S(A') =A \}\right) \cdot
        \mu_{\Q'}\left(\bigcup \{ B' \in \F_{\Q'} \mid \pi_T(B') = B \} \right)
  \\
  &= \mu_{\P'\otimes\Q'}\left(
    \left(\bigcup \{A' \in \F_{\P'} \mid \pi_S(A') =A \}\right) \sep
    \left( \bigcup \{ B' \in \F_{\Q'} \mid \pi_T(B') = B \} \right) \right)
  \intertext{Let $A'' = \bigcup \{A' \in \F_{\P'} \mid \pi_S(A') =A \}$ and $B'' = \bigcup \{ B' \in \F_{\Q'} \mid \pi_T(B') = B \}$. Clearly $A'' \in \F_{\P'}$ and $B''\in\F_{\Q'}$ since both sets are closed under countable unions. Therefore:}
  &= \mu_{\P' \otimes \Q'}\left( A'' \sep B'' \right)
  \intertext{We show that the sets above and below are equal by showing the inclusion in both directions. The forward inclusion is trivial; $A''\sep B''$ is clearly an element of $\F_{\P'\otimes \Q'}$ and has the property $\pi_{S\cup T}(A''\sep B'') = \pi_S(A'') \sep \pi_T(B'') = A\sep B$.
  For the reverse inclusion, take some element of the set below, which must have the form $\sigma\uplus\tau$ where $\sigma\in\Omega_{\P'}$ and $\tau\in\Omega_{\Q'}$ and $\sigma\uplus\tau\in E$ for some $E \in \F_{\P'\otimes\Q'}$ such that $\pi_{S\cup T}(E) = A\sep B$. By \Cref{lem:prod-event-space}, $E = \biguplus_{(E_1,E_2)\in S} E_1 \sep E_2$ for some $S \subseteq \ev(\P')\times \ev(\Q')$, so $\sigma\in E_1$ and $\tau \in E_2$ for some $(E_1,E_2)\in S$. Since $\pi_{S\cup T}(E) = A\sep B$, then $\pi_S(E) = \biguplus_{(E_1, -)\in S} E_1 = A$ and $\pi_T(E) = \biguplus_{(-,E_2)\in S} E_2 = B$, therefore $\sigma \in E_1 \subseteq \pi_S(E) \subseteq A''$ and $\tau\in E_2\subseteq \pi_T(E)\subseteq B''$, so $\sigma\uplus\tau \in A''\sep B''$.}
  &= \mu_{\P' \otimes \Q'}\left( \bigcup \{E \in \F_{\P'\otimes\Q'} \mid \pi_{S\cup T}(E) = A\sep B \}\right)
\end{align*}
\end{proof}

\begin{lemma}[Monotonicity of $\oplus$]\label{lem:oplus-mono}
If $\P_v \preceq \P'_v$ for each $v \in \supp(\nu)$, then $\bigoplus_{v\sim\nu} \P_v \preceq \bigoplus_{v\sim\nu}\P'_v$.
\end{lemma}
\begin{proof}
Let $\P = \bigoplus_{v\sim\nu} \P_v$ and $\P' = \bigoplus_{v\sim\nu}\P'_v$, so we need to prove that $\P \preceq \P'$. Also, let $S$ and $S'$ be sets such that $\Omega_{\P_v} \subseteq \mem S$ and $\Omega_{\P'_v} \subseteq \mem{S'}$ for all $v$. We first establish the required property on the sample space:
\[
  \Omega_\P
  = \smashoperator{\bigcup_{v \in \supp(\nu)}} \Omega_{\P_v}
  \subseteq \smashoperator{\bigcup_{v \in \supp(\nu)}} \pi_{S}(\Omega_{\P'_v})
  = \pi_S(\Omega_{\P'})
\]
Now, we establish the property on $\sigma$-algebras.
\begin{align*}
  \F_\P
  &= \{ A \mid A\subseteq \Omega_\P, \forall v.\ A\cap \Omega_{\P_v} \in \F_{\P_v} \}
  \\
  &\subseteq \{ A \mid A\subseteq \pi_S(\Omega_{\P'}), \forall v.\ A\cap \pi_S(\Omega_{\P'_v}) \in \{ \pi_S(B) \mid B \in \F_{\P'_v}\} \}
  \\
  &= \{ \pi_S(A) \mid A\subseteq \Omega_{\P'}, \forall v.\ A\cap \Omega_{\P'_v} \in \F_{\P'_v} \}
  \\
  &= \{ \pi_S(A) \mid A \in \F_{\P'} \}
\end{align*}
Finally, we establish the condition on probability measures:
\begin{align*}
  \mu_{\P}(A)
  &= \smashoperator{\sum_{v\in\supp(\nu)}} \nu(v)\cdot \mu_{\P_v}(A \cap \Omega_{\P_v})
  \\
  &= \smashoperator{\sum_{v\in\supp(\nu)}} \nu(v)\cdot \mu_{\P'_v}\left( \bigcup \{B \in \F_{\P'_v} \mid \pi_S(B) =A \cap \Omega_{\P_v} \} \right)
  \intertext{Since $A\subseteq \Omega_{\P} \subseteq \pi_S(\Omega_{\P'})$, then, we can move the intersection out of the set limits:}
  &= \smashoperator{\sum_{v\in\supp(\nu)}} \nu(v)\cdot \mu_{\P'_v}\left( \left( \bigcup \{B \in \F_{\P'} \mid \pi_S(B) =A \} \right) \cap \Omega_{\P'_v} \right)
  \\
  &= \mu_{\P'}\left( \bigcup \{B \in\F_{\P'} \mid \pi_S(B) =A \} \right)
\end{align*}
\end{proof}

\begin{lemma}\label{lem:sum-convex-mono}
  For any complete probability spaces $\P$, $(\P_i)_{i\in I}$, and $\mu \in \D(I)$ such that $\Omega_{\P} \subseteq \mem S$, each $\Omega_{\P_i} \subseteq \mem T$, and $\bigcup_{i\in I} \Omega_{\P_i} = \mem T$,
  if $\P \preceq \comp(\P_i)$ for all $i\in \supp(\mu)$, then $\P \preceq \bigoplus_{i\sim \mu} \P_i$
\end{lemma}
\begin{proof}
Let $\Q = \bigoplus_{i\sim \mu} \P_i$.
Since $\P \preceq \P_i$, then it must be that $S \subseteq T$.
The property on sample spaces is simple:
\[
  \Omega_\P
  = \smashoperator{\bigcup_{i \in \supp(\mu)}} \Omega_{\P}
  \subseteq \smashoperator{\bigcup_{i\in\supp(\mu)}} \pi_S(\Omega_{\comp(\P_i)})
  = \smashoperator{\bigcup_{i\in\supp(\mu)}} \pi_S(\mem T)
  = \pi_S\left( \smashoperator[r]{\bigcup_{i\in\supp(\mu)}} \mem T \right)
  = \pi_S( \Omega_\Q)
\]
Next, we verify the property on $\sigma$-algebras.
\begin{align*}
  \F_\P
  &= \bigcap_{i\in I} \F_\P
  \\
  &\subseteq \bigcap_{i\in I} \{ \pi_S(A) \mid A \in \F_{\comp(\P_i)} \}
  \\
  &= \{ \pi_S(A) \mid  \forall i\in I.\ A \in \F_{\comp(\P_i)} \}
  \intertext{Note that the completion adds information about events outside of $\Omega_{\P_i}$, so if $A$ is in all of the $\F_{\comp(\P_i)}$ sets, then its projection into each $\Omega_{\P_i}$ must be in $\F_{\P_i}$.}
  &= \{ \pi_S(A) \mid A \subseteq \Omega_Q, \forall i.\ A \cap \Omega_{\P_i} \in \F_{\P_i} \}
  \\
  &= \{ \pi_S(A) \mid A \in \F_\Q \}
\end{align*}
Finally, we show the property on probability measures.
\begin{align*}
  \mu_\P(A)
  &= \smashoperator{\sum_{i\in \supp(\mu)}} \mu(i) \cdot \mu_\P(A)
  \\
  &= \smashoperator{\sum_{i\in \supp(\mu)}} \mu(i) \cdot \mu_{\comp(\P_i)}\left( \bigcup \{B \in \F_{\comp(\P_i)} \mid \pi_S(B) = A \} \right)
  \intertext{The completion assigns zero probability to events outside of $\Omega_{\P_i}$, so removing the completion and projecting into $\Omega_{\P_i}$ will yield the same value.}
  &= \smashoperator{\sum_{i\in \supp(\mu)}} \mu(i) \cdot \mu_{\P_i}\left( \left(\bigcup \{ B \in \F_{\Q} \mid \pi_S(B) = A \}\right) \cap \Omega_{\P_i} \right)
  \\
  &= \mu_\Q\left( \bigcup \{ B\in\F_\Q \mid \pi_S(B) = A \} \right)
\end{align*}
\end{proof}

\begin{lemma}\label{lem:oplus-mono2}
  If $\P \preceq \comp(\P_i)$ for all $i \in \supp(\mu)$, then $\P \preceq \bigoplus_{i\sim \mu} \P_i$.
\end{lemma}
\begin{proof}
Let $S$ be the set such that $\Omega_\P = \mem S$, $T$ be such that $\mem T = \bigcup_{i\in\supp(\mu)} \Omega_{\P_i}$, and $\Q = \preceq \bigoplus_{i\sim \mu} \P_i$.
The condition on sample spaces is simple, since $\Omega_\P = \mem S$ and $\Omega_\Q = \mem T$, and clearly $\mem S = \pi_S(\mem T)$. For the condition on $\sigma$-algebras, we have:
\begin{align*}
  \F_\P
  &= \smashoperator{\bigcap_{i\in\supp(\mu)}} \F_\P
  \\
  &\subseteq \smashoperator{\bigcap_{i\in\supp(\mu)}} \{ \pi_S(A) \mid A \in \F_{\comp(\P_i)} \}
  \\
  &= \smashoperator{\bigcap_{i\in\supp(\mu)}} \{ \pi_S(A \cup B) \mid A \in \F_{\P_i}, B \subseteq \mem T\setminus \Omega_{\P_i} \}  
  \intertext{Since for each $i$ in the intersection above, we can measure every sample outside of $\Omega_{\P_i}$, then after taking the intersection we are only able to measure samples from $\Omega_{\P_i}$ according to the information provided by $\F_{\P_i}$, which provides the \emph{least} information about those samples.}
  &= \{ \pi_S(A) \mid A \subseteq \mem T, \forall i.\ A \cap \Omega_{\P_i} \in \F_{\P_i} \}
  \\
  &= \{ \pi_S(A) \mid A \in \F_\Q \}
\end{align*}
Now, we show the condition on probability measures:
\begin{align*}
  \mu_\P(A)
  &= \smashoperator{\sum_{i\in\supp(\mu)}} \mu(i)\cdot \mu_\P(A)
  \\
  &= \smashoperator{\sum_{i\in\supp(\mu)}} \mu(i)\cdot \mu_{\comp(\P_i)}\left( \bigcup \{B \in \F_{\comp(\P_i)} \mid \pi_S(B) = A \} \right)
  \\
  &= \smashoperator{\sum_{i\in\supp(\mu)}} \mu(i)\cdot \mu_{\P_i}\left( \left(\bigcup \{B \in \F_\Q \mid \pi_S(B) = A\}\right) \cap \Omega_{\P_i}\right)
 \\
 &= \mu_\Q\left(\bigcup \{B \in \F_\Q \mid \pi_S(B) = A\}\right)
\end{align*}

\end{proof}

\subsection{The Convex Powerset}

\begin{lemma}
\label{lem:bignd_dist}
For any finite index set $I$,  $f \colon X \to \C(Y)$, and $(S_i)_{i\in I} \in \C(X)^I$:
\[
f^\dagger(\bignd_{i\in I} S_i) = \bignd_{i\in I} f^\dagger(S_i)
\]
\end{lemma}
\begin{proof}
Let $g\colon I \to \C(X)$ be defined as $g(i) \triangleq S_i$. Now, observe that:
\[
  \bignd_{i\in I} S_i
  = \left\{ \sum_{i\in\supp(\mu)} \mid \mu\in \D(I), \forall i.\ \nu_i \in g(i) \right\}
  = g^\dagger(\D(I))
\]
So, we get:
\begin{align*}
  f^\dagger( \bignd_{i\in I} S_i )
  &= f^\dagger(g^\dagger(\D(I)))
  \\
  &= (f^\dagger \circ g)^\dagger(\D(I))
  \\
  &= \left\{ \smashoperator[r]{\sum_{i\in\supp(\mu)}} \mu(i) \cdot \nu_i \;\;\Big|\;\; \mu \in \D(I), \forall i.\ \nu_i \in f^\dagger(g(i)) \right\}
  \\
  &= \bignd_{i\in I} f^\dagger(S_i)
\end{align*}
\end{proof}

For any $S\in \C(\mem V)$ and $A \subseteq \mem V$, let:
\[
  \minProb(S, A) \triangleq \inf_{\mu\in S} \mu(A)
\]
Since $S\in \C(\mem V)$, it is a closed subset of $\D(\mem{V}_\bot)$, and so there must be a $\mu \in S$ such that $\mu(A) = \minProb(S, A)$, therefore $\minProb(S,A) = \min_{\mu\in S} \mu(A)$.

\begin{lemma}[Monotonicity of $\minProb$]\label{lem:minProb-mono}
 If $S \sqsubseteq_\C T$, then:
 \[
   \minProb(S, A) \le \minProb(T, A)
 \]
\end{lemma}
\begin{proof}
By definition $S \sqsubseteq_\C T$ iff $S \supseteq T$. So, we get:
 \[
   \minProb(S, A)
   = \inf_{\mu \in S} \mu(A)
   = \min(\inf_{\mu \in T} \mu(A), \inf_{\mu \in S\setminus T} \mu(A))
   \le \inf_{\mu \in T} \mu(A)
   = \minProb(T, A)
 \]
\end{proof}

\begin{lemma}\label{lem:minProb-lb}
  For any directed set $D \subseteq \C(\mem V)$:
  \[
    \sup_{S\in D} \minProb(S, A) \le \minProb(\sup D, A)
  \]
\end{lemma}
\begin{proof}
Since $S \sqsubseteq_\C \sup D$ for all $S \in D$, then by \Cref{lem:minProb-mono} we know that $\minProb(S, A) \le \minProb(\sup D, A)$. Therefore, since the supremum is the least upper bound, it must be that $\sup_{S\in D} \minProb(S, A) \le \minProb(\sup D, A)$.
\end{proof}

\begin{lemma}[Scott Continuity of $\minProb$]\label{lem:minProb-cont}
  For any directed set $D \subseteq \C(\mem S)$ such that $\mu(A) = p$ for all $\mu \in \sup D$:
  \[
    \sup_{S\in D} \minProb(S, A) = \minProb(\sup D, A)
  \]
\end{lemma}
\begin{proof}
By \citet[Corollary 3]{markowsky1976chain-complete}, we know that any chain-continuous function on a DCPO is Scott continuous, therefore it will suffice to show that $\minProb$ is chain continuous.
Let $(S_\delta)_{\delta < \zeta}$ be a transfinite chain, so that $S_\delta \in \C(\mem V)$ for all $\delta < \zeta$, and $S_\delta \sqsubseteq_\C S_{\delta'}$ for all $\delta < \delta'$. We will now prove that:
\[
  \sup_{\delta < \zeta} \minProb(S_\delta, A) = \minProb\left(\sup_{\delta<\zeta} S_\delta, A\right)
\]
From \Cref{lem:minProb-lb}, we already know that $\sup_{S\in D} \minProb(S, A) \le \minProb(\sup D, A)$. We will now also show that $\sup_{S\in D} \minProb(S, A) < \minProb(\sup D, A)$ is a contradiction, and therefore the two quantities are equal.
Let $p = \minProb(\sup_{\delta <\zeta} S_\delta, A)$,
and suppose for the sake of contradiction that:
\[
  \sup_{\delta < \zeta} \minProb(S_\delta, A) < \minProb(\sup_{\delta <\zeta} S_\delta, A)
\]
So, there is some $\varepsilon > 0$ such that $\sup_{\delta < \zeta} \minProb(S_\delta, A) = p-\varepsilon$.

Let $T_\delta = \{ \mu \in S_\delta \mid \mu(A) \le p-\varepsilon \}$ for all $\delta <\zeta$. If there exists a $\delta$ such that $T_\delta = \emptyset$, then $\sup_{\delta'<\zeta} \minProb(S_{\delta'}, A) \ge \minProb(S_\delta, A) > p-\varepsilon$, which is a contradiction, therefore $T_\delta \neq\emptyset$ for all $\delta < \zeta$. Further, by Lemma B.4.3 of \citet{mciver2005abstraction}, the $T_\delta$ sets are closed, therefore their intersection must be nonempty by the finite intersection property. That means that there is some $\mu \in \bigcap_{\delta<\zeta} T_\delta \subseteq \bigcap_{\delta<\zeta} S_\delta = \sup_{\delta<\zeta} S_\delta$ such that $\mu(A) \le p - \varepsilon$, but this is a contradiction too since we know that $\minProb( \sup_{\delta<\zeta} S_\delta, A) = p$. Therefore, it cannot be that $\sup_{\delta < \zeta} \minProb(S_\delta, A) < \minProb(\sup_{\delta <\zeta} S_\delta, A)$, and instead $\sup_{\delta < \zeta} \minProb(S_\delta, A) = \minProb(\sup_{\delta <\zeta} S_\delta, A)$.

\end{proof}

\begin{lemma}\label{lem:minProb-nd}
\[
  \minProb(S \nd T, A) = \min(\minProb(S, A), \minProb(T, A))
\]
\end{lemma}
\begin{proof}
\begin{align*}
  \minProb(S \nd T, A)
  &= \minProb( \{ \mu \oplus_p \nu \mid \mu\in S, \nu\in T, p\in[0,1] \}, A)
  \\
  &= \inf_{\mu\in S}\inf_{\nu\in T} \inf_{p\in [0,1]} p\cdot\mu(A) + (1-p)\cdot \nu(A)
  \\
  &= \inf_{p\in [0,1]} p\cdot\left(\inf_{\mu\in S} \mu(A) \right) + (1-p)\cdot \left( \inf_{\nu\in T} \nu(A) \right)
  \\
  &= \inf_{p\in [0,1]} p\cdot\minProb(S, A) + (1-p)\cdot \minProb(T, A)
  \intertext{Now, there are three cases, if $\minProb(S, A) < \minProb(T, A)$, then the infimum occurs when $p=1$. If instead $\minProb(S, A) > \minProb(T, A)$, then the infimum occurs when $p=0$. If $\minProb(S, A) = \minProb(T, A)$, then the expression is equal for all $p$. So, the infimum always occurs at one of the extremes ($p=0$ or $p=1$), and therefore:}
  &= \min(\minProb(S, A), \minProb(T, A))
\end{align*}
\end{proof}

\begin{lemma}\label{lem:minProb-kleisli}
For any $f \colon \mem U \to \C(\mem U)$, $S \in \C(\mem U)$, and $A \subseteq \mem U$:
\[
  \minProb(f^\dagger(S), A)
  = \inf_{\mu \in S} \sum_{\sigma\in \supp(\mu) \cap \mem U} \mu(\sigma) \cdot \minProb(f(\sigma), A)
\]
\end{lemma}
\begin{proof}
\begin{align*}
\minProb(f^\dagger(S), A)
&= \minProb\left(
  \left\{
    \sum_{\sigma\in\supp(\mu)} \mu(\sigma) \cdot \nu_\sigma
    ~\middle|~
    \mu\in S,
    \forall \sigma\in\supp(\mu).\ \nu_\sigma\in f_\bot(\sigma)
  \right\},
  A
\right)
\\
&= \inf_{\mu\in S} \inf_{\nu_\sigma\in f_\bot(\sigma), \forall\sigma\in\supp(\mu)} \sum_{\sigma\in\supp(\mu)} \mu(\sigma) \cdot \nu_\sigma(A) 
\\
&= \inf_{\mu\in S} \sum_{\sigma\in\supp(\mu)} \mu(\sigma) \cdot \inf_{\nu_\sigma\in f_\bot(\sigma)}\nu_\sigma(A) 
\\
&= \inf_{\mu\in S} \sum_{\sigma\in\supp(\mu)} \mu(\sigma) \cdot \minProb( f_\bot(\sigma), A)
\\
&= \inf_{\mu\in S} \sum_{\sigma\in\supp(\mu) \cap \mem U} \mu(\sigma) \cdot \minProb( f_\bot(\sigma), A) + \mu(\bot)\cdot \minProb(f_\bot(\bot), A)
\\
&= \inf_{\mu\in S} \sum_{\sigma\in\supp(\mu) \cap \mem U} \mu(\sigma) \cdot \minProb( f(\sigma), A) + 0
\\
&= \inf_{\mu\in S} \sum_{\sigma\in\supp(\mu) \cap \mem U} \mu(\sigma) \cdot \minProb( f(\sigma), A)
\end{align*}
\end{proof}

\section{Concurrency Lemmas}

\subsection{Invariant Sensitive Execution}

\begin{lemma}[Check Monotonicity]\label{lem:check-mono}
For any $U,V,W \subseteq \mathsf{Var}$ such that $U\cap V = \emptyset$, $\I \subseteq \mem U$, $\J \subseteq \mem V$, and $U\cup V \subseteq W$:
\[
  \mathsf{check}^{\I\sep\J} \lec^\bullet \mathsf{check}^\I
\]
\end{lemma}
\begin{proof}
Take any $\sigma \in \mem W$. If $\pi_U(\sigma) \notin \I$ or $\pi_V(\sigma) \notin \J$, then we get:
\[
  \mathsf{check}^{\I\sep\J}(\sigma) = \bot_\C \lec \mathsf{check}^\I(\sigma)
\]
So, we are done. If not, then $\pi_U(\sigma) \in \I$ and $\pi_V(\sigma) \in \J$, so we get:
\[
  \mathsf{check}^{\I\sep\J}(\sigma) = \eta(\sigma) = \mathsf{check}^\I(\sigma)
\]
\end{proof}

\begin{lemma}[Action Monotonicity]\label{lem:act-mono}
For any $U,V,W \subseteq \mathsf{Var}$ such that $U\cap V = \emptyset$, $\I \subseteq \mem U$, $\J \subseteq \mem V$, and $U\cup V \subseteq W$:
\[
  \de{a}^{\I \sep \J}_\act \lec^\bullet \de{a}^\I
\]
\end{lemma}
\begin{proof}
Let $W' = W \setminus (U\cup V)$.
Take any $\sigma \in \mem W$. if $\pi_V(\sigma) \notin \J$ or $\pi_U(\sigma) \notin \I$, then:
\[
  \de{a}^{\I\sep\J}_\act(\sigma) = \bot_\C \lec \de{a}^\I(\sigma)
\]
So, we are done. If instead $\pi_V(\sigma) \in \J$ and $\pi_U(\sigma) \in \I$, then we get:
\begin{align*}
  \de{a}^{\I\sep\J}_\act(\sigma)
  &= (\mathsf{check}^{\I\sep\J})^\dagger\left( \de{a}_\act^\dagger \left( (\mathsf{replace}^{\I \sep \J})^\dagger(\mathsf{check}^{\I\sep\J}(\sigma) ) \right) \right)
  \\
  &= (\mathsf{check}^{\I\sep\J})^\dagger\left( \de{a}_\act^\dagger \left( \mathsf{replace}^{\I\sep\J}(\sigma) \right) \right)
  \\
  &= (\mathsf{check}^{\I\sep\J})^\dagger\left( \bignd_{\tau \in \I\sep \J} \de{a}_\act(\pi_{W'}(\sigma) \uplus \tau) \right)
  \intertext{Since $\pi_V(\sigma) \in \J$ and $\lec$ is $\supseteq$.}
  &\lec (\mathsf{check}^{\I\sep\J})^\dagger\left( \bignd_{\tau \in \I} \de{a}_\act(\pi_{W' \cup V}(\sigma) \uplus \tau) \right)
  \intertext{By \Cref{lem:check-mono} and monotonicity of Kleisli extension.}
  &\lec (\mathsf{check}^{\I})^\dagger\left( \bignd_{\tau \in \I} \de{a}_\act(\pi_{W' \cup V}(\sigma) \uplus \tau) \right)
  \\
  &= (\mathsf{check}^{\I})^\dagger\left( \de{a}_\act^\dagger\left( \mathsf{replace}^\I(\sigma) \right) \right)
  \intertext{Since $\pi_U(\sigma) \in \I$}
  &= (\mathsf{check}^{\I})^\dagger\left( \de{a}_\act^\dagger\left( (\mathsf{replace}^\I)^\dagger(\mathsf{check}^{\I}(\sigma)) \right) \right)
  \\
  &= \de{a}^\I_\act(\sigma)
\end{align*}
\end{proof}

\invMono*
\begin{proof}
For any pomset $\Alpha$, let $\Alpha^\I$ be the pomset obtained by replacing each action $a$ with the tuple $\tuple{a, \I}$. The evaluation function for these actions and order is as follows:
\[
  \de{\tuple{a, \I}}_\act(\sigma) \triangleq \de{a}^\I_\act(\sigma)
  \qquad
  \tuple{a, \I} \sqsubseteq_\act \tuple{a', \J}
  \quad\text{iff}\quad
  a=a'
  \quad\text{and}\quad
  \exists \J'.\ \I = \J \sep \J'
\]
So, clearly $\Alpha^{\I\sep\J} \lepom \Alpha^\I$, and $\lin^\I(\Alpha) = \lin(\Alpha^\I)$. Note that $\sqsubseteq_\act$ is not pointed or finitely proceeded, but this doesn't matter; as we will see in the following proof, we only need this order for monotonicity, not for the extension lemma \cite[Lemma 3.8]{zilberstein2025denotational}. 
We now complete the proof as follows:
\begin{align*}
  \lin^{\I\sep\J}(\Alpha)
  &= \sup_{\Alpha' \ll \Alpha} \linfin^{\I\sep\J}(\Alpha')
  \\
  &= \sup_{\Alpha' \ll \Alpha} \linfin(\Alpha'^{\I\sep\J})
  \intertext{By Lemma H.3 of \citet{zilberstein2025denotational} and \Cref{lem:act-mono}.}
  &\lec \sup_{\Alpha' \ll \Alpha} \linfin(\Alpha'^{\I})
  \\
  &= \sup_{\Alpha' \ll \Alpha} \linfin^\I(\Alpha')
  = \lin^\I(\Alpha)
\end{align*}
\end{proof}

\subsection{Parallel Composition}

\begin{lemma}\label{lem:minProb-linlpo}
For any pairwise disjoint $U_1,U_2,V \subseteq \mathsf{Var}$, $\alpha$ and $\beta$ such that $\free_\act(\alpha) \subseteq U_1 \cup V$, $\free_\test(\alpha)\subseteq U_1$, $\free_\act(\beta) \subseteq U_2\cup V$, and $\free_\test(\beta)\subseteq U_2$; $\I \subseteq \mem V$, $S \subseteq N_{\alpha\parallel_x\beta}$, $\psi_1,\psi_2 \in\Form$, $\sigma_1 \in \mem{U_1}$, $\sigma_2 \in \mem{U_2}$, $\tau \in \I$, $A \subseteq \mem{U_1}$, and $B\subseteq\mem{U_2}$:
\begin{align*}
  &\minProb(\linlpo^\I(\alpha\parallel_x\beta, \psi_1 \land \psi_2, S)(\sigma_1 \uplus \sigma_2 \uplus \tau), A\sep B \sep \I)
  \\
  &=
  \minProb(\linlpo^\I(\alpha, \psi_1, S \cap N_\alpha)(\sigma_1 \uplus \tau), A\sep \I) \cdot
    \minProb(\linlpo^\I(\beta, \psi_2, S \cap N_\beta)(\sigma_2 \uplus \tau), B\sep\I)
\end{align*}
\end{lemma}
\begin{proof}
The proof is by induction on the size of $\nextt^\star(\alpha\parallel_x\beta, \psi_1\land \psi_2, S)$. If the set is empty, then we have:
\begin{align*}
  &\minProb(\linlpo^\I(\alpha\parallel_x\beta, \psi_1 \land \psi_2, S)(\sigma_2 \uplus_2 \uplus \tau), A\sep B\sep \I)
  \\
  &= \minProb( \eta(\sigma_1 \uplus \sigma_2\uplus\tau), A\sep B \sep \I )
  \\
  &= \sum_{\sigma_1' \in A} \sum_{\sigma_2'\in B} \sum_{\tau' \in \I} \delta_{\sigma_1 \uplus\sigma_2\uplus\tau}(\sigma_1' \uplus\sigma_2' \uplus \tau')
  \intertext{Since we already know that $\tau \in \I$:}
  &= \sum_{\sigma_1' \in A} \sum_{\sigma_2'\in B} \sum_{\tau' \in \I} \delta_{\sigma_1}(\sigma_1') \cdot \delta_{\sigma_2}(\sigma_2') \cdot \delta_\tau(\tau')
    \\
  &= \left( \sum_{\sigma_1' \in A} \delta_{\sigma_1}(\sigma_1') \right) \cdot \left( \sum_{\sigma_2'\in B}  \delta_{\sigma_2}(\sigma_2') \right)
    \\
  &= \left( \sum_{\sigma' \in A} \sum_{\tau' \in \I} \delta_{\sigma_1\uplus\tau}(\sigma_1' \uplus \tau') \right) \cdot \left( \sum_{\tau'\in B} \sum_{\tau' \in \I} \delta_{\sigma_2\uplus\tau}(\sigma_2' \uplus \tau') \right)
    \\
  &= \minProb(\eta(\sigma_1 \uplus\tau), A\sep\I) \cdot \minProb(\eta(\sigma_2\uplus\tau), B\sep \I)
    \\
  &= \minProb(\linlpo^\I(\alpha, \psi_1, S\cap N_\alpha)(\sigma_1 \uplus\tau), A\sep\I) \cdot \minProb(\linlpo^\I(\beta, \psi_2, S\cap N_\beta)(\sigma_2 \uplus\tau, B\sep\I)
\end{align*}
Now suppose that the set of next elements is not empty. If $\nextt(\alpha\parallel_x\beta, \psi_1 \land \psi_2, S) = \{x\}$, then we know that $\lambda_{\alpha\parallel_x\beta}(x) = \fork$, and so:
\begin{align*}
  &\minProb(\linlpo^\I(\alpha \parallel_x \beta, \psi_1 \land \psi_2, S)(\sigma_1\uplus\sigma_2\uplus\tau), A\sep B\sep \I)
  \\
  &= \minProb(\linlpo^\I(\alpha \parallel_x \beta, \psi_1 \land \psi_2, S \cup\{x\})(\sigma_1\uplus\sigma_2\uplus\tau), A\sep B\sep\I)
  \intertext{By the induction hypothesis:}
  &= \minProb(\linlpo^\I(\alpha, \psi_1, S \cap N_\alpha)(\sigma_1\uplus\tau), A\sep\I)
      \cdot \minProb(\linlpo^\I(\beta, \psi_2, S \cap N_\beta)(\sigma_2\uplus\tau), B\sep\I)
\end{align*}
If not, then $x$ has already been scheduled and so $\nextt(\alpha\parallel_x\beta, \psi_1\land\psi_2, S) = \nextt(\alpha, \psi_1, S\cap N_\alpha) \cup \nextt(\beta, \psi_2, S\cap N_\beta)$.
Take any $y\in N_\alpha$, we will now show that:
\begin{align*}
  &\minProb\left(\linnode^\I(\alpha\parallel_x \beta, \psi_1 \land\psi_2, S, y)(\sigma_1 \uplus\sigma_2\uplus\tau), A\sep B\sep\I\right)
  \\
  &= \minProb\left(\linnode^\I(\alpha, \psi_1, S\cap N_\alpha, y)(\sigma_1 \uplus\tau), A\sep\I\right) \cdot \minProb\left(\linlpo^\I(\beta,\psi_2, S\cap N_\beta)(\sigma_2 \uplus \tau), B\sep\I\right)
\end{align*}
There are four cases:
\begin{enumerate}
\item $\lambda_{\alpha\parallel_x \beta}(y) = \lambda_\alpha(y) \in \act$. Let $a = \lambda_{\alpha\parallel_x \beta}(y)$, so:
{\footnotesize\begin{align*}
  &\minProb(\linnode^\I(\alpha\parallel_x\beta,\psi_1\land\psi_2,S, y)(\sigma_1 \uplus\sigma_2 \uplus\tau), A\sep B\sep\I)
  \\
  &= \minProb(\linlpo^\I(\alpha\parallel_x \beta, \psi_1\land\psi_2, S\cup \{y\})^\dagger( \de{a}_\act^\I(\sigma_1 \uplus\sigma_2\uplus\tau) ), A\sep B\sep\I)
  \intertext{By \Cref{lem:minProb-kleisli}.}
  &= \inf_{\mu \in \de{a}_\act^\I(\sigma_1 \uplus\sigma_2\uplus\tau)} \sum_{\rho \in \supp(\mu) \cap\mem{U_1 \cup V}} \mu(\rho) \cdot 
    \minProb(\linlpo^\I(\alpha\parallel_x \beta, \psi_1\land\psi_2, S\cup \{y\})(\rho), A\sep B\sep\I)
  \intertext{Since $\free(\alpha) \cap \dom(\sigma_2) = \emptyset$.}
  &= \inf_{\mu \in \de{a}_\act^\I(\sigma_1 \uplus\tau)} \sum_{\rho \in \supp(\mu) \cap\mem{U_1 \cup V}} \mu(\rho) \cdot 
    \minProb(\linlpo^\I(\alpha\parallel_x \beta, \psi_1\land\psi_2, S\cup \{y\})(\pi_{U_1}(\rho) \uplus \sigma_2 \uplus \pi_V(\rho)), A\sep B\sep\I)
  \intertext{Note that by the definition of invariant sensitive execution, either $\pi_V(\rho) \subseteq \I$, or $\de{a}_\act^\I(\sigma_1 \uplus\tau) = \bot_\C$. In the former case, we can apply the induction hypothesis. In the latter case, the entire expression must be zero, and therefore so is $\minProb(\linnode^\I(\alpha, \psi_1, S, y)(\sigma_1 \uplus\tau), A\sep\I)$, so the claim holds trivially. So, we presume the former is true and apply the induction hypothesis to get:}
  &= \inf_{\mu \in \de{a}_\act^\I(\sigma_1 \uplus\tau)} \sum_{\rho \in \supp(\mu)} \mu(\rho) \cdot 
    \minProb(\linlpo^\I(\alpha, \psi_1, S\cap N_\alpha \cup \{y\})(\rho), A\sep \I)
    \cdot \minProb(\linlpo^\I(\beta, \psi_2, S\cap N_\beta)(\sigma_2 \sep \pi_V(\rho)), B\sep \I)
  \intertext{Note that $\linlpo^\I(\beta, \psi_2, S\cap N_\beta)(\sigma_2 \uplus \tau) = \linlpo^\I(\beta, \psi_2, S\cap N_\beta)(\sigma_2 \uplus \tau')$ for any $\tau' \in \I$, since the invariant sensitive execution reassigns the $V$ variables at each step and tests only depend on local state.}
  &= \left(\inf_{\mu \in \de{a}_\act^\I(\sigma_1 \uplus\tau)} \sum_{\rho \in \supp(\mu)} \mu(\rho) \cdot 
    \minProb(\linlpo^\I(\alpha, \psi_1, S\cap N_\alpha \cup \{y\})(\rho), A\sep\I)\right)
    \cdot \minProb(\linlpo^\I(\beta, \psi_2, S\cap N_\beta)(\sigma_2 \uplus\tau, B\sep\I)
  \\
  &= \left( \minProb(\linlpo^\I(\alpha, \psi_1, S\cap N_\alpha \cup \{y\})^\dagger(\de{a}_\act^\I(\sigma_1 \uplus\tau)), A\sep\I)\right)
    \cdot \minProb(\linlpo^\I(\beta, \psi_2, S\cap N_\beta)(\sigma_2\uplus\tau), B\sep\I)
  \\
  &= \minProb( \linnode(\alpha,\psi_1, S\cap N_\alpha, y)(\sigma_1 \uplus\tau), A\sep\I)
    \cdot \minProb(\linlpo(\beta, \psi, S\cap N_\beta)(\sigma_2\uplus\tau), B\sep\I)
\end{align*}}

\item $\lambda_{\alpha\parallel_x\beta}(y) = \lambda_\alpha(y) \in\test$. Let $b = \lambda_{\alpha\parallel_x\beta}(y)$, and we therefore have:
\begin{align*}
  &\minProb(\linnode^\I(\alpha\parallel_x\beta,\psi_1\land\psi_2,S, y)(\sigma_1 \uplus \sigma_2 \uplus\tau), A\sep B\sep\I)
  \\
  &= \minProb(\linlpo^\I(\alpha\parallel_x \beta, \psi_1\land\psi_2 \land\sem{y = \de{b}_\test(\sigma_1 \uplus\sigma_2\uplus\tau)}, S\cup \{y\})(\sigma_1\uplus\sigma_2\uplus\tau), A\sep B\sep\I)
  \intertext{By the induction hypothesis, and the fact that $\free(b) \subseteq \free_\test(\alpha) \subseteq U_1$ (therefore it does not depend on $\sigma_2 \in \mem{U_2}$ and $\tau \in \mem V$.}
  &= \minProb(\linlpo^\I(\alpha, \psi_1\land\sem{y = \de{b}(\sigma_1)}, S \cap N_\alpha \cup \{y\})(\sigma_1 \uplus \tau), A\sep\I)
    \\ & \quad\cdot \minProb( \linlpo(\beta, \psi_2, S\cap N_\beta)(\sigma_2\uplus\tau), B\sep\I)
  \\
  &= \minProb(\linnode^\I(\alpha, \psi_1, S \cap N_\alpha, y)(\sigma_1 \uplus\sigma_2 \uplus\tau), A\sep\I)
    \cdot \minProb( \linlpo^\I(\beta, \psi_2, S\cap N_\beta)(\sigma_1\uplus\sigma_2\uplus\tau), B\sep\I)
\end{align*}

\item $\lambda_{\alpha\parallel_x\beta}(y) = \lambda_\alpha(y) = \bot$.
\begin{align*}
  &\minProb(\linnode^\I(\alpha\parallel_x\beta,\psi_1\land\psi_2,S, y)(\sigma_1 \uplus\sigma_2 \uplus\tau), A\sep B\sep\I)
  \\
  &= \minProb(\bot_\C, A\sep B)
  \\
  &= 0
  \\
  &= 0 \cdot \minProb( \linlpo^\I(\beta, \psi_2, S\cap N_\beta)(\sigma_2\uplus\tau), B\sep\I)
  \\
  &= \minProb(\linnode^\I(\alpha, \psi_1, S \cap N_\alpha,y)(\sigma_1\uplus\tau), A\sep\I)
    \cdot \minProb( \linlpo^\I(\beta, \psi_2, S\cap N_\beta)(\sigma_2\uplus\tau), B\sep\I)
\end{align*}

\item $\lambda_{\alpha\parallel_x\beta}(y) = \lambda_\alpha(y) = \fork$.
\begin{align*}
  &\minProb(\linnode^\I(\alpha\parallel_x\beta,\psi_1\land\psi_2,S,y)(\sigma_1\uplus\sigma_2\uplus\tau), A\sep B\sep\I)
  \\
  &= \minProb(\linlpo^\I(\alpha\parallel_x\beta, \psi_1\land\psi_2, S\cup\{y\})(\sigma_1\uplus\sigma_2\uplus\tau), A\sep B\sep\I)
  \intertext{By the induction hypothesis.}
  &= \minProb(\linlpo^\I(\alpha, \psi_1, S\cap N_\alpha\cup\{y\})(\sigma_1\uplus\tau), A\sep\I)
    \cdot \minProb( \linlpo^\I(\beta, \psi_2, S\cap N_\beta)(\sigma_2\uplus\tau), B\sep\I)
  \\
  &= \minProb(\linnode^\I(\alpha, \psi_1, S \cap N_\alpha,y)(\sigma_1\uplus\tau), A\sep\I)
    \cdot \minProb( \linlpo^\I(\beta, \psi_2, S\cap N_\beta)(\sigma_1\uplus\tau), B\sep\I)
\end{align*}

\end{enumerate}
By a nearly identical argument, we also get that for any $y\in N_\beta$:
\begin{align*}
  &\minProb(\linnode^\I(\alpha\parallel_x \beta, \psi_1 \land\psi_2, S,y)(\sigma_1\uplus\sigma_2\uplus\tau), A\sep B\sep\I)
  \\&= \minProb(\linlpo^\I(\alpha, \psi_1, S\cap N_\alpha)(\sigma_1\uplus\tau), A) \cdot \minProb(\linnode^\I(\beta,\psi_2, S\cap N_\beta,y)(\sigma_2\uplus\tau), B\sep\I)
\end{align*}
We now complete the proof as follows:
\begin{align*}
  &\minProb(\linlpo^\I(\alpha\parallel_x\beta, \psi_1\land\psi_2, S)(\sigma_1\uplus\sigma_2 \uplus \tau), A\sep B\sep\I)
  \\
  &= \minProb\left(
      \bignd_{y \in \nextt(\alpha\parallel_x\beta, \psi_1\land\psi_2, S)} \linnode^\I(\alpha\parallel_x\beta, \psi_1\land\psi_2, S,y)(\sigma_1\uplus\sigma_2\uplus\tau) , A\sep B\sep\I \right)
  \intertext{By \Cref{lem:minProb-nd}}
  &= {\min\limits_{y \in \nextt(\alpha\parallel_x\beta, \psi_1\land\psi_2, S)}} \minProb\left(\linnode^\I(\alpha\parallel_x\beta, \psi_1\land\psi_2, S,y)(\sigma_1\uplus\sigma_2\uplus\tau), A\sep B\sep\I \right)
  \\
  &= \min\Big( \min_{y \in \nextt(\alpha, \psi_1, S\cap N_\alpha)} \minProb(\linnode^\I(\alpha\parallel_x\beta, \psi_1\land\psi_2, S,y)(\sigma_1\uplus\sigma_2\uplus\tau), A\sep B\sep\I ),
    \\&\qquad
    \min_{y \in \nextt(\beta, \psi_2, S\cap N_\beta)} \minProb(\linnode^\I(\alpha\parallel_x\beta, \psi_1\land\psi_2, S, y)(\sigma_1\uplus\sigma_2\uplus\tau), A\sep B\sep\I ) \Big)
  \\
  &= \min\Big( \min_{y \in \nextt(\alpha, \psi_1, S\cap N_\alpha)}
      \minProb(\linnode^\I(\alpha, \psi_1, S\cap N_\alpha, y)(\sigma_1\uplus\tau), A\sep\I ) \cdot \minProb(\linlpo^\I(\beta, \psi_2, S\cap N_\beta)(\sigma_2\uplus\tau), B\sep\I),
    \\&\qquad
    \min_{y \in \nextt(\beta, \psi_2, S\cap N_\beta)}
        \minProb(\linlpo^\I(\alpha, \psi_1, S\cap N_\alpha)(\sigma_1\uplus\tau), A\sep\I ) \cdot \minProb(\linnode^\I(\beta, \psi_2, S\cap N_\beta, y)(\sigma_2\uplus\tau), B\sep\I) \Big)
  \\
  &= \min\Big( \left(\min_{y \in \nextt(\alpha, \psi_1, S\cap N_\alpha)}
      \minProb(\linnode^\I(\alpha, \psi_1, S\cap N_\alpha, y)(\sigma_1\uplus\tau), A\sep\I )\right) \cdot \minProb(\linlpo^\I(\beta, \psi_2, S\cap N_\beta)(\sigma_2\uplus\tau), B\sep\I),
    \\&\qquad
    \minProb(\linlpo^\I(\alpha, \psi_1, S\cap N_\alpha)(\sigma_1\uplus\tau), A\sep\I ) \cdot \left(\min_{y \in \nextt(\beta, \psi_2, S\cap N_\beta)}
         \minProb(\linnode^\I(\beta, \psi_2, S\cap N_\beta, y)(\sigma_2\uplus\tau), B\sep\I)\right) \Big)
  \\
  &= \min\Big(  \minProb \left(\bignd_{y \in \nextt(\alpha, \psi_1, S\cap N_\alpha)}
      \linnode^\I(\alpha, \psi_1, S\cap N_\alpha, y)(\sigma_1\uplus\tau), A\sep\I )\right) \cdot \minProb(\linlpo^\I(\beta, \psi_2, S\cap N_\beta)(\sigma_2\uplus\tau), B\sep\I),
    \\&\qquad
    \minProb(\linlpo^\I(\alpha, \psi_1, S\cap N_\alpha)(\sigma_1\uplus\tau), A\sep\I ) \cdot \minProb\left(\bignd_{y \in \nextt(\beta, \psi_2, S\cap N_\beta)}
         \linnode^\I(\beta, \psi_2, S\cap N_\beta, y)(\sigma_2\uplus\tau), B\sep\I\right) \Big)
  \\
  &= \min\Big(  \minProb(\linlpo^\I(\alpha, \psi_1, S\cap N_\alpha)(\sigma_1\uplus\tau), A\sep\I ) \cdot \minProb(\linlpo^\I(\beta, \psi_2, S\cap N_\beta)(\sigma_2\uplus\tau), B\sep\I),
    \\&\qquad
    \minProb(\linlpo^\I(\alpha, \psi_1, S\cap N_\alpha)(\sigma_1\uplus\tau), A\sep\I ) \cdot \minProb(\linlpo^\I(\beta, \psi_2, S\cap N_\beta)(\sigma_2\uplus\tau), B\sep\I)
    \Big)
  \\
  &= \minProb(\linlpo^\I(\alpha, \psi_1, S\cap N_\alpha)(\sigma_1\uplus\tau), A\sep\I ) \cdot \minProb(\linlpo^\I(\beta, \psi_2, S\cap N_\beta)(\sigma_2\uplus\tau), B\sep\I)
\end{align*}

\end{proof}

\begin{lemma}\label{lem:minProb-ind}
For any pairwise disjoint $U_1,U_2,V \subseteq \mathsf{Var}$ and $\Alpha,\Beta\in\pom_\fin$ such that $\free_\act(\Alpha) \subseteq U_1 \cup V$, $\free_\test(\Alpha)\subseteq U_1$, $\free_\act(\Beta) \subseteq U_2\cup V$, and $\free_\test(\Beta)\subseteq U_2$; $\I \subseteq \mem V$, $\sigma_1 \in \mem{U_1}$, $\sigma_2 \in \mem{U_2}$, $\tau \in \I$, $A \subseteq \mem{U_1}$, and $B\subseteq\mem{U_2}$:
\[
  \minProb\left(\linfin^\I(\Alpha \parallel \Beta)(\sigma_1\uplus\sigma_2\uplus \tau), A\sep B\sep\I\right)
  =
  \minProb\left(\linfin^\I(\Alpha)(\sigma_1\uplus\tau), A\sep\I\right) \cdot
  \minProb\left(\linfin^\I(\Beta)(\sigma_2\uplus\tau), B\sep\I\right)
\]
\end{lemma}
\begin{proof}
Fix any $\alpha \in \Alpha$, $\beta\in\Beta$, and $x \notin N_\alpha \cup N_\beta$. This give us $\Alpha \parallel \Beta = [ \alpha \parallel_x \beta ]$. So, we get:
\begin{align*}
  &\minProb\left(\linfin^\I(\Alpha \parallel \Beta)(\sigma_1\uplus\sigma_2\uplus \tau), A\sep B\sep\I\right)
  \\
  &= \minProb\left(\linfin^\I([ \alpha \parallel_x \beta])(\sigma_1\uplus\sigma_2\uplus \tau), A\sep B\sep\I\right)
  \\
  &= \minProb\left(\linlpo^\I(\alpha \parallel_x \beta, \tru, \emptyset)(\sigma_1\uplus\sigma_2\uplus \tau), A\sep B\sep\I\right)
  \intertext{By \Cref{lem:minProb-linlpo}.}
  &= \minProb\left(\linlpo^\I(\alpha, \tru, \emptyset)(\sigma_1\uplus\tau), A\sep\I\right) \cdot
        \minProb\left(\linlpo^\I(\beta, \tru, \emptyset)(\sigma_2\uplus\tau), B\sep\I\right)
  \\
  &= \minProb\left(\linfin^\I([\alpha])(\sigma_1\uplus\tau), A\sep\I\right) \cdot
          \minProb\left(\linlpo^\I([\beta])(\sigma_2\uplus\tau), B\sep\I\right)
  \\
  &= \minProb\left(\linfin^\I(\Alpha)(\sigma_1\uplus\tau), A\sep\I\right) \cdot \minProb\left(\linlpo^\I(\Beta)(\sigma_2\uplus\tau), B\sep\I\right)
\end{align*}

\end{proof}

\begin{lemma}\label{lem:par-ind}
For any pairwise disjoint $U_1,U_2,V \subseteq \mathsf{Var}$ and $\Alpha,\Beta\in\pom$ such that $\free_\act(\Alpha) \subseteq U_1 \cup V$, $\free_\test(\Alpha)\subseteq U_1$, $\free_\act(\Beta) \subseteq U_2\cup V$, and $\free_\test(\Beta)\subseteq U_2$; $\I \subseteq \mem V$, $\sigma_1 \in \mem{U_1}$, $\sigma_2 \in \mem{U_2}$, $\tau \in \I$, $A \subseteq \mem{U_1}$, and $B\subseteq\mem{U_2}$, if there exists $p$ and $q$ such that:
\[
  \forall \nu_1 \in \lin^\I(\Alpha)(\sigma_1\uplus\tau). \ \nu_1(B_1\sep \I) = p
  \qquad\text{and}\qquad
  \forall \nu_2 \in \lin^\I(\Beta)(\sigma_2\uplus\tau). \ \nu_2(B_2\sep \I) = q
\]
Then:
\[
  \forall \nu \in \lin^\I(\Alpha \parallel \Beta)(\sigma_1 \uplus\sigma_2 \uplus \tau). \ \nu(B_1 \sep B_2\sep \I) \ge p \cdot q
\]
\end{lemma}
\begin{proof}
Take any $\nu \in \lin^\I(\Alpha\parallel\Beta)(\sigma_1\uplus\sigma_2\uplus\tau)$. We have:
\begin{align*}
  \nu(B_1 \sep B_2\sep\I)
  &\ge \minProb(\lin^\I(\Alpha \parallel \Beta)(\sigma_1\uplus\sigma_2\uplus\tau), B_1 \sep B_2\sep\I)
  \\
  &= \minProb\left(\sup_{\GGamma \ll \Alpha \parallel\Beta} \linfin^\I(\GGamma)(\sigma_1\uplus\sigma_2\uplus\tau), B_1 \sep B_2 \sep\I \right)
  \intertext{By \Cref{lem:minProb-cont}.}
  &= \sup_{\GGamma \ll \Alpha \parallel \Beta} \minProb\left(\linfin^\I(\GGamma)(\sigma_1\uplus\sigma_2\uplus\tau), B_1 \sep B_2\sep\I\right)
  \\
  &= \sup_{\Alpha' \ll_1 \Alpha} \sup_{\Beta' \ll_1 \Beta} \minProb\left(\linfin^\I(\Alpha' \parallel \Beta')(\sigma_1\uplus\sigma_2\uplus\tau), B_1 \sep B_2\sep\I\right)
  \\
  \intertext{By \Cref{lem:minProb-ind}.}
  &= \sup_{\Alpha' \ll_1 \Alpha} \sup_{\Beta' \ll_1 \Beta} \minProb\left(\linfin^\I(\Alpha')(\sigma_1\uplus\tau), B_1\sep\I\right) \cdot \minProb\left(\linfin^\I(\Beta')(\sigma_2\uplus\tau), B_2\sep\I\right)  
  \\
  &= \left( \sup_{\Alpha' \ll_1 \Alpha} \minProb\left(\linfin^\I(\Alpha')(\sigma_1\uplus\tau), B_1\sep\I\right)\right) \cdot
    \left( \sup_{\Beta' \ll_1 \Beta} \minProb\left(\linfin^\I(\Beta')(\sigma_2\uplus\tau), B_2\sep\I\right) \right)
  \intertext{By \Cref{lem:minProb-cont}.}
  &= \minProb\left(\lin^\I(\Alpha)(\sigma_1\uplus\tau), B_1\sep\I\right) \cdot  \minProb\left(\lin^\I(\Beta)(\sigma_2\uplus\tau), B_2\sep\I\right)
  \\
  &= p \cdot q
\end{align*}
\end{proof}

\begin{lemma}\label{lem:par}
For any pairwise disjoint $U_1,U_2,V \subseteq \mathsf{Var}$ and $\Alpha_1,\Alpha_2\in\pom$ such that $\free_\act(\Alpha_k) \subseteq U_k \cup V$ and $\free_\test(\Alpha_k)\subseteq U_k$; $\I \subseteq \mem V$,
let $\P_\I$ be the trivial probability space where $\mu_{\P_\I}(\I) = 1$. For all $k\in\{1, 2\}$, $\P_k$, $\Q_k$, and $\mu \in \D(\mem{U_1\cup U_2\cup V})$ such that $\P_1 \otimes \P_2 \otimes \P_\I \preceq \mu$, if:
\[
  \forall \mu_k.\ \P_k \otimes \P_\I \preceq \mu_k \implies \forall \nu_k \in \evl^\I(\Alpha_k)^\dagger(\mu_k).\ \Q_k \otimes \P_\I \preceq \nu_k
\]
Then:
\[
  \forall \nu \in \evl^\I(\Alpha_1\parallel \Alpha_2)^\dagger(\mu).\ \Q_1\otimes\Q_2\otimes \P_\I \preceq \nu
\]
\end{lemma}
\begin{proof}

Let $(A_{1,i})_{i\in I_1}$ and $(A_{2,i})_{i\in I_2}$ be the most precise, disjoint measurable events from $\ev(\P_1)$ and $\ev(\P_2)$, respectively. Since $\P_1\otimes \P_2 \otimes \P_\I \preceq\mu$, we know that for any $(i,j) \in I_1 \times I_2$:
\begin{equation}\label{eq:ind-precondition}
  \sum_{\sigma_1\in A_{1,i}} \sum_{\sigma_2\in A_{2,j}} \sum_{\tau\in\I} \mu(\sigma_1\uplus\sigma_2 \uplus \tau)
  =
  \mu_{\P_1}(A_{1,i}) \cdot \mu_{\P_2}(A_{2,j})
  =
  \mu_{\P_1\otimes \P_\I}(A_{1,i} \sep\I) \cdot \mu_{\P_2 \otimes\I}(A_{2,j} \sep \I)
\end{equation}
Also, for any $\mu_k$ such that $\P_k\otimes \P_\I \preceq \mu_k$, we know that any $\nu_k \in \evl^\I(\alpha_k)^\dagger(\mu_k)$ has the form:
\[
  \nu_k = \smashoperator{\sum_{\sigma\in \supp(\mu_k)}} \mu_k(\sigma)\cdot \nu_\sigma
\]
Where each $\nu_\sigma \in \evl^\I(\Alpha_k)(\sigma)$. Since $\Q_k\otimes \P_\I \preceq \nu_k$, we know that for any $B \in \F_{\Q_k}$:
\begin{align*}
  \mu_{\Q_k\otimes \P_\I}(B \sep \I) &= \sum_{\tau \in B\sep\I} \nu_k(\tau) = \smashoperator{\sum_{\sigma\in \supp(\mu_k)}} \mu_k(\sigma)\cdot \sum_{\tau\in B\sep\I} \nu_\sigma(\tau)
  \intertext{Let $p_{B,\sigma} = \sum_{\tau\in B\sep \I} \nu_\sigma(\tau)$.}
  &= \smashoperator{\sum_{\sigma\in \supp(\mu_k)}} \mu_k(\sigma)\cdot p_{B,\sigma}
\end{align*}
We will now show that $\sum_{\tau\in B\sep\I}\xi(\tau) = p_{B,\sigma'}$ for any $\sigma'\in\supp(\mu_k)$ and $\xi\in\lin^\I(\Alpha_k)(\sigma')$. Take any such $\xi$. By construction, $\mu_k(\sigma')\cdot\xi + \sum_{\sigma\in\supp(\mu_k)\setminus\{\sigma'\}} \mu_k(\sigma)\cdot\nu_\sigma \in \lin^\I(\Alpha_k)^\dagger(\mu_k)$, therefore:
\begin{align*}
  \mu_{\Q_k \otimes \P_\I}(B\sep \I) = \sum_{\tau\in B\sep\I} \left(\mu_k(\sigma')\cdot\xi + \smashoperator{\sum_{\sigma\in\supp(\mu_k)\setminus\{\sigma'\}}} \mu_k(\sigma)\cdot\nu_\sigma \right)(\tau)
  &= \smashoperator{\sum_{\sigma\in \supp(\mu_k)}} \mu_k(\sigma)\cdot p_{B,\sigma}
  \\
  \mu_k(\sigma')\cdot \sum_{\tau\in B\sep\I} \xi(\tau) + \cancel{\smashoperator{\sum_{\sigma\in\supp(\mu_k)\setminus\{\sigma'\}}} \mu_k(\sigma)\cdot p_{B,\sigma}}
  &= \mu_k(\sigma')\cdot p_{B,\sigma'} + \cancel{\smashoperator{\sum_{\sigma\in \supp(\mu_k)\setminus\{\sigma'\}}} \mu_k(\sigma)\cdot p_{B,\sigma}}
  \\
  \cancel{\mu_k(\sigma')}\cdot \sum_{\tau\in B\sep\I} \xi(\tau)
  &= \cancel{\mu_k(\sigma')}\cdot p_{B,\sigma'}
  \\
  \sum_{\tau\in B\sep\I} \xi(\tau) &= p_{B,\sigma'}
\end{align*}
Now, let $\mu_k$ be constructed by fixing a single state $\sigma_i \in A_{k,i} \sep \I$ for each $i \in I_k$ and setting $\mu_k(\sigma_i) = \mu_{\P_k\otimes \P_\I}(A_{k,i})$, so clearly $\P_k\otimes \P_\I \preceq \mu_k$. Now let $p_{B,k,i} = p_{B,\sigma_i}$. Based on what we have just showed, this gives us:
\[
  \mu_{\Q_k\otimes \P_\I}(B\sep\I)
  = \smashoperator{\sum_{\sigma\in \supp(\mu_k)}} \mu_k(\sigma) \cdot p_{B, \sigma}
  = \smashoperator{\sum_{i \in I_k}} \mu_k(\sigma_i) \cdot p_{B, \sigma_i}
  = \smashoperator{\sum_{i \in I_k}} \mu_{\P_k\otimes \P_\I}(A_{k,i} \sep \I) \cdot p_{B, k,i}
\]
We will now show that $p_{B,\sigma'} = p_{B,k,j}$ for any $j\in I_k$ and $\sigma' \in A_{k,j}\sep \I$. Take any such $\sigma'$, then we get:
{\footnotesize\begin{align*}
  \mu_{\Q_k\otimes\P_\I}(B\sep\I) &=
  \\\mu_{\P_k\otimes \P_\I}(A_{k,j}\sep\I) \cdot p_{B,\sigma'} + \smashoperator{\sum_{i \neq j}} \mu_{\P_k\otimes\P_\I}(A_{k,i}\sep\I) \cdot p_{B, \sigma_i}
  &= \smashoperator{\sum_{i \in I_k}} \mu_{\P_k\otimes\P_\I}(A_{k,i}\sep\I) \cdot p_{B, k, i}
  \\
  \mu_{\P_k\otimes\P_\I}(A_{k,j}\sep\I) \cdot p_{B,\sigma'} + \cancel{\smashoperator{\sum_{i \neq j}} \mu_{\P_k\otimes\P_\I}(A_{k,i}\sep\I) \cdot p_{B, k, i}}
  &= \mu_{\P_k\otimes\P_\I}(A_{k,j}\sep\I) \cdot p_{B,k, j} + \cancel{\smashoperator{\sum_{i \neq j}} \mu_{\P_k\otimes\P_\I}(A_{k,i}\sep\I) \cdot p_{B, k, i}}
  \\
  \cancel{\mu_{\P_k\otimes\P_\I}(A_{k,j}\sep\I)} \cdot p_{B,\sigma'}
  &= \cancel{\mu_{\P_k\otimes\P_\I}(A_{k,j}\sep\I)} \cdot p_{B,k, j}
  \\
  p_{B,\sigma'} &= p_{B,k, j}
\end{align*}}
We have therefore shown that $\sum_{\tau\in B} \nu_k(\tau) = p_{B,i,k}$ for any $\nu_k \in \lin^\I(\Alpha_k)(\sigma)$ where $B \in \F_{\Q_k}$, $i\in I_k$, and $\sigma \in A_{k, i}$

Now take any $\nu \in \evl^\I(\Alpha_1 \parallel\Alpha_2)^\dagger(\mu)$, which must have the form $\nu=\sum_{\sigma\in\supp(\mu)} \mu(\sigma)\cdot \nu_\sigma$ where $\nu_\sigma \in \evl^\I(\Alpha_1 \parallel \Alpha_2)(\sigma)$ for each $\sigma$. Since $\Q_1$ and $\Q_2$ operate over different address spaces, clearly $\Q_1\otimes\Q_2$ exists. It just remains to show that $\Q_1\otimes\Q_2\otimes \P_\I \preceq \nu$, which we do as follows.
By \Cref{lem:prod-equiv}, it suffices to show that the probability measures agree on the product of disjoint partitions. 
Let $\{ B_{k,j} \mid j\in J_k \} = \ev(\Q_k)$ for $k\in\{1,2\}$. For any $i\in J_1$ and $j\in J_2$, we have:
\begin{align*}
\nu(B_{1,i} \sep B_{2,j} \sep \I)
   &= {\sum_{\sigma\in\supp(\mu)}} \mu(\sigma) \cdot \nu_\sigma(B_{1,i} \sep B_{2,j}\sep \I)
   \intertext{Instead of summing over the support of $\sigma$, we can alternatively sum over the elements of the $A_{1,i'}$ and $A_{2,j'}$ sets.}
   &= \sum_{i'\in I_1} \sum_{j'\in I_2}  \sum_{\sigma_1\in A_{1,i'}} \sum_{\sigma_2\in A_{2,j'}} \sum_{\tau\in \I} \mu(\sigma_1 \uplus\sigma_2\uplus \tau) \cdot
      \nu_\sigma(B_{1,i} \sep B_{2,j}\sep \I)
   \intertext{We previously showed that $\nu_k(B_{k, \ell} \sep\I) = p_{B_{k, \ell},k,i'}$ (a constant) for any $\nu_k\in \evl^\I(\Alpha_k)(\sigma)$ where $\sigma \in A_{k,i'}$. So, we can use \Cref{lem:par-ind} to conclude that:}
  &\ge \sum_{i'\in I_1} \sum_{j'\in I_2}  \sum_{\sigma_1\in A_{1,i'}} \sum_{\sigma_2\in A_{2,j'}} \sum_{\tau\in \I} \mu(\sigma_1 \uplus\sigma_2\uplus \tau)
   \cdot p_{B_{1,i},1,i'}\cdot p_{B_{2,j},2,j'}
  \\
  &= \sum_{i'\in I_1} \sum_{j'\in I_2} 
      \mu(A_{1, i'}\sep A_{2,j'} \sep\I)
   \cdot p_{B_{1,i},1,i'}\cdot p_{B_{2,j},2,j'}
  \intertext{By \Cref{eq:ind-precondition}.}
  &= \sum_{i'\in I_1} \sum_{j'\in I_2} \mu_{\P_1\otimes \P_\I}(A_{1,i'}\sep\I)  \cdot \mu_{\P_2\otimes\P_\I}(A_{2,j'}\sep\I) \cdot p_{B_{1,i},1,i'}\cdot p_{B_{2,j},2,j'}
\\
  &= \left( \sum_{i'\in I_1} \mu_{\P_1\otimes\P_\I}(A_{1,i'}\sep\I) \cdot p_{B_{1,i},1,i'} \right) \cdot \left( \sum_{j'\in I_2} \mu_{\P_2\otimes\P_\I}(A_{2,j'}\sep\I) \cdot p_{B_{2,j},2,j'} \right)
  \\
  &=\mu_{\Q_1\otimes\P_\I}(B_{1,i}\sep\I) \cdot \mu_{\Q_2\otimes\P_\I}(B_{2,j}\sep\I)
  \\
  &=\mu_{\Q_1}(B_{1,i}) \cdot \mu_{\Q_2}(B_{2,j}) \cdot \mu_{\P_\I}(\I)
\end{align*}
Now, we have shown that $\nu(B_{1,i}\sep B_{2,j}\sep\I) \ge \mu_{\Q_1 \otimes \Q_2\otimes \P_\I}(B_{1,i} \sep B_{2,j} \sep\I)$ for all $i\in J_1$ and $j\in J_2$. Since $\sum_{i\in J_1} \sum_{j\in J_2} \mu_{\Q_1 \sep \Q_2}(B_{1,i} \sep B_{2,j}) = 1$, then $\nu(B_{1,i}\sep B_{2,j}\sep\I)$ cannot be strictly greater then $\mu_{\Q_1}(B_{1,i})\cdot \mu_{\Q_2}(B_{2,j})$ for any $i$ or $j$, and therefore the quantities must be equal.

\end{proof}

\section{Almost Sure Termination}

For any $S \in \C(\mem V)$, let $\minterm(S) = \minProb(S, \mem V)$, \ie it is the minimum probability that the program terminates. Also, let $\Psi_{\tuple{b,C,\I}} \colon (\mem V \to \C(\mem V)) \to \mem V \to \C(\mem V)$ be defined as follows:
\[
  \Psi_{\tuple{b,C, \I}}(f)(\sigma) \triangleq \left\{
    \begin{array}{ll}
      f^\dagger(\lin^\I(\de{C})(\sigma)) & \text{if}~ \de{b}_\test(\sigma) = \tru
      \\
      \eta(\sigma) & \text{if}~ \de{b}_\test(\sigma) = \fls
    \end{array}
  \right.
\]
By \citet[Lemma 5.2]{zilberstein2025denotational}:
\[
  \lin^\I(\de{\whl bC})
  = \mathsf{lfp}\left(\Psi_{\tuple{b,C,\I}}\right)
  = \sup_{n\in\mathbb N} \Psi^n_{\tuple{b,C,\I}}(\bot_\C^\bullet)
\]
Now, for any test $b$ and distribution $\mu$ we define a conditioning operator as follows:
\[
  (b\?\mu)(\sigma) \triangleq \left\{
    \begin{array}{ll}
      \frac{\mu(\sigma)}{\mu(b)} & \text{if}~ \de{b}_\test(\sigma) = \tru
      \\
      0 & \text{if}~ \de{b}_\test(\sigma) = \fls
    \end{array}
  \right.
  \qquad\text{where}\qquad
  \mu(b) \triangleq \smashoperator{\sum_{\sigma\in\supp(\mu)\mid \de{b}_\test(\sigma) = \tru}} \mu(\sigma)
\]
Note that if $\mu(b) = 0$, then $b\?\mu$ is not well-defined. In that case, we just let $b\?\mu = \bot_\D$. It is also clearly true that $\mu = (b\?\mu) \oplus_{\mu(b)} (\lnot b\?\mu)$ for any $b$ and $\mu$. In addition, we call $\tuple{\varphi,\psi}$ an loop invariant pair for $\whl bC$ under the resource invariant $I$ iff:
\begin{enumerate}
\item $\varphi\Rightarrow\sure{b\mapsto\tru}$
\item $\psi\Rightarrow \sure{b\mapsto\fls}$
\item $I\vDash_\wk\triple{\varphi}C{\varphi\nd\psi}$
\item $\precise\psi$
\end{enumerate}
Given these new definitions, we prove some partial correctness results, which show that the $\psi$ (as defined above) holds on the terminating portion of the result of a while loop.

\begin{lemma}\label{lem:par-cor-fin}
Take any $\Gamma$, let $\I = \sem{I}_\Gamma$, $\tuple{\varphi,\psi}$ be an invariant pair for $\whl bC$ under $I$, and $\Q$ be the unique smallest probability space satisfying $\psi$ (which exists since $\precise\psi$).
For any $\mu$, $n\in\mathbb N$, and $A \in \F_\Q$ such that $\Gamma,\mu\vDash\varphi\sep\sure I$:
\[
  \minProb\left( \Psi_{\tuple{b, C, \I}}^n\left(\bot_\C^\bullet\right)^\dagger(\mu), A\sep\I \right)
  =
  \minterm\left( \Psi_{\tuple{b, C, \I}}^n\left(\bot_\C^\bullet\right)^\dagger(\mu) \right) \cdot \mu_\Q(A)
\]
\end{lemma}
\begin{proof}
The proof is by induction on $n$. Suppose that $n=0$, and so we have:
\begin{align*}
  \minProb\left( \Psi_{\tuple{b, C, \I}}^0\left(\bot_\C^\bullet\right)^\dagger(\mu), A\sep\I \right)
  &= \minProb\left(  \bot_\C, A\sep\I \right)
  \\
  &= 0
  \\
  &= \minterm\left(  \bot_\C \right) \cdot \mu_\Q(A)
  \\
  &= \minterm\left(  \Psi_{\tuple{b, C, \I}}^0\left(\bot_\C^\bullet\right)^\dagger(\mu) \right) \cdot \mu_\Q(A)
\end{align*}
Now suppose that $n=1$, and so we have:
\begin{align*}
  \minProb\left( \Psi_{\tuple{b, C, \I}}^1\left(\bot_\C^\bullet\right)^\dagger(\mu), A\sep\I \right)
  &= \minProb\left( \left(\bot_\C^\bullet\right)^\dagger\left(\lin^\I(\de{C})^\dagger(\mu) \right), A\sep\I \right)
  \\
  &= \minProb\left( \bot_\C, A\sep\I \right)
  \\
  &= 0
  \\
  &= \minterm\left(  \bot_\C \right) \cdot \mu_\Q(A)
  \\
  &= \minterm\left(  \Psi_{\tuple{b, C, \I}}^1\left(\bot_\C^\bullet\right)^\dagger(\mu) \right) \cdot \mu_\Q(A)
\end{align*}
Now suppose that $n > 1$. Then, we have:
\begin{align*}
  &\minProb\left( \Psi_{\tuple{b, C, \I}}^n\left(\bot_\C^\bullet\right)^\dagger(\mu), A\sep\I \right)
  \\
  &= \minProb\left( \Psi_{\tuple{b, C, \I}}^{n-1}\left(\bot_\C^\bullet\right)^\dagger(\lin^\I(\de{C})^\dagger(\mu)), A\sep\I \right)
  \\
  &= \inf_{\nu\in \lin^\I(\de{C})^\dagger(\mu)} \minProb\left( \Psi_{\tuple{b, C, \I}}^{n-1}\left(\bot_\C^\bullet\right)^\dagger(\nu), A\sep\I \right)
  \intertext{Since $\tuple{\varphi,\psi}$ is an invariant pair and $\Gamma,\mu\vDash\varphi\sep\sure I$, then every such $\nu$ above can be split into $\nu = b\?\nu \oplus_{\nu(b)} \lnot b\?\nu$ such that $\Gamma,b\?\nu\vDash\varphi\sep\sure I$ and $\Gamma,\lnot b\?\nu\vDash\psi\sep\sure I$.}
  &= \inf_{\nu\in \lin^\I(\de{C})^\dagger(\mu)} \minProb\left( \Psi_{\tuple{b, C, \I}}^{n-1}\left(\bot_\C^\bullet\right)^\dagger(b\?\nu \oplus_{\nu(b)} \lnot b\?\nu), A\sep\I \right)
  \\
  &= \inf_{\nu\in \lin^\I(\de{C})^\dagger(\mu)} \minProb\left(
    \Psi_{\tuple{b, C, \I}}^{n-1}\left(\bot_\C^\bullet\right)^\dagger(b\?\nu) \oplus_{\nu(b)} \Psi_{\tuple{b, C, \I}}^{n-1}\left(\bot_\C^\bullet\right)^\dagger(\lnot b\?\nu)
    , A\sep\I \right)
  \intertext{Since $n>1$, then $\Psi_{\tuple{b, C, \I}}^{n-1}(\bot_\C^\bullet)^\dagger(\lnot b\?\nu) = \{ \lnot b\?\nu \}$.}
  &= \inf_{\nu\in \lin^\I(\de{C})^\dagger(\mu)} \minProb\left(
    \Psi_{\tuple{b, C, \I}}^{n-1}\left(\bot_\C^\bullet\right)^\dagger(b\?\nu) \oplus_{\nu(b)} \{ \lnot b\?\nu \}
    , A\sep\I \right)
  \\
  &= \inf_{\nu\in \lin^\I(\de{C})^\dagger(\mu)} \nu(b)\cdot\minProb\left(
    \Psi_{\tuple{b, C, \I}}^{n-1}\left(\bot_\C^\bullet\right)^\dagger(b\?\nu)
    , A\sep\I \right) + (1 - \nu(b)) \cdot  (\lnot b\?\nu)(A\sep\I)
  \intertext{Since $\Gamma,\lnot b\?\nu\vDash\psi\sep\sure I$, then $(\lnot b\?\nu)(A\sep\I) = \mu_\Q(A)$. Also using the induction hypothesis, we get:}
  &= \inf_{\nu\in \lin^\I(\de{C})^\dagger(\mu)} \nu(b)\cdot\minterm\left(
    \Psi_{\tuple{b, C, \I}}^{n-1}\left(\bot_\C^\bullet\right)^\dagger(b\?\nu)\right)\cdot \mu_\Q(A) + (1 - \nu(b)) \cdot  \mu_\Q(A)
  \\
  &= \left( 
      \inf_{\nu\in \lin^\I(\de{C})^\dagger(\mu)} \minterm\left(
        \Psi_{\tuple{b, C, \I}}^{n-1}\left(\bot_\C^\bullet\right)^\dagger(b\?\nu) \oplus_{\nu(b)} \{ \lnot b\?\nu \}
    \right)
    \right)\cdot  \mu_\Q(A)
  \\
  &= \left( 
      \inf_{\nu\in \lin^\I(\de{C})^\dagger(\mu)} \minterm\left(
        \Psi_{\tuple{b, C, \I}}^{n-1}\left(\bot_\C^\bullet\right)^\dagger(\nu)
    \right)
    \right)\cdot  \mu_\Q(A)
  \\
  &= \minterm\left(\Psi_{\tuple{b, C, \I}}^{n}\left(\bot_\C^\bullet\right)^\dagger(\mu)\right)\cdot  \mu_\Q(A)
\end{align*}
\end{proof}

\begin{lemma}\label{lem:par-cor}
Take any $\Gamma$, let $\I = \sem{I}_\Gamma$, $\tuple{\varphi,\psi}$ be an invariant pair for $\whl bC$ under $I$, and $\Q$ be the unique smallest probability space satisfying $\psi$ (which exists since $\precise\psi$).
For any $\mu$ such that $\Gamma,\mu\vDash\varphi\sep\sure I$ and $A \in \F_\Q$:
\[
  \minProb\left( \lin^\I\left( \de{\whl bC} \right)^\dagger(\mu), A\sep\I \right)
  =
  \minterm\left( \lin^\I\left( \de{\whl bC} \right)^\dagger(\mu) \right) \cdot \mu_\Q(A)
\]
\end{lemma}
\begin{proof}
The proof proceeds as follows:
\begin{align*}
  \minProb\left( \lin^\I\left( \de{\whl bC} \right)^\dagger(\mu), A\sep\I \right)
  &= \minProb\left( \sup_{n\in\mathbb N} \Psi_{\tuple{b, C, \I}}^n\left(\bot_\C^\bullet\right)^\dagger(\mu) , A\sep\I \right)
  \\
  &= \sup_{n\in\mathbb N} \minProb\left( \Psi_{\tuple{b, C, \I}}^n\left(\bot_\C^\bullet\right)^\dagger(\mu), A\sep\I \right)
  \intertext{By \Cref{lem:par-cor-fin}.}
  &= \sup_{n\in\mathbb N} \minterm\left( \Psi_{\tuple{b, C, \I}}^n\left(\bot_\C^\bullet\right)^\dagger(\mu)\right) \cdot \mu_\Q(A)
  \\
  &= \minterm\left( \sup_{n\in\mathbb N} \Psi_{\tuple{b, C, \I}}^n\left(\bot_\C^\bullet\right)^\dagger(\mu)\right) \cdot \mu_\Q(A)
  \\
  &= \minterm\left( \lin^\I(\de{\whl bC})^\dagger(\mu)\right) \cdot \mu_\Q(A)
\end{align*}
\end{proof}

\begin{corollary}[Almost Sure Termination]\label{cor:ast}
Take any $\Gamma$ and $0 < p \le 1$, let $\I = \sem{I}_\Gamma$, and $\tuple{\varphi,\psi}$ be an invariant pair for $\whl bC$ under $I$. If additionally $\minterm(\lin^\I(\de{\whl bC})^\dagger(\mu)) \ge p$ for all $\mu$ such that $\Gamma,\mu\vDash\varphi\sep\sure I$, then $\minterm(\lin^\I(\de{\whl bC})^\dagger(\mu)) = 1$.
\end{corollary}
\begin{proof}
Follows by an identical argument to Lemma D.3 of \citet{zilberstein2025demonic}.
\end{proof}

\section{Logic and Rules}

\begin{lemma}[Monotonicity]
\label{lem:mono-sat}
If $\Gamma,\P \vDash \varphi$ and $\P\preceq \Q$, then $\Gamma,\Q\vDash\varphi$.
\end{lemma}
\begin{proof}
By induction on the structure of $\varphi$.

\begin{itemize}
\item $\varphi = \top$. Since $\Gamma,\Q\vDash\top$ for all $\Q$, the claim holds trivially.

\item $\varphi=\bot$. The premise is false, therefore this case is vacuous.

\item $\varphi = \varphi_1 \land\varphi_2$. By the induction hypothesis, we know that $\Gamma,\Q\vDash\varphi_1$ and $\Gamma,\Q\vDash\varphi_2$, therefore $\Gamma,\Q \vDash\varphi_1\land\varphi_2$.

\item $\varphi = \varphi_1 \lor\varphi_2$. Without loss of generality, suppose that $\Gamma,\P\vDash \varphi_1$. By the induction hypothesis, we know that $\Gamma,\Q\vDash\varphi_1$, therefore by weakening $\Gamma,\Q \vDash\varphi_1\lor\varphi_2$.

\item $\varphi = \exists X.\psi$. We know that there is some $v\in\mathsf{Val}$ such that $\Gamma[X\coloneqq v], \P \vDash\psi$. By the induction hypothesis, we get that $\Gamma[X\coloneqq v], \Q \vDash \psi$. This implies that $\Gamma,\Q \vDash\exists X.\psi$.

\item $\varphi = \bigoplus_{X\sim d(E)} \psi$. Immediate, since the semantics stipulates that $\P$ is greater than the direct sum, therefore $\Q$ is also greater than the direct sum.

\item $\varphi = \bignd_{X\in E} \psi$. We know that $\Gamma,\P\vDash\bigoplus_{X\sim \mu}$ for some $\mu \in \D(\de{E}_{\mathsf{LExp}}(\Gamma))$. By the previous case, we get that $\Gamma,\Q\vDash\bigoplus_{X\sim\mu}\psi$. This implies that $\Gamma,\Q\vDash\bignd_{X\in E}\psi$.

\item $\varphi = \varphi_1 \sep_m \varphi_2$. Immediate, since we know that $\P_1 \diamond_m \P_2 \preceq \P$ such that $\Gamma,\P_i\vDash\varphi_i$ for each $i$, and therefore $\P_1 \diamond_m \P_2 \preceq \Q$ as well.

the semantics stipulates that $\P$ is greater than the independent product, therefore $\Q$ is also greater than the independent product.

\item $\varphi = \sure P$. Let $\Omega_\P = \mem S$ and $\Omega_\Q = \mem T$, and note that $S\subseteq T$ since $\P \preceq\Q$. We know that $\sem{P}_\Gamma^S \in \F_\P$ and $\mu_\P(\sem{P}_\Gamma^S ) = 1$.
We also know that $\mu_\P(\sem{P}_\Gamma^S) = \mu_\Q( \bigcup_{B\mid \pi_S(B) = \sem{P}_\Gamma^S} B ) = 1$. Note that by definition $\bigcup_{B\mid \pi_S(B) = \sem{P}_\Gamma^S} B \subseteq \sem{P}_\Gamma^T$. Since that set has probability 1 and $\Q$ is a complete probability space, then $\sem{P}_\Gamma^T \in \F_\Q$, and also has probability 1.

\end{itemize}
\end{proof}

\subsection{Precise and Convex Assertions}

\begin{lemma}
$\precise{\sure P}$
\end{lemma}
\begin{proof}
Take any $\Gamma$. If $\sem{P}_\Gamma = \emptyset$, then $P$ is unsatisfiable under $\Gamma$, so we are done. If not, then let $\Omega = \mem{\mathsf{fv}(P)}$, $\F = \{ A\subseteq\Omega \mid \sem{P}_\Gamma \subseteq A \} \cup \{ A \subseteq\Omega \mid A \cap \sem{P}_\Gamma = \emptyset \}$ and:
\[
  \mu(A) = \left\{
    \begin{array}{ll}
      1 & \text{if} ~ \sem{P}_\Gamma \subseteq A \\
      0 & \text{otherwise}
    \end{array}\right.
\]
It is relatively easy to see that $\mu$ is a probability measure since $\sem{P}_\Gamma$ is the smallest measurable set with nonzero probability, and it has probability 1, so the countable additivity property holds. By definition, $\Gamma,\tuple{\Omega, \F, \mu}\vDash\sure P$. Clearly it is also minimal, since any other $\P$ such that $\Gamma,\P\vDash \sure P$ must also include $\F$ as measurable sets by definition, and must assign probability 1 to the event $\sem{P}_\Gamma$.
\end{proof}

\begin{lemma}
If $\precise{\varphi,\psi}$, then $\precise{\varphi\sep\psi}$.
\end{lemma}
\begin{proof}
Take any $\Gamma$, if either $\varphi$ or $\psi$ is not satisfiable under $\Gamma$, then neither is $\varphi\sep\psi$, and then the claim holds vacuously. If both are satisfiable, then there are unique smallest $\P_1$ and $\P_2$ such that $\Gamma,\P_1\vDash\varphi$ and $\Gamma,\P_2\vDash\psi$. Clearly, this means that $\Gamma,\P_1\otimes\P_2 \vDash \varphi\sep\psi$. We now argue that $\P_1\otimes\P_2$ is minimal.
Take any $\Q$ such that $\Gamma,\Q\vDash\varphi\sep\psi$. This means that there are $\Q_1 \otimes \Q_2 \preceq \Q$ such that $\Gamma,\Q_1 \vDash\varphi$ and $\Gamma,\Q_2\vDash\psi$. By precision of $\varphi$ and $\psi$, we know that $\P_1 \preceq \Q_1$ and $\P_2 \preceq\Q_2$. Using \Cref{lem:otimes-mono}, we get:
\[
  \P_1 \otimes \P_2 \preceq \Q_1 \otimes \Q_2 \preceq \Q
\]
\end{proof}

\begin{lemma}
If $\precise\varphi$ and $\varphi\Rightarrow \sure{ e \mapsto X }$, then $\precise{ \bigoplus_{X \sim d(E)} \varphi }$.
\end{lemma}
\begin{proof}
Take any $\Gamma$ and let $\nu = d(\de{E}_{\mathsf{LExp}}(\Gamma)$, if $\varphi$ is unsatisfiable under any $\Gamma[X \coloneqq v]$, then so is $\bigoplus_{X \sim d(E)} \varphi $, so the claim holds vacuously. If not, then there is a unique smallest $\P_v$ such that $\Gamma[X \coloneqq v], \P_v\vDash\varphi$ for each $v\in\supp(\nu)$. Since $\varphi\Rightarrow \sure{e \mapsto X}$, we can create new disjoint probability spaces $\P'_v$, where each $\Omega_{\P'_v} =  \{ \sigma\in\Omega_{\P_v} \mid \de{e}_\expr(\sigma) = v \}$. Note that this does not remove any samples that have positive probability, and clearly $\P_v = \comp(\P'_v)$. Let $\P = \bigoplus_{v\sim \nu} \P_v'$, then $\Gamma,\P \vDash \bigoplus_{X \sim d(E)} \varphi$. It only remains to show that $\P$ is minimal.

Take any $\Q$ such that $\Gamma,\Q\vDash\bigoplus_{X \sim d(E)} \varphi$. That means that $\Gamma[X\coloneqq v], \comp(\Q_v)\vDash\varphi$, where $\bigoplus_{v\sim\nu}\Q_v \preceq \Q$. Since $\varphi$ is precise, then $\P_v = \comp(\P'_v) \preceq \comp(\Q_v)$ for each $v$. Since the completion only expands the sample space with zero probability events, this must mean that $\P'_v\preceq\Q_v$ as well. Therefore, by \Cref{lem:oplus-mono}:
\[
  \P = \bigoplus_{v\sim\nu} \P'_v \preceq \bigoplus_{v\sim\nu} \Q_v \preceq \Q
\]
\end{proof}

\begin{lemma}
If $\convex{\varphi_1,\varphi_2}$, then $\convex{\varphi_1\osep\varphi_2}$.
\end{lemma}
\begin{proof}
Take any $\Gamma$, if $\varphi_1\osep\varphi_2$ is satisfiable under $\Gamma$, then so are each $\varphi_k$.
Since each $\varphi_k$ is convex, we know that there exist $\Omega_k$, $\F_k$, and $S_k$ such that $\Gamma,\P\vDash\varphi$ iff $\tuple{\Omega_k, \F_k, \mu}\preceq\P$ for some $\mu \in S_k$.
Let $U$ and $V$ be sets such that $\Omega_1 = \mem U$ and $\Omega_2 = \mem V$, let $\Omega = \Omega_1 \sep \Omega_2$, $\F$ be the smallest $\sigma$-algebra containing $\{ A \sep B \mid A\in \F_1, B\in\F_2 \}$, and $S = \{ \mu \mid \pi_U(\mu) \in S_1, \pi_V(\mu) \in S_2 \}$ ($S$ is clearly convex since $S_1$ and $S_2$ are convex).

We complete the proof by showing that $\Gamma,\P\vDash \varphi_1\osep\varphi_2$ iff $\tuple{\Omega,\F,\mu} \preceq \P$ for some $\mu \in S$. For the forward direction, suppose that $\Gamma,\P\vDash \varphi_1\osep\varphi_2$, so $\Gamma,\pi_U(\P)\vDash \varphi_1$ and $\Gamma,\pi_V(\P)\vDash \varphi_2$. Due to convexity of $\varphi_1$ and $\varphi_2$, there must be $\mu_1 \in S_1$ and $\mu_2\in S_2$ such that $\tuple{\Omega_1, \F_1, \mu_1} \preceq \pi_U(\P)$ and $\tuple{\Omega_2, \F_2, \mu_2} \preceq \pi_V(\P)$. Now, let:
\[
  \mu(A) \triangleq \mu_\P\left(
    \bigcup \left\{
      B \in \F_\P \mid \pi_{U \cup V}(B) = A
    \right\}
  \right)
\]
This gives us:
\[
  \pi_U(\mu)(A)
  = \mu(A \sep \mem V)
  = \mu_\P\left(\bigcup \left\{B \in \F_\P \mid \pi_{U \cup V}(B) = A\sep\mem{V}  \right\}\right)
  = \mu_\P\left(\bigcup \left\{B \in \F_\P \mid \pi_{U}(B) = A \right\}\right)
  = \mu_1(A)
\]
And similarly, $\pi_V(\mu)(B) = \mu_2(B)$, therefore $\mu \in S$ by construction, and so also by construction $\tuple{\Omega, \F, \mu}\preceq \P$.

For the reverse direction, suppose that $\tuple{\Omega, \F, \mu}\preceq \P$ for some $\mu \in S$. That means that $\pi_U(\mu) \in S_1$ and $\pi_V(\mu) \in S_2$, and therefore $\Gamma,\tuple{\Omega_1, \F_1, \pi_U(\mu)}\vDash\varphi_1$ and $\Gamma,\tuple{\Omega_2, \F_2, \pi_V(\mu)}\vDash\varphi_2$. Since $\tuple{\Omega, \F, \mu} \in \tuple{\Omega_1, \F_1, \pi_U(\mu)}\diamond_\wk \tuple{\Omega_2, \F_2, \pi_V(\mu)}$, then $\Gamma,\P\vDash\varphi_1\osep\varphi_2$.
\end{proof}

\begin{lemma}
If $\convex{\varphi_1,\varphi_2}$, $\varphi_1\Rightarrow \sure{e\mapsto 1}$, and $\varphi_2\Rightarrow\sure{e\mapsto 0}$, then $\convex{\varphi_1\oplus_{\ge E}\varphi_2}$.
\end{lemma}
\begin{proof}
Take any $\Gamma$, if $\varphi_1\oplus_{\ge E}\varphi_2$ is satisfiable under $\Gamma$, then so are $\varphi_1$ and $\varphi_2$.
Let $p = \de{E}_{\mathsf{LExp}}(\Gamma)$ and let $X$ be the variable that is bound by $\oplus_{\ge E}$.
Since the $\varphi_k$ are convex, then for each $k\in\{1,2\}$ there exist $\Omega_k$, $\F_k$, and $S_k$ such that $\Gamma[X\coloneqq 2-k], \P\vDash\varphi_k$ iff $\tuple{\Omega_k, \F_k, \mu} \preceq \P$ for some $\mu \in S_k$. Now, let:
\begin{mathpar}
\Omega \triangleq \Omega_1 \cup \Omega_2

\F \triangleq \left\{ A \subseteq \Omega \mid \{ \sigma\in A \mid \de{e}_\expr(\sigma) = 1 \} \in \F_1, \{ \sigma\in A \mid \de{e}_\expr(\sigma) = 0 \} \in \F_2 \right\}

S \triangleq \{ \mu \oplus_q \nu \mid \mu\in S_1, \nu \in S_2, p \le q \le 1 \}

\text{where}\quad
(\mu \oplus_q \nu)(A) = q\cdot \mu(\{ \sigma\in A \mid \de{e}_\expr(\sigma) = 1 \}) + (1-q) \cdot \nu(\{ \sigma\in A \mid \de{e}_\expr(\sigma) = 0 \})
\end{mathpar}
We complete the proof by showing that $\Gamma,\P\vDash \varphi_1\oplus_{\ge E}\varphi_2$ iff $\tuple{\Omega,\F,\mu} \preceq \P$ for some $\mu \in S$. For the forward direction, suppose that $\Gamma,\P\vDash \varphi_1\oplus_{\ge E}\varphi_2$, so there exists a $q \ge p$ such that $\Gamma[X\coloneqq 1], \comp(\P_1) \vDash\varphi_1$ and $\Gamma[X\coloneqq 0], \comp(\P_2) \vDash\varphi_2$ for some $\P_1$ and $\P_2$ such that $\P_1 \oplus_q \P_2 \preceq \P$. This means that $\tuple{\Omega_k, \F_k, \mu_k} \preceq \comp(P_k)$ for some $\mu_k \in S_k$ for each $k\in\{1,2\}$. So, letting $\mu = \mu_1 \oplus_q \mu_2$, clearly $\mu\in S$ by construction. Finally, we get:
\[
  \tuple{\Omega, \F, \mu}
  = \tuple{\Omega_1, \F_1, \mu_1} \oplus_q \tuple{\Omega_2, \F_2, \mu_2}
  \preceq \P_1 \oplus_q \P_2 
  \preceq \P
\]
For the reverse direction, suppose that $\tuple{\Omega,\F,\mu} \preceq \P$ for some $\mu\in S$.
Since $\mu \in S$, then there exist $\mu_1\in S_1$, $\mu_2\in S_2$, and $q \ge p$ such that $\mu = \mu_1 \oplus_q \mu_2$. We know by construction that $\Gamma[X\coloneqq 2-k],\tuple{\Omega_k, \F_k, \mu_k}\vDash\varphi_k$ for each $k$, therefore clearly $\Gamma,\tuple{\Omega,\F,\mu}\vDash\varphi_1\oplus_q\varphi_2$. Therefore, by \Cref{lem:mono-sat}, we know that $\Gamma,\P\vDash\varphi_1\oplus_q\varphi_2$. Finally, we weaken this assertion to get $\Gamma,\P\vDash\varphi_1\oplus_{\ge E}\varphi_2$.
\end{proof}

\begin{lemma}
If $\convex{\varphi}$ and $\varphi\Rightarrow\sure{e\mapsto X}$, then $\convex{\bignd_{X\in E}\varphi}$.
\end{lemma}
\begin{proof}
Take any $\Gamma$, if $\bigoplus_{X\in E}\varphi$ is satisfiable under $\Gamma$, then so is $\varphi$. Let $S = \de{E}_{\mathsf{LExp}}(\Gamma)$, so for each $v\in S$, there must be $\Omega_v$, $\F_v$, and $S_v$ such that $\Gamma[X\coloneqq v],\P\vDash \varphi$ iff $\tuple{\Omega_v, \F_v, \mu} \preceq \P$ for some $\mu \in S_v$.
Now let:
\begin{mathpar}
  \Omega \triangleq \bigcup_{v \in S} \Omega_v
  
  \F \triangleq \left\{ A \subseteq \Omega \mid \forall v\in S.\ \{ \sigma \in A \mid \de{e}_\expr(\sigma) = v \} \in \F_v \right\}
  
  T \triangleq \left\{ \sum_{v\in S} \xi(v)\cdot \mu_v \mid \xi \in \D(S), \forall v.\ \mu_v \in S_v \right\}
\end{mathpar}
Clearly $T$ is convex, since it is constructed as countable convex combinations of convex sets \cite[Lemma B.1]{zilberstein2025demonic}.
We now show that $\Gamma,\P\vDash\bignd_{X\in E} \varphi$ iff $\tuple{\Omega,\F, \mu} \preceq \P$ for some $\mu \in T$. For the forward direction, suppose that $\Gamma,\P\vDash\bignd_{X\in E}\varphi$, so $\Gamma[X\coloneqq v], \comp(\P_v)\vDash\varphi$ for each $v\in \supp(\nu)$ where $\nu \in \D(S)$ and $\bigoplus_{v\sim\nu}\P_v \preceq\P$. This means that for each $v$, there is a $\mu_v$ such that $\tuple{\Omega_v, \F_v, \mu_v} \preceq \comp(\P_v)$. So, letting $\mu = \sum_{v\in S} \nu(v) \cdot \mu_v$, we get that $\tuple{\Omega,\F, \mu} \preceq \P$ by construction.

For the reverse direction, suppose that $\tuple{\Omega,\F, \mu} \preceq \P$ for some $\mu\in T$. Since $\mu\in T$, then there exists a $\xi \in \D(S)$ and $\mu_v \in S_v$ for all $v\in S$ such that $\mu = \sum_{v\in S}\xi(v)\cdot \mu_v$. By definition, $\Gamma[X\coloneqq v],\tuple{\Omega_v, \F_v, \mu_v} \vDash\varphi$. Let $\P_v$ be the restriction of $\tuple{\Omega_v, \F_v, \mu_v}$ to the states where $\de{e}_\expr(\sigma) = v$, then in order to show that $\Gamma,\P\vDash\bignd_{X\in E}\varphi$, it will suffice to show that $\bigoplus_{v\sim\xi}\P_v \preceq \tuple{\Omega,\F, \mu}$, and then we can conclude that $\bigoplus_{v\sim\xi}\P_v \preceq \P$ by transitivity and complete the proof by \Cref{lem:mono-sat}. The conditions on $\Omega$ and $\F$ hold by construction. For the probability measure, we have:
\begin{align*}
  \P_{\bigoplus_{v\sim\xi}\P_v}(A)
  &= \sum_{v\in S} \xi(v) \cdot \mu_{\P_v}(A \cap \Omega_v)
  \\
  &= \sum_{v\in S} \xi(v) \cdot \mu_v(A\cap\Omega_v)
  \intertext{Since each $\mu_v$ assigns 0 probability outside of $\Omega_v$, we can remove the intersection.}
  &= \sum_{v\in S} \xi(v) \cdot \mu_v(A)
  \\
  &= \mu(A)
\end{align*}
\end{proof}

\subsection{Entailment Rules}

\begin{lemma}
The entailment rules in \Cref{fig:entailment} are valid.
\end{lemma}
\begin{proof}
\;

\begin{enumerate}[leftmargin=*]

\item $\inferrule{P \vdash Q}{\sure P\vdash \sure{Q}}$

Suppose that $\Gamma,\P\vDash \sure P$, where $\Omega_\P = \mem S$. That means that $\sem{P}_\Gamma^S \in \F_\P$ and $\mu(\sem{P}_\Gamma^S) = 1$. Since $P\vdash Q$, then it must be that $\sem{P}_\Gamma^S \subseteq \sem{Q}_\Gamma^S$, therefore $\sem{Q}_\Gamma^S \in\F_\P$ because $\P$ is a complete probability space and all samples outside of $\sem{P}_\Gamma^S$ must have measure 0. By the additivity property of probability measures $\mu(\sem{Q}_\Gamma^S) = 1$.

\medskip
\item $\inferrule{\varphi\vdash\varphi' \\ \psi\vdash\psi'}{\varphi\sep_m\psi \vdash\varphi'\sep_m\psi'}$

Suppose that $\Gamma,\P\vDash\varphi\sep_m\psi$, so $\Gamma,\P_1\vDash\varphi$ and $\Gamma,\P_2\vDash\psi$ and $\P' \preceq \P$ for some $\P_1$, $\P_2$, and $\P' \in \P_1 \diamond_m \P_2$. Since $\varphi\vdash\varphi'$ and $\psi\vdash\psi'$, then $\Gamma,\P_1 \vDash\varphi'$ and $\Gamma,\P_2\vDash\psi'$, therefore we immediately conclude that $\Gamma,\P\vDash\varphi'\sep_m\psi'$.

\medskip
\item \label{item:wk-sep}$\varphi\sep\psi \vdash \varphi\osep\psi$

Suppose that $\Gamma,\P\vDash\varphi\sep\psi$, so $\Gamma,\P_1\vDash\varphi$ and $\Gamma,\P_2\vDash\psi$ for some $\P_1$ and $\P_2$ such that $\P_1 \otimes \P_2 \preceq \P$. Clearly, $\P_1 \otimes \P_2 \in \P_1 \diamond_\wk \P_2$, therefore this immediately implies that $\Gamma,\P\vDash\varphi\osep \psi$.

\medskip
\item $\sure{P\sep Q} \dashv\vdash \sure P \sep_m \sure Q$

We first show that $\sure{P\sep Q} \vdash \sure P \sep_m \sure Q$. Suppose that $\Gamma,\P\vDash \sure{P\sep Q}$, where $\Omega_\P = \mem S$. That means that $\sem{P\sep Q}^S_\Gamma \in \F_\P$ and $\mu_\P(\sem{P\sep Q}^S_\Gamma) = 1$. 
Now let $\P_1$ and $\P_2$ be the smallest probability spaces such that $\Gamma,\P_1\vDash \sure P$ and $\Gamma,\P_2\vDash\sure Q$, so clearly $\P_1 \diamond_\wk \P_2 = \P_1\diamond_\st \P_2 = \{ \P_1 \otimes \P_2 \}$ and therefore $\Gamma,\P_1 \otimes \P_2 \vDash \sure{P}\sep_m\sure{Q}$. It is also clearly the case that $\P_1 \otimes \P_2 \preceq \P$ since the smallest event with nonzero measure in $\P_1 \otimes \P_2$ is $\sem{P}_\Gamma \sep \sem{Q}_\Gamma$, which has probability 1, therefore everything larger also has probability 1, and it suffices to show that:
\[
  \mu_{\P_1\otimes \P_2}(\sem{P}_\Gamma \sep \sem{Q}_\Gamma)
  = 1
  = \mu_\P\left(\sem{P\sep Q}^S_\Gamma\right)
  = \mu_\P\left(\textstyle\bigcup_{B\mid \pi_{\mathsf{fv}(P,Q)}(B) = \sem{P\sep Q}^S_\Gamma} B\right)
\]
Therefore, by \Cref{lem:mono-sat}, $\Gamma,\P\vDash \sure{P}\sep_m\sure{Q}$.

Now we show that $\sure P \sep_m \sure Q\vdash \sure{P\sep Q}$, suppose that $\Gamma,\P\vDash\sure P \sep_m \sure Q$. That means that $\P' \preceq \P$ and $\Gamma,\P_1\vDash\sure P$ and $\Gamma,\P_2 \vDash \sure Q$ for some $\P_1$, $\P_2$, and $\P'\in\P_1\diamond_m\P_2$. This also means that $\mu_{\P_1}(\sem{P}^S_\Gamma) = 1$ and $\mu_{\P_2}(\sem{Q}^T_\Gamma) = 1$ (where $\Omega_{\P_1} = \mem S$ and $\Omega_{\P_2}= \mem T$). Therefore, we know that $\pi_S(\mu_{\P'})(\sem{P}_\Gamma^S) = 1$ and $\pi_T(\mu_{\P'})(\sem{Q}_\Gamma^T) = 1$, and therefore we know that anything outside of $P$ and $Q$ has measure zero, and therefore it must be the case that $\mu_{\P'}(\sem{P\sep Q}_\Gamma^{S\cup T}) = 1$.
So, $\Gamma,\P'\vDash \sem{P\sep Q}$, and since $\P'\preceq \P$, then by \Cref{lem:mono-sat}, $\Gamma,\P\vDash \sem{P\sep Q}$.
%

\medskip
\item $\varphi\sep\sure P \dashv\vdash \varphi\osep\sure P$

The forward direction follows immediately from \Cref{item:wk-sep}, so we prove only the reverse direction. Suppose that $\Gamma,\P\vDash \varphi\osep \sure{P}$. This means that $\Gamma,\P_1\vDash\varphi$ and $\Gamma,\P_2\vDash\sure{P}$ for some $\P_1$, $\P_2$, and $\P' \in \P_1 \diamond_\wk \P_2$. Now let $\P_2'$ be the smallest probability space such that $\Gamma,\P_2' \vDash \sure P$. 
Since $\P_2'$ contains events only of measure 0 or 1, then $\P_1 \diamond_\wk \P_2' = \{ \P_1 \otimes \P_2' \}$, therefore it must be that $\P_1 \otimes \P_2' \preceq \P' \preceq \P$. Therefore, by definition, $\Gamma,\P\vDash\varphi\sep\sure P$.

\medskip
\item $\bigoplus_{X\sim d(E)} \varphi \vdash \bignd_{X\in\supp(d(E))} \varphi$

Suppose that $\Gamma,\P\vDash \bigoplus_{X\sim d(E)}\varphi$, therefore $\Gamma[X\coloneqq v],\comp(\P_v)\vDash\varphi$ for each $v\in\supp(\xi)$ where $\xi = d(\de{E}_{\mathsf{LExp}}(\Gamma))$. Obviously $\xi \in \D(\supp(\xi))$, so this immediately implies that $\Gamma,\P\vDash\bignd_{\supp(d(E))}\varphi$.

\medskip
\item $\varphi[E/X]\dashv\vdash\bignd_{X\in\{E\}}\varphi$

For the forward direction, suppose that $\Gamma,\P\vDash\varphi[E/X]$. Let $v = \de{E}_{\mathsf{LExp}}(\Gamma)$. We therefore have that $\Gamma[X \coloneqq v]\vDash\varphi$. Obviously, $\delta_v \in \D(\{ v \})$, therefore we get that $\Gamma,\P\vDash \bignd_{X \in \{E\}}\varphi$.

For the reverse direction, suppose that $\Gamma,\P\vDash\bignd_{X\in \{E\}}\varphi$. Since the only distribution over a singleton support is the point-mass distribution, this immediately gives us $\Gamma[X \coloneqq v]\vDash\varphi$, where again $v = \de{E}_{\mathsf{LExp}}(\Gamma)$. Finally, we conclude that $\Gamma,\P\vDash\varphi[E/X]$.

\medskip
\item $\sure{E\subseteq E'} \sep \bignd_{X\in E}\varphi \vdash \bignd_{X\in E'}\varphi$

Suppose that $\Gamma,\P\vDash \sure{E\subseteq E'} \sep \bignd_{X\in E}\varphi$. Let $S = \de{E}_{\mathsf{LExp}}(\Gamma)$ and $S' = \de{E'}_{\mathsf{LExp}}(\Gamma)$. We therefore know that $S \subseteq S'$ and $\Gamma[X \coloneqq v], \comp(\P_v)\vDash\varphi$ for all $v \in \supp(\xi)$ and some $\xi \in \D(S)$ such that $\bigoplus_{v\sim \xi}\P_v \preceq \P$. Obviously, it is also the case that $\xi \in \D(S')$, since $S \subseteq S'$, therefore we immediately have that $\Gamma,\P\vDash\bignd_{X\in E'}\varphi$.

\medskip
\item $\inferrule{\varphi \vdash \psi}{\bigoplus_{X\sim d(E)} \varphi \vdash \bigoplus_{X\sim d(E)}\psi}$

Suppose that $\Gamma,\P\vDash \bigoplus_{X\sim d(E)} \varphi$ and let $\nu = d(\de{E}_{\mathsf{LExp}}(\Gamma)$. So, $\bigoplus_{v\sim\nu} \P_v \preceq \P$ such that $\Gamma[X\coloneqq v], \comp(\P_v) \vDash \varphi$ for each $v$. Since $\varphi \vdash \psi$, we get that $\Gamma[X\coloneqq v], \comp(\P_v) \vDash \psi$ for each $v$. Therefore, $\Gamma,\P\vDash \bigoplus_{X\sim d(E)}\psi$. Note that $\psi$ may not witness a partition of the sample space, but the $\P_v$s are still disjoint.

\medskip
\item $\inferrule{Y \notin\mathsf{fv}(\varphi)}{ \bigoplus_{X\sim d(E)} \varphi \vdash \bigoplus_{Y\sim d(E)} \varphi[Y/X]}$

Suppose that $\Gamma,\P\vDash \bigoplus_{X\sim d(E)} \varphi$, and let $\nu = d(\de{E}_{\mathsf{LExp}}(\Gamma))$. This means that $\bigoplus_{v\sim \nu} \P_v \preceq \P$ such that $\Gamma[X \coloneqq v],  \comp(\P_v) \vDash \varphi$ for each $v \in \supp(\nu)$. Since $Y \notin\mathsf{fv}(\varphi)$, then clearly $\Gamma[Y \coloneqq v],\comp(\P_v) \vDash \varphi[Y/X]$. Therefore, we get that $\Gamma,\bigoplus_{v\sim\nu} \P_v \vDash \bigoplus_{Y\sim d(E)} \varphi[Y/X]$, and since $\bigoplus_{v\sim \nu} \P_v \preceq \P$, then $\Gamma,\P \vDash \bigoplus_{Y\sim d(E)} \varphi[Y/X]$ by \Cref{lem:mono-sat}.

\medskip
\item $\inferrule{X\notin\mathsf{fv}(\psi)}{(\smashoperator{\bigoplus_{X\sim d(E)}} \varphi) \sep \psi \vdash \smashoperator{\bigoplus_{X\sim d(E)}} (\varphi\sep \psi)}$

Suppose that $\Gamma,\P\vDash (\bigoplus_{X\sim d(E)} \varphi) \sep \psi$ and let $\nu = d(\de{E}_{\mathsf{LExp}}(\Gamma))$. So, $(\bigoplus_{v\sim\nu} \P_v) \otimes \P' \preceq \P$ such that $\Gamma[X\coloneqq v],\comp(\P_v) \vDash \varphi$ for each $v$ and $\Gamma,\P' \vDash\psi$. Since $X\notin\mathsf{fv}(\psi)$, then $\Gamma[X\coloneqq v],\P' \vDash\psi$ for each $v$. Also note that $\comp(\P_v) \otimes \P' = \comp(\P_v\otimes \P')$ since $\P'$ is already a complete probability space. This gives us $\Gamma[X\coloneqq v],\comp(\P_v\otimes \P')\vDash \varphi\sep \psi$. Now, by \Cref{lem:otimes-oplus-dist}, we have:
\[
  \bigoplus_{v\sim\nu} (\P_v \otimes \P')
  = \left( \bigoplus_{v\sim\nu} \P_v \right) \otimes \P'
  \preceq \P
\]
Therefore, $\Gamma,\P \vDash \bigoplus_{X\sim d(E)} (\varphi\sep\psi)$.

\medskip
\item $\inferrule{X\notin\mathsf{fv}(\psi)  \\  \precise\psi}{\bigoplus_{X\sim d(E)} (\varphi\sep \psi) \vdash (\bigoplus_{X\sim d(E)} \varphi) \sep \psi}$

Suppose that $\Gamma,\P\vDash \bigoplus_{X\sim d(E)} (\varphi\sep \psi)$, and let $\nu = d(\de{E}_{\mathsf{LExp}}(\Gamma))$. This means that $\bigoplus_{v\sim\nu}(\P_v \otimes \Q_v) \preceq \P$ such that $\Gamma[X\coloneqq v], \comp(\P_v) \vDash\varphi$ and $\Gamma[X\coloneqq v], \comp(\Q_v) \vDash\psi$. Since $X\notin\mathsf{fv}(\psi)$, we also know that $\Gamma, \comp(\Q_v) \vDash\psi$, and since $\psi$ is precise, there is a unique $\Q$ such that $\Gamma,\Q\vDash\psi$ and $\Q \preceq \comp(\Q_v)$ for all $v$. Therefore, by recombining the components, we get that $\Gamma,\P\vDash (\bigoplus_{X\sim d(E)} \varphi) \sep \psi$.

\medskip
\item $\inferrule{X\notin\mathsf{fv}(\psi)  \\  \convex{\psi}}{\bigoplus_{X\sim d(E)} (\varphi\osep \psi) \vdash (\bigoplus_{X\sim d(E)} \varphi) \osep \psi}$

Suppose that $\Gamma,\P\vDash\bigoplus_{X\sim d(E)} (\varphi\osep \psi)$, which means that $\Gamma[X\coloneqq v],\comp(\P_v)\vDash\varphi\osep \psi$ for each $v\in\supp(\xi)$ where $\xi = d(\de{E}_{\mathsf{LExp}}(\Gamma))$ and $\bigoplus_{v\sim\xi}\P_v \preceq \P$. Let $S$ be the set such that $\Omega_\P = \mem S$ and $T = \free(\psi)$, so $\Gamma[X\coloneqq v], \pi_{S\setminus T}(\comp(\P_v))\vDash\varphi$, and therefore $\Gamma,\pi_{S\setminus T}(\P)\vDash\varphi$. In addition, we know that $\Gamma[X\coloneqq v],\pi_T(\comp(\P_v))\vDash \psi$ for each $v$, and since $\psi$ is convex, then $\Gamma,\pi_T(\P)\vDash\psi$ (see \Cref{item:convex-idem} for more details). Combining these two facts, we get that $\Gamma,\P\vDash (\bigoplus_{X\sim d(E)} \varphi) \osep \psi$.

\medskip
\item\label{item:convex-idem} $\inferrule{X\notin\mathsf{fv}(\varphi) \\ \convex\varphi}{\bigoplus _{X \sim d(E)} \varphi \vdash \varphi}$

Suppose that $\Gamma,\P\vDash \bigoplus _{X \sim d(E)} \varphi$, and let $\nu = d(\de{E}_{\mathsf{LExp}}(\Gamma))$. This means that $\bigoplus_{v\sim\nu}\P_v\preceq \P$ and $\Gamma[X\coloneqq v],\comp(\P_v)\vDash \varphi$. Since $X\notin\mathsf{fv}(\varphi)$,this also means that $\Gamma,\comp(\P_v)\vDash \varphi$. Since $\varphi$ is convex, there exist $\Omega$, $\F$, and $S$ and for each $v$ $\mu_v \in S$ such that $\Gamma,\tuple{\Omega,\F, \mu_v}\vDash \varphi$ and $\tuple{\Omega,\F, \mu_v} \preceq \comp(\P_v)$. Now, let $\mu = \sum_{x\in\supp(\nu)} \nu(v) \cdot \mu_v$, so clearly $\mu \in S$ since it is a convex combination of elements of $S$, and therefore, $\Gamma,\tuple{\Omega, \F, \mu}\vDash \varphi$.
It remains only to show that $\tuple{\Omega,\F, \mu} \preceq \P$. The conditions on $\Omega$ and $\F$ hold trivially since $\tuple{\Omega,\F, \mu_v} \preceq \comp(\P_v)$ for all $v$. Now, we show the condition on $\mu_\P$:
\begin{align*}
   \mu(A)
   &= \sum_{v\in \supp(\nu)}\nu(v)\cdot\mu_v(A)
   \intertext{Let $U$ be the set such that $\Omega = \mem U$. Since $\tuple{\Omega,\F, \mu_v} \preceq \comp(\P_v)$.}
   &= \sum_{v\in \supp(\nu)}\nu(v)\cdot\mu_{\P_v}\left( \bigcup \left\{ B \in \F_{\P_v} \mid \pi_U(B) = A \right\} \right)
   \intertext{Since $\tuple{\Omega,\F, \mu_v} \preceq \comp(\P_v)$, then $\mu_v$ gives measure 0 to all events outside of $\Omega_{\P_v}$, therefore we can add the following intersection.}
   &= \sum_{v\in \supp(\nu)}\nu(v)\cdot\mu_{\P_v}\left( \bigcup \left\{ B \in \F_{\P_v} \mid \pi_U(B) = A \right\} \cap \Omega_{\P_v}\right)
   \\
   &= \mu_{\bigoplus_{v\sim\nu} \P_v}\left( \bigcup \left\{ B \in \F_{\P} \mid \pi_U(B) = A \right\} \right)
   \intertext{Since $\bigoplus_{v\sim\nu}\P_v\preceq \P$.}
   &= \mu_{\P}\left( \bigcup \left\{ B \in \F_{\P} \mid \pi_U(B) = A \right\} \right)
\end{align*}

\end{enumerate}
\end{proof}

\subsection{Soundness of Inference Rules}

We start by providing a lemma stating that weak triples can be stated without the frame preservation property.

\begin{lemma}[Alternative Characterization of Weak Triples]
\label{lem:weak-triple}
\[
  I\vDash_\wk\triple\varphi{C}\psi
  \qquad\text{iff}\qquad
  \forall \Gamma,\mu.\
  \Gamma,\mu\vDash\varphi\sep\sure I \implies \forall \nu\in\lin^{\sem{I}_\Gamma}(\de{C})^\dagger(\mu).\
  \Gamma,\nu\vDash\psi\sep\sure I
\]
\end{lemma}
\begin{proof}
We show both directions:
\iffcases{
  Suppose that $\Gamma,\mu\vDash\varphi\sep\sure I$.
  Let $\P_F$ be the trivial probability space on $\mem\emptyset$, so clearly $\mu \diamond_\wk \P_F = \{ \mu \}$.
  Now take any $\nu \in \lin^{\sem{I}_\Gamma}(\de{C})^\dagger(\mu)$. Since $I\vDash_\wk\triple\varphi{C}\psi$, we get that there exists a $\Q$ and $\Q'\in \Q\diamond_\wk \P_F$ such that $\Q' \preceq \nu$ and $\Gamma,\Q \vDash \psi\sep\sure I$. Since $\Q\diamond_\wk \P_F = \{\Q\}$, then $\Q' = \Q$, therefore $\Gamma, \Q'\vDash\psi\sep\sure I$. In addition, since $\Q'\preceq \nu$, then by \Cref{lem:mono-sat}, $\Gamma,\nu\vDash\psi \sep\sure I$.
}{
  Suppose that $\P' \in \P\diamond_\wk \P_F$ such that $\P' \preceq\mu$ and $\Gamma,\P\vDash\varphi\sep\sure I$.  
  Let $U \subseteq\mathsf{Var}$ be the variables of $\P$ and $I$ and $V$ be the variables of $\P_F$. Clearly $\Gamma,\pi_U(\mu)\vDash\varphi\sep\sure I$. Since $\pi_U(\mu) \preceq \mu$, then by \Cref{lem:mono-sat}, $\Gamma,\mu\vDash\varphi\sep\sure I$ too. So, by the premise, $\Gamma,\nu\vDash\psi\sep\sure I$ for any $\nu \in \lin^{\sem{I}_\Gamma}(\de{C})^\dagger(\mu)$. But $\psi$ and $I$ only depend on the variables $U$, so $\Gamma,\pi_U(\nu)\vDash\psi\sep\sure I$ too. In addition, it must be the case that $\P_F\preceq \pi_V(\nu)$ since $C$ does not alter any variables in $V$. Now, let $\Q = \pi_{U}(\nu)$ and construct $\Q'$ as follows:
  \[
    \Omega_{\Q'}\triangleq \Omega_{\P'}
    \qquad
    \F_{\Q'} \triangleq \sigma(\{ A _1 \sep A_2 \mid A_1 \in\F_\Q, A_2 \in\F_{\P_F} \})
    \qquad
    \mu_{\Q'}(A) = \nu(A)
\]
So clearly by construction, $\Q' \in \Q  \diamond_\wk \P_F$ and $\Q' \preceq \nu$ and $\Gamma,\Q\vDash\psi\sep\sure{I}$. Therefore, we are done.
}
\end{proof}

\begin{figure}
\begin{mathpar}
%
\ruledef{Disj}{
  I\vdash_m\triple{\varphi_1}C{\psi_1}
  \\
  I\vdash_m\triple{\varphi_2}C{\psi_2}
}{
    I\vdash_m\triple{\varphi_1\vee\varphi_2}C{\psi_1\vee\psi_2}
}

\ruledef{Exists2}{
  I\vdash_m\triple{\varphi}C{\psi}
  \\
  X\notin\free(\psi,I)
}{
    I\vdash_m\triple{\exists X.\ \varphi}C{\psi}
}

\ruledef{Subst}{
  I\vdash_m\triple{\varphi}C{\psi}
}{
  I[E/X]\vdash_m\triple{\varphi[E/X]}C{\psi[E/X]}
}
\end{mathpar}
\caption{Additional Inference Rules}
\label{fig:add-rules}
\end{figure}

\soundnessthm*
\begin{proof}
The proof is by induction on the derivation.
\begin{itemize}[leftmargin=*]
\item\ruleref{Skip}.
\[
  \inferrule{\;}{I\vdash_m \triple\varphi\skp\varphi}{\ruleref{Skip}}
\]
Suppose that $\P' \in \P\diamond_m \P_F$, $\P' \preceq \mu$, and $\Gamma,\P \vDash \varphi\sep\sure{I}$. By \citet[Lemma 5.2]{zilberstein2025denotational}, we have:
\[
  \lin^{\sem{I}_\Gamma}(\de{\skp})^\dagger(\mu) = \eta^\dagger(\mu) = \{\mu\}
\]
So, letting $\Q = \P$ and $\Q' = \P'$, clearly $\Q' \in \Q\diamond_m \P_F$ and $\Q' \preceq\mu$ and $\Gamma,\Q\vDash\varphi\sep\sure I$.

\item\ruleref{Seq}.
\[
  \inferrule{
    I\vdash_m\triple\varphi{C_1}{\vartheta}
    \\
    I\vdash_m\triple\vartheta{C_1}\psi
  }{
    I\vdash_m\triple\varphi{C_1\fatsemi C_2}\psi
  }{\ruleref{Seq}}
\]
Suppose that $\P' \in \P\diamond_m \P_F$, $\P' \preceq \mu$, and $\Gamma,\P \vDash \varphi\sep\sure I$. By \citet[Lemma 5.2]{zilberstein2025denotational}, we have:
\begin{align*}
  \lin^{\sem{I}_\Gamma}(\de{C_1\fatsemi C_2})^\dagger(\mu)
  &= \lin^{\sem{I}_\Gamma}(\de{C_1}\fatsemi \de{C_2})^\dagger(\mu)
  \\
  &= \lin^{\sem{I}_\Gamma}(\de{C_2})^\dagger(\lin^{\sem{I}_\Gamma}(\de{C_1}^\dagger(\mu))
  \\
  &= \smashoperator{\bigcup_{\mu'\in \lin^{\sem{I}_\Gamma}(\de{C_1})^\dagger(\mu)}} \lin^{\sem{I}_\Gamma}(\de{C_2})^\dagger(\mu')
\end{align*}
So, for every $\nu \in \lin^{\sem{I}_\Gamma}(\de{C_1\fatsemi C_2})^\dagger(\mu)$, there is a $\mu' \in\lin^{\sem{I}_\Gamma}(\de{C_1})^\dagger(\mu)$ such that $\nu \in \lin^{\sem{I}_\Gamma}(\de{C_2})^\dagger(\mu')$. By the induction hypotheses, we know there exist $\Q_1$ and $\Q_1' \in \Q_1 \diamond_m \P_F$ such that $\Q_1' \preceq \mu'$ and $\Gamma,\Q_1 \vDash\vartheta\sep\sure I$. By the induction hypothesis again, we get that there exist $\Q$ and $\Q' \in \Q \diamond_m \P_F$ such that $\Q' \preceq \nu$ and $\Gamma,\Q\vDash\psi\sep\sure I$.

\item\ruleref{IfT}.
\[
  \inferrule{
    \varphi \Rightarrow \sure {b \mapsto \tru}
    \\
    I\vdash_m\triple\varphi{C_1}\psi
  }{
    I\vdash_m\triple\varphi{\iftf b{C_1}{C_2}}{\psi}
  }{\ruleref{IfT}}
\]
Suppose that $\P' \in \P\diamond_m \P_F$, $\P' \preceq \mu$, and $\Gamma,\P \vDash \varphi\sep\sure I$. We therefore know that $\de{b}_\test(\sigma) = \tru$ for all $\sigma\in\supp(\mu)$. So, by \citet[Lemma 5.2]{zilberstein2025denotational}, we get that:
\begin{align*}
  \lin^{\sem{I}_\Gamma}(\de{\iftf b{C_1}{C_2}})^\dagger(\mu)
  &=  \lin^{\sem{I}_\Gamma}(\mathsf{guard}(b, \de{C_1}, \de{C_2}))^\dagger(\mu)
  \\
  &= \left(\lambda\sigma.\left\{
    \begin{array}{ll}
      \lin^{\sem{I}_\Gamma}(\de{C_1})(\sigma) & \text{if}~ \de{b}_\test(\sigma) = \tru \\
      \lin^{\sem{I}_\Gamma}(\de{C_2})(\sigma) & \text{if}~ \de{b}_\test(\sigma) = \fls
    \end{array}
    \right.\right)^\dagger(\mu)
    \\
    &= \lin^{\sem{I}_\Gamma}(\de{C_1})^\dagger(\mu)
\end{align*}
So any $\nu \in \lin^{\sem{I}_\Gamma}(\de{\iftf b{C_1}{C_2}})^\dagger(\mu)$ must also be in $\lin^{\sem{I}_\Gamma}(\de{C_1})^\dagger(\mu)$, and therefore we can use the induction hypothesis to conclude that there is a $\Q$ and $\Q' \in \Q\diamond_m \P_F$ such that $\Q' \preceq \nu$ and $\Gamma,\Q\vDash\psi\sep\sure I$.

\item\ruleref{IfF}. Symmetric to the \ruleref{IfT} case.

\item\ruleref{Assign}.
\[
  \inferrule{
    \varphi \Rightarrow \sure{e\mapsto E} \land (\psi \sep \sure{\own(x)})
  }{
    I\vdash_m\triple{\varphi}{x \coloneqq e}{\psi\sep \sure{x \mapsto E}}
  }{\ruleref{Assign}}
\]
We prove only the $m=\st$ case, as the $m=\wk$ case follows from the soundness of the \ruleref{Weakening} rule.
Suppose that $\P\otimes \P_F \preceq\mu$ and $\Gamma,\P\vDash \varphi\sep\sure{x\mapsto E}\sep\sure I$. Let $S$ be the set of variables such that $\mu \in \D(\mem S)$.
Since $\varphi \Rightarrow \sure{e \mapsto E}$, we know that $\de{e}_\expr(\sigma) = \de{E}_{\mathsf{LExp}}(\Gamma)$ for all $\sigma\in\supp(\mu)$. Now take any $\nu \in \lin^{\sem{I}_\Gamma}(\de{x \coloneqq e})^\dagger(\mu)$. We know from \citet[Lemma 5.2]{zilberstein2025denotational} that:
\begin{align*}
  \lin^{\sem{I}_\Gamma}(\de{x \coloneqq e})^\dagger(\mu)
  &= \lin^{\sem{I}_\Gamma}(\singleton{x \coloneqq e})^\dagger(\mu)
  \\
  &= \left( \de{x \coloneqq e}^{\sem{I}_\Gamma}_\act \right)^\dagger(\mu)
  \intertext{From the premise of the rule, we know that $x\coloneqq e$ does not depend on $\free(I)$, so the invariant sensitive semantics is the same as the regular semantics.}
  &= \de{x \coloneqq e}_\act^\dagger(\mu)
  \\
  &= \left\{
    \smashoperator[r]{\sum_{\sigma \in \supp(\mu)}} \mu(\sigma) \cdot \nu_\sigma
    \;\;\Big|\;\;
    \forall \sigma.\
    \nu_\sigma \in \de{x\coloneqq e}_\act(\sigma)
  \right\}
  \\
  &= \left\{
    \smashoperator[r]{\sum_{\sigma \in \supp(\mu)}} \mu(\sigma) \cdot \delta_{\sigma[ x \coloneqq \de{e}_\expr(\sigma) ]}
  \right\}
  \\
  &= \left\{
    \smashoperator[r]{\sum_{\sigma \in \supp(\mu)}} \mu(\sigma) \cdot \delta_{\sigma[ x \coloneqq \de{E}_{\mathsf{LExp}}(\Gamma) ]}
  \right\}
\end{align*}
Now, since $\varphi\Rightarrow \psi\sep\sure{\own(x)}$, we get that $\Gamma, \P \vDash \psi \sep \sure{\own(x)} \sep \sure{I}$, so there exist $\P_1$, $\P_2$, and $\P_3$ such that $\P_1\otimes\P_2\otimes \P_3\preceq \P$ and $\Gamma,\P_1\vDash \psi$ and $\Gamma,\P_2\vDash\sure{\own(x)}$ and $\Gamma,\P_3\vDash\sure I$. Since $\sure{\own(x)}$ is precise, let $\P_2'$ be the unique smallest probability space such that $\Gamma,\P_2'\vDash\sure{\own(x)}$, and note that $\P_2'\preceq\P_2$ and $\Omega_{\P_2'} = \mem{\{x\}}$. So, by \Cref{lem:otimes-mono}, we have:
\[
  \P_1\otimes \P_2' \otimes \P_3 \otimes \P_F
  \preceq \P_1\otimes \P_2 \otimes \P_3 \otimes \P_F
  \preceq \P\otimes \P_F
  \preceq \mu
\]
From above, we know that $\nu = {\sum_{\sigma \in \supp(\mu)}} \mu(\sigma) \cdot \delta_{\sigma[ x \coloneqq \de{E'}_{\mathsf{LExp}}(\Gamma) ]}$, so it is easy to see that $\pi_{S\setminus\{x\}}(\nu) = \pi_{S\setminus\{x\}}(\mu)$, therefore:
\[
  \P_1 \otimes \P_3\otimes \P_F\preceq \pi_{S\setminus\{x\}}(\mu) = \pi_{S\setminus\{x\}}(\nu)
\]
Let $\Q$ be the trivial probability space that satisfies $\sure{x\mapsto E}$, so that $\Omega_\Q = \mem{\{x\}}$ all events in $\F_\Q$ have probability 0 or 1, and $\Q \preceq \pi_{\{x\}}(\nu)$. So, clearly $\Gamma,\P_1\otimes \Q\otimes\P_3\vDash\psi\sep\sure{x\mapsto E}\sep\sure{I}$ and: $(\P_1 \otimes \Q\otimes\P_3)\otimes \P_F\preceq \nu$.

\item\ruleref{Samp}.
\[
  \inferrule{
    \varphi \Rightarrow \sure{ e \mapsto E} \land (\psi\sep\sure{\own(x)}
  }{
    I\vdash_m\triple{\varphi}{x\samp d( e)}{\psi \sep (x \sim d( E)}
  }{\ruleref{Samp}}
\]
We prove only the $m=\st$ case, as the $m=\wk$ case follows from the soundness of the \ruleref{Weakening} rule.
Suppose that $\P\otimes \P_F \preceq\mu$ and $\Gamma,\P\vDash \varphi \sep\sure I$, where $\mu \in \D(\mem S)$.
Since $\varphi \Rightarrow \sure{e \mapsto E}$, we know that $\de{e}_\expr(\sigma) = \de{E}_{\mathsf{LExp}}(\Gamma)$ for all $\sigma\in\supp(\mu)$. Now take any $\nu \in \lin^{\sem{I}_\Gamma}(\de{x \samp d(e)})^\dagger(\mu)$. We know that:
\begin{align*}
  \lin^{\sem{I}_\Gamma}(\de{x \samp d(e)})^\dagger(\mu)
  &= \left( \de{x \samp d(e)}^{\sem{I}_\Gamma}_\act \right)^\dagger(\mu)
  \intertext{From the premise of the rule, we know that $x\samp e$ does not depend on $\free(I)$, so the invariant sensitive semantics is the same as the regular semantics.}
  &= \de{x \samp d(e)}_\act^\dagger(\mu)
  \\
  &= \left\{
    \smashoperator[r]{\sum_{\sigma \in \supp(\mu)}} \mu(\sigma) \cdot \nu_\sigma
    \;\;\Big|\;\;
    \forall \sigma.\
    \nu_\sigma \in \de{x\samp d(e)}_\act(\sigma)
  \right\}
  \\
  &= \left\{
    \smashoperator[r]{\sum_{\sigma \in \supp(\mu)}} \mu(\sigma) \cdot
    \smashoperator{\sum_{v \in \mathsf{Val}}} d(\de{e}_\expr(\sigma))(v) \cdot
    \delta_{\sigma[ x \coloneqq v ]}
  \right\}
  \\
  &= \left\{
    {\sum_{v \in \mathsf{Val}}} d(\de{E}_{\mathsf{LExp}}(\Gamma))(v) \cdot
    \smashoperator{\sum_{\sigma \in \supp(\mu)}} \mu(\sigma) \cdot
    \delta_{\sigma[ x \coloneqq v ]}
  \right\}
\end{align*}
Let $\nu_v = {\sum_{\sigma \in \supp(\mu)}} \mu(\sigma) \cdot \delta_{\sigma[ x \coloneqq v ]}$, therefore $\nu = \sum_{v \in \mathsf{Val}} d(\de{E}_{\mathsf{LExp}}(\Gamma))(v) \cdot \nu_v$. Note also that $\pi_{S \setminus\{x\}}(\mu) = \pi_{S \setminus\{x\}}(\nu)$.
Since $\varphi \Rightarrow \psi\sep\sure{\own(x)}$, we know that $\Gamma,\P \vDash \psi \sep \sure{\own(x)}\sep\sure{I}$, so there exist $\P_1$, $\P_2$, and $\P_3$ such that $\Gamma,\P_1\vDash\psi$, $\Gamma,\P_2\vDash \sure{\own(x)}$, and $\Gamma,\P_3\vDash\sure I$. This gives us:
\[
  \P_1 \otimes \P_3 \otimes \P_F\preceq \pi_{S \setminus\{x\}}(\mu) = \pi_{S \setminus\{x\}}(\nu) = \pi_{S \setminus\{x\}}(\nu_v)
\]
Now, let $\Q_v$ be the trivial probability space such that $\Omega_{\Q_v} = \{ \sigma \in \mem{\{x\}} \mid \sigma(x) = v \}$ and $\Gamma[X \coloneqq v],\Q_v\vDash \sure{x\mapsto X}$. By Construction:
\[
  \P_1 \otimes \comp(\Q_v) \otimes \P_3 \otimes \P_F
  = (\P_1 \otimes \P_3\otimes\P_F) \otimes \comp(\Q_v)
  \preceq \pi_{S\setminus\{x\}}(\nu_v) \otimes \pi_{\{x\}}(\nu_v)
  = \nu_v
\]
Let $\xi = d(\de{E}_{\mathsf{LExp}}(\Gamma))$ and $\Q = \P_1 \otimes (\bigoplus_{v\sim \xi}\Q_v) \otimes \P_3$, then using \Cref{lem:oplus-mono,lem:otimes-oplus-dist} we get:
\[
  \Q\otimes P_F
  = \P_1 \otimes \left(\bigoplus_{v\sim \xi}\Q_v \right) \otimes \P_3 \otimes \P_F
  = \bigoplus_{v\sim \xi} \P_1 \otimes \Q_v \otimes \P_3 \otimes \P_F
  \preceq \sum_{v\in\supp(\xi)} \xi(v) \cdot \nu_v
  = \nu
\]
By construction, we also have $\Gamma,\Q \vDash \psi\sep(x \sim d(E))\sep\sure{I}$.

\item\ruleref{Par}.
\[
\inferrule{
    I\vdash\triple{\varphi_1}{C_1}{\psi_1}
    \\
    I\vdash\triple{\varphi_2}{C_2}{\psi_2}
    \\
    \precise{\psi_1,\psi_2}
}{
    I\vdash\triple{\varphi_1\sep\varphi_2}{C_1 \parallel C_2}{\psi_1 \sep\psi_2}
}{\ruleref{Par}}
\]
Since we are using a strong triple, suppose that $\P_1 \otimes \P_2 \otimes \P_I \otimes \P_F \preceq \mu$ such that $\Gamma,\P_1\vDash \varphi_1$ and $\Gamma,\P_2\vDash\varphi_2$ and $\Gamma,\P_I\vDash\sure I$.
Without loss of generality, suppose that $\P_I$ is minimal.
To complete the proof, we need to establish the premise of \Cref{lem:par}. 
Since $\psi_1$ and $\psi_2$ are precise, let $\Q_1$ and $\Q_2$ be the unique smallest probability spaces that satisfy them under $\Gamma$ (note that if $\psi_1$ or $\psi_2$ is unsatisfiable, then the premise of the rule is false, so the claim holds vacuously).
Take any $\mu_1$ such that $\P_1\otimes \P_I \preceq\mu_1$ and any $\nu_1 \in \lin^{\sem{I}_\Gamma}(\de{C_1})^\dagger(\mu_1)$. By the premise of the \ruleref{Par} rule, we know that there is a $\Q'_1 \preceq \nu_1$ such that $\Gamma,\Q'_1\vDash\psi_1\sep\sure I$, therefore $\Q_1\sep\P_I \preceq \Q'_1$, and therefore we have shown that:
\[
  \forall \mu_1.\ \P_1\otimes\P_I \preceq \mu_1 \implies \forall \nu_1\in \lin^{\sem{I}_\Gamma}(\de{C_1})^\dagger(\mu_1).\ \Q_1\otimes \P_I \preceq \nu_1
\]
We now perform a nearly identical argument for $C_2$, but we also handle the frame. Take any $\mu_2$ such that $\P_2 \otimes\P_I\otimes \P_F \preceq \mu_2$ and any $\nu_2 \in \lin^{\sem{I}_\Gamma}(\de{C_2})^\dagger(\mu_2)$. By the second premise of the \ruleref{Par} rule, we get that there is a $\Q'_2$ such that $\Q'_2 \otimes \P_F \preceq \nu_2$ and $\Gamma,\Q'_2\vDash \psi_2\sep\sure I$. Due to precision, we know that $\Q_2 \sep\P_I \preceq \Q'_2$ and so by \Cref{lem:otimes-mono}, $\Q_2 \otimes\P_I\otimes \P_F \preceq \Q'_2 \otimes \P_F\preceq \nu_2$, so we have shown that:
\[
  \forall \mu_2.\ \P_2 \otimes\P_I\otimes\P_F \preceq \mu_2 \implies \forall \nu_2.\ \lin^{\sem{I}_\Gamma}(\de{C_2})^\dagger(\mu_2).\ \Q_2 \otimes\P_I\otimes \P_F \preceq \nu_2
\]
Now, by \Cref{lem:par}, $\Q_1 \otimes \Q_2 \otimes \P_I\otimes \P_F\preceq \nu$ for all $\nu \in \lin^{\sem{I}_\Gamma}(\de{C_1}\parallel\de{C_2})^\dagger(\mu)$. Let $\Q = \Q_1\otimes \Q_2\otimes \P_I$. So $\Q\otimes \P_F\preceq\nu$ and since $\Gamma,\Q_1\vDash\psi_1$ and $\Gamma,\Q_2\vDash\psi_2$, we get that $\Gamma,\Q\vDash\psi_1\sep\psi_2\sep\sure{I}$, so we are done.

\item\ruleref{Atom}.
\[
\inferrule{
  J \vdash_m \triple{\varphi \sep \sure I}a{\psi\sep \sure I}
}{
  {I \sep J}\vdash_m\triple\varphi{a}\psi
}{\ruleref{Atom}}
\]
Suppose that $\P' \in \P\diamond_m \P_F$, $\P' \preceq \mu$, and $\Gamma,\P \vDash \varphi\sep\sure{I\sep J}$. This means that there exist $\P_1$ and $\P_2$ such that $\P_1 \otimes \P_2 \otimes \P_3\preceq \P$ and $\Gamma,\P_1\vDash\varphi$ and $\Gamma,\P_2\vDash\sure{I\sep J}$. Let $\P_I$ and $\P_J$ be the trivial probability spaces that satisfy $\sure I$ and $\sure J$, respectively, so clearly $\P_I\otimes \P_J\preceq \P_2$.
Now, take any $\nu\in\lin^{\sem{I\sep J}_\Gamma}(\de{a})^\dagger(\mu)$ and note that:
\begin{align*}
  \lin^{\sem{I\sep J}_\Gamma}(\de{a})^\dagger(\mu)
  &= (\de{a}^{\sem{I\sep J}_\Gamma}_\act)^\dagger(\mu)
  \\
  &= \left(
    \left(\mathsf{check}^{\sem{I\sep J}_\Gamma}\right)^\dagger \circ \de{a}_\act^\dagger \circ \left(\mathsf{replace}^{\sem{I \sep J}_\Gamma}\right)^\dagger \circ \left(\mathsf{check}^{\sem{I\sep J}_\Gamma}\right) \right)^\dagger(\mu)
  \\
  &= \left(\mathsf{check}^{\sem{I}_\Gamma\sep \sem{J}_\Gamma}\right)^\dagger \left( \de{a}_\act^\dagger \left( \left(\mathsf{replace}^{\sem{I}_\Gamma\sep \sem{J}_\Gamma}\right)^\dagger \left( \left( \mathsf{check}^{\sem{I}_\Gamma\sep \sem{J}_\Gamma}\right)^\dagger(\mu) \right)\right)\right)
  \\
  &= \left(\mathsf{check}^{\sem{I}_\Gamma} \right)^\dagger\left( \lin^{\sem{J}_\Gamma}(\de{a})^\dagger\left( \left(\mathsf{replace}^{\sem{I}_\Gamma}\right)^\dagger \left( \left( \mathsf{check}^{\sem{I}_\Gamma}\right)^\dagger(\mu) \right)\right)\right)
 \intertext{Since we already assumed that $\mu$ satisfies $I\sep J$, the first check does nothing.}
  &= \left(\mathsf{check}^{\sem{I}_\Gamma} \right)^\dagger\left( \lin^{\sem{J}_\Gamma}(\de{a})^\dagger\left( \left(\mathsf{replace}^{\sem{I}_\Gamma}\right)^\dagger(\mu)\right)\right)
\end{align*}
Therefore, we know that $\nu= (\mathsf{check}^{\sem{I}_\Gamma})^\dagger(\nu')$ for some $\nu' \in \lin^{\sem{J}_\Gamma}(\de{a})^\dagger(\mu')$ and $ \mu' \in (\mathsf{replace}^{\sem{I}_\Gamma})^\dagger(\mu) )$.
We know that $\P_1 \otimes \P_I\otimes \P_J\preceq \P_1 \otimes \P_2 \preceq \P$, so there must be a $\P'' \in (\P_1 \otimes \P_I\otimes \P_J)\diamond_m \P_F$ such that $\P'' \preceq \P'$ and therefore $\P'' \preceq\mu$. Since $\P''$ only contains trivial information about $I$, and $\mu$ and $\mu'$ differ only in the states that satisfy $I$, then $\P'' \preceq \mu'$ too.
Therefore, by the induction hypothesis, there must be $\Q$ and $\Q' \in \Q \diamond_m \P_F$ such that $\Q' \preceq \nu'$ and $\Gamma,\Q\vDash \psi\sep\sure I\sep\sure J$. This implies that $\Gamma,\nu'\vDash\sure I$, and therefore $\nu = \nu'$, since the check operation does nothing. Therefore, we get that $\Q' \preceq \nu' = \nu$, and so we are done.

\item\ruleref{Share}.
\[
\inferrule{
  I\sep J \vdash_m \triple{\varphi}C\psi
  \\
  \mathsf{finitary}(I)
}{
  J \vdash_m \triple{\varphi\sep \sure I}C{\psi\sep \sure I}
}{\ruleref{Share}}
\]
Suppose that $\P' \in \P\diamond_m \P_F$, $\P' \preceq \mu$, and $\Gamma,\P \vDash \varphi \sep\sure{I}\sep\sure{J}$. Now, take any $\nu \in \lin^{\sem{J}_\Gamma}(\de{C})^\dagger(\mu)$. By \Cref{lem:inv-mono}, we know that:
\[
  \lin^{\sem{I\sep J}_\Gamma}(\de{C})^\dagger(\mu)
  =
  \lin^{\sem{I}_\Gamma\sep\sem{J}_\Gamma}(\de{C})^\dagger(\mu)
  \lec
  \lin^{\sem{J}_\Gamma}(\de{C})^\dagger(\mu)
\]
And since $\lec$ is equivalent to $\supseteq$, then $\nu \in \lin^{\sem{I\sep J}_\Gamma}(\de{C})^\dagger(\mu)$ as well. Therefore, by the induction hypothesis, we know that there exist $\Q$ and $\Q' \in \Q\diamond_m \P_F$ such that $\Q' \preceq \nu$ and $\Gamma,\Q\vDash\psi \sep \sure{I\sep J}$. Since $\sure{I\sep J} \Leftrightarrow \sure{I}\sep\sure{J}$, we are done.

\item\ruleref{Frame}.
\[
\inferrule{
  I\vdash_m\triple{\varphi}C{\psi}
}{
  I\vdash_m\triple{\varphi \sep_m \vartheta}C{\psi\sep_m \vartheta}
}{\ruleref{Frame}}
\]
Suppose that $\P' \in \P \diamond_m \P_F$ and $\P'\preceq \mu$ and $\Gamma,\P\vDash (\varphi \sep_m \vartheta) \sep \sure{I}$.
That means that there exist $\P_1$ and $\P_I$ such that $\P_1 \otimes \P_I \preceq \P$ and $\Gamma,\P_1\vDash \varphi\sep_m\vartheta$, and $\Gamma,\P_I\vDash\sure{I}$, and without loss of generality, suppose that $\P_I$ is the trivial probability space that satisfies $I$. We also know that $\P_1 \in \P_2\diamond_m\P_3$ where $\Gamma,\P_2\vDash\varphi$ and $\Gamma,\P_3\vDash\vartheta$. Since $\P_I$ is a trivial probability space, then $\P_2 \diamond_m \P_I = \{ \P_2\otimes\P_I \}$, so by associativity of projections, we get that there must be a $\P_F' \in \P_F \diamond_m \P_3$ such that $\P' \in (\P_2\otimes \P_I) \diamond_m \P_F'$.

Now, take any $\nu \in \lin^{\sem{I}_\Gamma}(\de{C})^\dagger(\mu)$. By the induction hypothesis, there exist $\Q$ and $\Q' \in \Q\diamond_m \P_F'$ such that $\Q'\preceq\nu$ and $\Q \vDash\psi\sep\sure I$. So, there must be a $\Q'' \in \Q \diamond_m \P_3$ and $\Q' \in \Q''\diamond_m \P_F$, and by constructions $\Gamma,\Q''\vDash (\psi\sep\sure I)\sep_m \vartheta$. Since $\sure I$ is a pure assertion, this also means that $\Gamma,\Q''\vDash (\psi\sep_m\vartheta) \sep \sure I$.

\item\ruleref{Weaken}.
\[
  \inferrule{
    I \vdash\triple{\varphi}C\psi
  }{
    I\vdash_\wk \triple{\varphi}C\psi
  }{\ruleref{Weaken}}
\]
Immediate from the induction hypothesis, letting $\P_F$ be the trivial probability space on $\mem\emptyset$, and \Cref{lem:weak-triple}.

\item\ruleref{Strengthen}.
\[
  \inferrule{
    I \vdash_\wk \triple{\varphi}C\psi
    \\
    \precise\psi
  }{
    I\vdash \triple{\varphi}C\psi
  }{\ruleref{Strengthen}}
\]
Suppose that $\P\otimes \P_F \preceq \mu$, and $\Gamma,\P \vDash \varphi\sep\sure{I}$.
This means that there exists a $\P'$ such that $\Gamma,\P'\vDash\varphi$ and $\Gamma,\P_I\vDash\sure I$ where $\P_I$ is the trivial probability space satisfying $\sure I$ and $\P'\otimes \P_I\preceq \P$.
Since $\psi$ is precise, let $\Q$ be the unique smallest probability space such that $\Gamma,\Q\vDash\psi$. From the premise of the rule, we know that:
\[
  \forall \mu_1.\ \P' \otimes \P_I \preceq \mu_1 \implies \forall \nu_1 \in \lin^{\sem{I}_\Gamma}(\de{C})^\dagger(\mu_1).\ \Q\otimes \P_I \preceq \nu_1
\]
In addition, since $\lin^{\sem{I}_\Gamma}(\de{\skp}) = \eta$ \cite[Lemma 5.2]{zilberstein2025denotational} it is obvious that:
\[
  \forall \mu_2.\ \P_F\otimes \P_I \preceq \mu_2 \implies \forall \nu_2 \in \lin^{\sem{I}_\Gamma}(\de{\skp})^\dagger(\mu_2).\ \P_F\otimes \P_I\preceq \nu_2
\]
So, by \Cref{lem:par}, we get that $\Q \otimes \P_F\otimes\P_I \preceq \nu$ for any $\nu \in \lin^{\sem{I}_\Gamma}(\de{C} \parallel\de{\skp})^\dagger(\mu)$. But clearly $\lin^{\sem{I}_\Gamma}(\de{C} \parallel\de{\skp}) = \lin^{\sem{I}_\Gamma}(\de{C})$, so we are done.

\item\ruleref{Split1}.
\[
\inferrule{
  I\vdash_m\triple{\varphi}C{\psi}
  \\
  \psi \Rightarrow \sure{e \mapsto X}
  \\
  X\notin\mathsf{fv}(I)
}{
  I\vdash_m\triple{\smashoperator{\bigoplus_{X\sim d(E)}}\varphi}C{\smashoperator{\bigoplus_{X\sim d(E)}}\psi}
}{\ruleref{Split1}}
\]
We first prove the claim for the case where $m = \st$.
Suppose that $\P\otimes \P_F \preceq \mu$ and $\Gamma,\P\vDash(\bigoplus_{X\sim d(E)}\varphi) \sep\sure I$. This means that $(\bigoplus_{v\sim \xi} \P_v) \otimes \P_I \preceq \P$ such that $\Gamma[X \coloneqq v], \comp(\P_v)\vDash \varphi$ for each $v$ and $\xi = d(\de{E}_{\mathsf{LExp}}(\Gamma))$ and $\P_I$ is the trivial probability space satisfying $\sure I$. By \Cref{lem:otimes-oplus-dist,lem:otimes-mono} we have:
\[
  \bigoplus_{v\sim\xi} (\P_v \otimes \P_I \otimes \P_F)
  =   \left(\bigoplus_{v\sim\xi} \P_v\right) \otimes \P_I \otimes \P_F
  \preceq   \P\otimes \P_F
  \preceq\mu
\]
Now, for each $v\in\supp(\xi)$, let $\mu_v(\sigma) \triangleq \frac{1}{\xi(v)} \cdot \mu(\sigma)$ if $\sigma\in \Omega_{\P_v}\sep \Omega_{\P_I} \sep \Omega_{\P_F}$, which clearly gives us $\P_v\otimes \P_I\otimes \P_F \preceq \mu_v$ and $\mu = \sum_{v\in\supp(\xi)} \xi(v) \cdot \mu_v$. Now, take any $\nu \in \lin^{\sem{I}_\Gamma}(\de{C})^\dagger(\mu)$. We know that:
\begin{align*}
  \lin^{\sem{I}_\Gamma}(\de{C})^\dagger(\mu)
  &= \lin^{\sem{I}_\Gamma}(\de{C})^\dagger(\smashoperator{\sum_{v\in\supp(\xi)}} \nu(v)\cdot \mu_v)
  \\
  &= \left\{ \smashoperator[r]{\sum_{v\in \supp(\xi)}} \xi(v) \cdot \nu_v \;\; \Big|\;\; \forall v.\ \nu_v\in \lin^{\sem{I}_\Gamma}(\de{C})^\dagger(\mu_v) \right\}
\end{align*}
So $\nu =\sum_{v\in \supp(\xi)} \xi(v)\cdot \nu_v$ where $\nu_v \in \lin^{\sem{I}_\Gamma}(\de{C})^\dagger(\mu_v)$ for each $v$.
Note that $X\notin\mathsf{fv}(I)$, so $\sem{I}_\Gamma = \sem{I}_{\Gamma[X\coloneqq v]}$, and therefore by the induction hypothesis, we get that there exists $\Q_v$ such that $\Q_v\otimes \P_F \preceq \nu_v$ and $\Gamma[X \coloneqq v],\Q_v \vDash \psi\sep\sure I$. Since $\psi\Rightarrow\sure{e=X}$, we can restrict the $\Q_v$s to disjoint probability spaces $\Q_v'$ such that $\comp(\Q_v') \otimes \P_I \preceq \Q_v$.

Therefore $\bigoplus_{v\sim \nu}\Q'_v$ exists and $\Gamma, \bigoplus_{v\sim \xi}\Q'_v\vDash \bigoplus_{v\sim d(E)}\psi$.
So, we get that:
\[
  \left(\bigoplus_{v\sim\xi} \Q'_v\right) \otimes \P_I\otimes\P_F
  =\bigoplus_{v\sim\xi}\left(\Q'_v\otimes \P_I\otimes\P_F\right)
  \preceq \sum_{v\in\supp(\xi)} \xi(v)\cdot \nu_v = \nu
\]

The case where $m=\wk$ follows immediately from \Cref{lem:weak-triple} using the proof above with $\P_F$ being the trivial probability space on empty memories.

\item\ruleref{NSplit1}.
\[
\inferrule{
  I\vdash_m\triple{\varphi}C{\psi}
  \\
  \psi \Rightarrow \sure{e \mapsto X}
  \\
  X\notin\mathsf{fv}(I)
}{
  I\vdash_m\triple{\smashoperator{\bignd_{X\in E}}\varphi}C{\smashoperator{\bignd_{X\in E}}\psi}
}{\ruleref{NSplit1}}
\]
Suppose that $\P' \in \P\diamond_m \P_F$ and $\P' \preceq \mu$ and $\Gamma,\P\vDash(\bignd_{X\in E}\varphi)\sep\sure I$. This means that there is some $\xi \in \D(\de{E}_\mathsf{LExp}(\Gamma))$ such that $\Gamma,\P\vDash(\bigoplus_{X\sim \xi}\varphi)\sep\sure I$. 
Now take any $\nu \in \lin^{\sem{I}_\Gamma}(\de{C})^\dagger(\mu)$.
Using the \ruleref{Split1} rule, we get that there exists a $\Q$ and $\Q' \in \Q\diamond_m \P_F$ such that $\Q' \preceq \nu$ and $\Gamma,\Q\vDash(\bigoplus_{X\sim\xi}\psi)\sep\sure I$. This immediately implies that $\Gamma,\Q\vDash(\bignd_{X\in E}\psi)\sep\sure I$.

\item\ruleref{Split2}.
\[
\inferrule{
  I\vdash_m\triple{\varphi}C{\psi}
  \\
  \convex{\psi}
  \\
  X\notin\free(I, \psi)
}{
  I\vdash_m\triple{\smashoperator{\bigoplus_{X\sim d(E)}}\varphi}C{\psi}
}{\ruleref{Split2}}
\]
We start with the case where $m = \st$.
Suppose that $\P \otimes \P_F \preceq \mu$ and $\Gamma,\P\vDash (\bigoplus_{X\sim d(E)} \varphi) \sep\sure I$ and let $\xi = d(\de{E}_{\mathsf{LExp}}(\Gamma))$. Since $X \notin\free(I)$, then $\Gamma,\P\vDash \bigoplus_{X\sim d(E)} \varphi \sep\sure I$ and so there exists a family of $\P_v$ such that $\bigoplus_{v\sim\xi}\P_v \preceq \P$ and $\Gamma[X\coloneqq v],\comp(\P_v)\vDash\varphi\sep\sure I$ for all $v\in\supp(\xi)$. By \Cref{lem:oplus-mono,lem:otimes-oplus-dist}, we also have that:
\[
  \bigoplus_{v\sim\xi} \left( \P_v\otimes \P_F \right)
  = \left( \bigoplus_{v\sim\xi} \P_v \right) \otimes \P_F
  \preceq \P \otimes \P_F
  \preceq\mu
\]
Now, for each $v\in\supp(\xi)$, let $\mu_v(\sigma) \triangleq \frac{1}{\xi(v)} \cdot \mu(\sigma)$ if $\sigma\in \Omega_{\P_v}$, which clearly gives us $\comp(\P_v)\otimes \P_F \preceq \mu_v$ and $\mu = \sum_{v\in\supp(\xi)} \xi(v) \cdot \mu_v$. Now, take any $\nu \in \lin^{\sem{I}_\Gamma}(\de{C})^\dagger(\mu)$. We know that:
\begin{align*}
  \lin^{\sem{I}_\Gamma}(\de{C})^\dagger(\mu)
  &= \lin^{\sem{I}_\Gamma}(\de{C})^\dagger(\smashoperator{\sum_{v\in\supp(\xi)}} \nu(v)\cdot \mu_v)
  \\
  &= \left\{ \smashoperator[r]{\sum_{v\in \supp(\xi)}} \xi(v) \cdot \nu_v \;\;\Big|\;\; \forall v.\ \nu_v\in \lin^{\sem{I}_\Gamma}(\de{C})^\dagger(\mu_v) \right\}
\end{align*}
So $\nu =\sum_{v\in \supp(\xi)} \xi(v)\cdot \nu_v$ where $\nu_v \in \lin^{\sem{I}_\Gamma}(\de{C})^\dagger(\mu_v)$ for each $v$.
Note that $X\notin\free(I)$, so $\sem{I}_\Gamma = \sem{I}_{\Gamma[X\coloneqq v]}$, and therefore by the induction hypothesis, we get that there exist a family of $\Q_v$ such $\Q_v\otimes \P_F \preceq \nu_v$ and $\Gamma[X \coloneqq v],\Q_v \vDash \psi\sep\sure{I}$. Since $X \notin\free(\psi)$, we can remove the update of $X$ to conclude that $\Gamma,\Q_v \vDash \psi\sep\sure I$.
Since $\psi$ is convex, then we can presume that there exists $U$, $\F$, and $S$ such that $\Q_v = \tuple{\mem U, \F, \xi_v}$ for some $\xi_v\in S$. Now, let $\Q = \tuple{\mem U, \F, \sum_{v\in \supp(\xi)} \xi(v)\cdot \xi_v}$. Since $\sum_{v\in \supp(\xi)} \xi(v)\cdot \xi_v$ is a convex combinations of elements of $S$, then $\sum_{v\in \supp(\xi)} \xi(v)\cdot \xi_v \in S$ and therefore by construction $\Gamma,Q\vDash \psi\sep\sure{I}$. We now complete the proof by showing that $\Q \otimes \P_F \preceq \nu$. Take any $A \in \F_{\Q}$ and $B \in \F_{\P_F}$, and let $V$ be the set such that $\Omega_{\P_F} = \mem V$:
\begin{align*}
  \mu_{\Q\otimes\P_F}(A\sep B)
  &= \mu_\Q(A) \cdot \mu_{\P_F}(B) 
  \\
  &= \left( \sum_{v\in\supp(\xi)} \xi(v)\cdot \xi_v(A) \right) \cdot \mu_{\P_F}(B) 
  \\
  &= \sum_{v\in\supp(\xi)} \xi(v)\cdot \xi_v(A) \cdot \mu_{\P_F}(B) 
  \\
  &= \sum_{v\in\supp(\xi)} \xi(v)\cdot \mu_{\Q_v}(A) \cdot \mu_{\P_F}(B) 
  \intertext{Since $\Q_v\otimes \P_F \preceq \nu_v$:}
  &= \sum_{v\in\supp(\xi)}\xi(v) \cdot \pi_{U\cup V}(\nu_v)(A\sep B)
  \\
  &= \pi_{U\cup V}(\nu)(A)
\end{align*}
The case where $m= \wk$ follows from \Cref{lem:weak-triple} using the same reasoning as above, but where $\P_F$ is the trivial probability space on $\mem\emptyset$.

\item\ruleref{NSplit2}.
\[
\inferrule{
  I\vdash_m\triple{\varphi}C{\psi}
  \\
  \convex{\psi}
  \\
  X\notin\mathsf{fv}(I, \psi)
}{
  I\vdash_m\triple{\smashoperator{\bignd_{X\in E}}\varphi}C{\psi}
}{\ruleref{NSplit2}}
\]
Suppose that $\P' \in \P\diamond_m \P_F$ and $\P' \preceq \mu$ and $\Gamma,\P\vDash(\bignd_{X\in E}\varphi)\sep\sure I$. This means that there is some $\xi \in \D(\de{E}_\mathsf{LExp}(\Gamma))$ such that $\Gamma,\P\vDash(\bigoplus_{X\sim \xi}\varphi)\sep\sure{I}$. 
Now take any $\nu \in \lin^{\sem{I}_\Gamma}(\de{C})^\dagger(\mu)$.
Using the \ruleref{Split2} rule, we get that there exists a $\Q$ and $\Q' \in \Q \diamond_m \P_F$ such that $\Q' \preceq \nu$ and $\Gamma,\Q\vDash\psi\sep\sure I$.

\item\ruleref{Exists}.
\[
\inferrule{
  I\vdash_\wk\triple{\textstyle\bignd_{X\in E}\sure{P}}C{\psi}
  \\
  \sure{P}\Rightarrow\sure{e \mapsto X}
}{
  I\vdash_{\wk}\triple{\sure{\exists X\in E.\ P}}C{\psi}
}{\ruleref{Exists}}
\]
By \Cref{lem:weak-triple}, it suffices to show that the claim holds without frame preservation.
Suppose that $\Gamma,\mu\vDash \sure{\exists X \in E.P}\sep\sure I$. Without loss of generality, suppose that $X$ is not free in $I$ (if not, then we could $\alpha$-rename $X$ in $P$ to obtain an equivalent assertion, \ie $\sure{\exists X\in E.\ P} \Leftrightarrow \sure{\exists Y\in E.\ P[Y/X]}$ where $Y$ is some fresh variable). So, we get that $\Gamma, \mu\vDash \sure{\exists X\in E.\ P\sep I}$, which means that for every $\sigma \in \supp(\mu)$ there exists $v_\sigma\in\de{E}_{\mathsf{LExp}}(\Gamma)$ such that $\Gamma[X\coloneqq v_\sigma], \sigma\vDash P\sep I$.
Let $\xi \triangleq \sum_{\sigma \in \supp(\mu)} \mu(\sigma) \cdot \delta_{v_\sigma}$, so $\xi(v)$ is the probability that $\Gamma[X \coloneqq v]$ is the context that satisfies $P$. Also, let:
\[
  \mu_v(\sigma) \triangleq \frac{1}{\xi(v)} \left\{
    \begin{array}{ll}
      \mu(\sigma) & \text{if}~ v = v_\sigma \\
      0 & \text{if}~ v \neq v_\sigma
    \end{array}
  \right.
\]
Clearly, by construction $\Gamma[X\coloneqq v],\mu_v\vDash\sure{P}\sep\sure{I}$ for all $v\in\supp(\xi)$.
Also, since $\sure{P}\Rightarrow\sure{e\mapsto X}$, then $\de{e}_\expr(\sigma) = v$ for each $\sigma\in\supp(\mu_v)$, and so each $\mu_v = \comp(\mu_v')$ where $\mu_v'$ has a restricted sample space to only be over states where $e$ evaluates to $v$. Thus, we clearly have $\bigoplus_{v\sim \xi} \mu'_v \preceq \mu$, and therefore $\Gamma,\mu\vDash (\bignd_{X\in E}\sure P) \sep \sure I$.
So, by the induction hypothesis and \Cref{lem:weak-triple}, $\Gamma, \nu \vDash \psi \sep \sure{I}$ for all $\nu \in \lin^{\sem{I}_{\Gamma}}(\de{C})^\dagger(\mu)$.

\item\ruleref{Consequence}.
\[
  \inferrule{
    \varphi'\Rightarrow\varphi
    \\
    I\vdash_m\triple\varphi{C}\psi
    \\
    \psi\Rightarrow\psi'
  }{
    I\vdash_m\triple{\varphi'}C{\psi'}
  }{\ruleref{Consequence}}
\]
Suppose that $\P' \in (\P\otimes \trivP{\sem{I}_\Gamma})\diamond_m \P_F$, $\P' \preceq \mu$, and $\Gamma,\P \vDash \varphi'$. 
Since $\varphi'\Rightarrow\varphi$, then $\Gamma,\P\vDash\varphi$. By the induction hypothesis, there exists $\Q$ and $\Q' \in (\Q\otimes \trivP{\sem{I}_\Gamma})\diamond_m \P_F$ such that $\Q' \preceq\nu$ and $\Gamma,\Q\vDash\psi$ for any $\nu \in \lin^{\sem{I}_\Gamma}(\de{C})^\dagger(\mu)$. Since $\psi\Rightarrow\psi'$, then $\Gamma,\Q\vDash\psi'$, so we are done.


\item\ruleref{BoundedRank}.
Subject to the following conditions:
\begin{align*}
  (1)\ & \varphi\Rightarrow\sure{\ell \le R\le h}
  &
  (2)\ &\varphi\sep\sure{R=\ell} \Rightarrow \sure{b\mapsto\fls}
  &
  (3)\ &\varphi\sep\sure{R>\ell} \Rightarrow \sure{b\mapsto\tru}
 \\
  (4)\ &\precise{\varphi[\ell/R]}
  &
  (5)\ &N\notin\free(\varphi)
  &
  (6)\ &0<p\le 1
\end{align*}
The following inference is valid:
\[
  \inferrule{\textstyle
    I \vdash_m \triple{\varphi \sep \sure{R = N > \ell}}C{(\bignd_{R=\ell}^{N-1} \varphi) \oplus_{\ge p} (\bignd_{R=N}^h \varphi)}
  }{\textstyle
    I\vdash_m\triple{\bignd_{R=\ell}^h\varphi}{\whl bC}{\varphi[\ell/R]}
  }{\ruleref{BoundedRank}}
\]
We prove the rule without frame preservation, as frame preservation follows from \Cref{lem:weak-triple} if $m= \wk$ and the \ruleref{Strengthen} rule if $m=\st$ (since the postcondition is precise).
Suppose that $\Gamma,\mu\vDash \bignd_{R=\ell}^h \varphi\sep\sure{I}$.
Let $\varphi' = \bignd_{R=\ell+1}^h\varphi$ and $\psi = \varphi[\ell/R]$. We first show that $\tuple{\varphi',\psi}$ is an invariant pair for $\whl bC$ under $I$:
\begin{enumerate}
  \item $\varphi' = \bignd_{R=\ell+1}^h\varphi \Rightarrow\bignd_{R=\ell+1}^h \sure{b\mapsto\tru} \Rightarrow\sure{b\mapsto\tru}$
  
  \item $\psi = \varphi[\ell/R] \Rightarrow \sure{b\mapsto\fls}$
  
  \item We can weaken the postcondition of the premise of rule as follows:
  \[
    (\bignd_{R=\ell}^{N-1} \varphi) \oplus_{\ge p} (\bignd_{R=N}^h \varphi)
    \;\;\Rightarrow\;\;
    (\bignd_{R=\ell}^{N-1} \varphi) \nd (\bignd_{R=N}^h \varphi)
    \;\;\Rightarrow\;\;
    \bignd_{R=\ell}^h \varphi
    \;\;\Rightarrow\;\;
    \varphi' \nd\psi
  \]
  So, after an application of \ruleref{NSplit2}, we get that $I\vDash_m\triple{\varphi'}C{\varphi'\nd\psi}$. If $m=\st$, \ruleref{Weaken} can be used to obtain $I\vDash_\wk\triple{\varphi'}C{\varphi'\nd\psi}$.
  
  \item By assumption, we have $\precise{\psi}$.
\end{enumerate}
In addition, the premise of the rule implies that each iteration, the rank decreases by at least 1 with probability at least $p$, so starting in any state satisfying $\varphi$, the loop will terminate with probability at least $p^{h-\ell} >0$. Therefore by \Cref{cor:ast}, $\minterm(\lin^{\sem{I}_\Gamma}(\de{\whl bC})^\dagger(\mu)) = 1$, \ie the loop almost surely terminates.

Since $\varphi[\ell/R]$ is precise, let $\Q$ be the unique smallest probability space such that $\Gamma,\Q\vDash\varphi[\ell/R]$. Take any event $B \in \F_\Q$. By \Cref{lem:par-cor}, we get:
\[
  \minProb\left(\lin^{\sem{I}_\Gamma}(\de{\whl bC})^\dagger(\mu), B\sep\sem{I}_\Gamma \right)
  = \minterm\left(\lin^{\sem{I}_\Gamma}(\de{\whl bC})^\dagger(\mu)\right) \cdot \mu_\Q(B)
  = \mu_\Q(B)
\]
So, $\mu_\Q(B) \le \nu(B\sep\sem{I}_\Gamma)$ for any $\nu \in \lin^{\sem{I}_\Gamma}(\de{\whl bC})^\dagger(\mu)$. However, we already established that the loop almost surely terminates, so $\nu(\bot) = 0$, therefore $\nu(B\sep\sem{I}_\Gamma) = 1-\nu(\mem V \setminus (B\sep\sem{I}_\Gamma))$ (where $V$ is the domain of $\nu$), which gives us:
\begin{align*}
  \mu_\Q(B)
  &\le \nu(B\sep\sem{I}_\Gamma)
  \\& = 1 - \nu(\mem V\setminus (B\sep\sem{I}_\Gamma))
  \\&\le 1 - \mu_\Q(\mem V\setminus (B\sep\sem{I}_\Gamma))
  \\&= 1 - (1 - \mu_\Q(B\sep\sem{I}_\Gamma))
  \\&= \mu_\Q(B)
\end{align*}
Therefore $\mu_\Q(B) = \nu(B\sep\sem{I}_\Gamma)$, and since this is true for any $B\in\F_\Q$, then $\Q \otimes \P_I\preceq \nu$ where $\P_I$ is the trivial probability space satisfying $I$. By definition, $\Gamma,\Q\vDash\varphi[\ell/R]$, so we are done.


\item\ruleref{Disj}.
\[
\inferrule{
  I\vdash_m\triple{\varphi_1}C{\psi_1}
  \\
  I\vdash_m\triple{\varphi_2}C{\psi_2}
}{
    I\vdash_m\triple{\varphi_1\vee\varphi_2}C{\psi_1\vee\psi_2}
}{\ruleref{Disj}}
\]
Suppose that $\P' \in \P \diamond \P_F$ and $\P' \preceq \mu$ and $\Gamma,\P\vDash (\varphi_1\vee\varphi_2)\sep \sure I$. This means that $\P_1 \otimes \P_2 \preceq \P$ such that $\Gamma,\P_1 \vDash \varphi_1 \vee \varphi_2$ and $\Gamma,\P_2\vDash\sure I$. Without loss of generality, suppose that $\Gamma,\P_1 \vDash \varphi_1$, therefore $\Gamma,\P\vDash\varphi_1\sep\sure I$. Now take any $\nu \in \lin^{\sem{I}_\Gamma}(\de{C})^\dagger(\mu)$. By the induction hypothesis, we know that there exists $\Q$ and $\Q' \in \Q\diamond_m\P_F$ such that $\Q' \preceq\nu$ and $\Gamma,\Q\vDash \psi_1\sep\sure I$. We can weaken this to conclude that $\Gamma,\Q\vDash(\psi_1\vee\psi_2)\sep \sure I$.

\item\ruleref{Exists2}.
\[
\inferrule{
  I\vdash_m\triple{\varphi}C{\psi}
  \\
  X\notin\free(\psi,I)
}{
    I\vdash_m\triple{\exists X.\ \varphi}C{\psi}
}{\ruleref{Exists2}}
\]
Suppose that $\P' \in \P \diamond \P_F$ and $\P' \preceq \mu$ and $\Gamma,\P\vDash (\exists X.\varphi)\sep \sure I$. This means that $\P_1 \otimes \P_2 \preceq \P$ such that $\Gamma,\P_1 \vDash \exists X.\varphi$ and $\Gamma,\P_2\vDash\sure I$. 
Therefore, $\Gamma[X\coloneqq v],\P_1 \vDash \varphi$ for some $v\in\mathsf{Val}$. Since $X \notin\free(I)$, then we also have that $\Gamma[X\coloneqq v],\P \vDash \varphi\sep\sure I$.
Now take any $\nu \in \lin^{\sem{I}_\Gamma}(\de{C})^\dagger(\mu)$.
Since $I\notin\free(I)$, then $\sem{I}_{\Gamma[X\coloneqq v]} = \sem{I}_{\Gamma}$, so $\nu \in \lin^{\sem{I}_{\Gamma[X\coloneqq v]}}(\de{C})^\dagger(\mu)$ as well.
By the induction hypothesis, we know that there exists $\Q$ and $\Q' \in \Q\diamond_m\P_F$ such that $\Q' \preceq\nu$ and $\Gamma[X\coloneqq v],\Q\vDash \psi\sep\sure I$. Since $X \notin \free(\psi,I)$, then this implies that $\Gamma,\Q\vDash \psi\sep\sure I$.

\item\ruleref{Subst}
\[
\inferrule{
  I\vdash_m\triple{\varphi}C{\psi}
}{
  I[E/X] \vdash_m\triple{\varphi[E/X]}C{\psi[E/X]}
}{\ruleref{Subst}}
\]
Suppose that $\P' \in \P\diamond_m\P_F$ such that $\P'\preceq \mu$ and $\Gamma,\P\vDash\varphi[E/X] \sep \sure{I[E/X]}$. Let $\Gamma' = \Gamma[X \coloneqq \de{E}_{\mathsf{LExp}}(\Gamma)]$, so by construction $\Gamma', \P\vDash \varphi\sep\sure I$, and $\sem{I[E/X]}_{\Gamma} = \sem{I}_{\Gamma'}$. Now, take any $\nu' \in \lin^{\sem{I[E/X]}_\Gamma}(\de{C})^\dagger(\mu) = \lin^{\sem{I}_{\Gamma'}}(\de{C})^\dagger(\mu)$. By the induction hypothesis, we know that there exist $\Q'$ and $\Q'\in \Q \diamond_m\P_F$ such that $\Q' \preceq \nu$ and $\Gamma', \Q\vDash\psi$. This means that $\Gamma,\Q\vDash\psi[E/X]$.

\end{itemize}
\end{proof}

\subsection{Derived Rules}


\begin{lemma}
The following inference rule is derivable:
\[
\ruledef{If}{
  I\vdash_m\triple{\varphi \sep \sure{X=\tru}}{C_1}{\psi}
  \\
  I\vdash_m\triple{\varphi \sep \sure{X=\fls}}{C_2}{\psi}
  \\
  \varphi\Rightarrow\sure{b\mapsto X}
  \\
  \psi\Rightarrow\sure{b\mapsto X}
}{
  I\vdash_m\triple{\textstyle\bigoplus_{X\sim\bern{p}}\varphi}{\iftf b{C_1}{C_2}}{\bigoplus_{X\sim\bern{p}}\psi}
}
\]
\end{lemma}
\begin{proof}
Note in the proof below that $\tru=1$ and $\fls=0$.
\[
\inferrule*[right=\rulereff{Consequence}]{
  \inferrule*[Right=\rulereff{Split1}]{
    \inferrule*[Right=\rulereff{Disj}]{
      \inferrule*[right=\rulereff{IfT}]{
        I\vdash_m\triple{\varphi \sep \sure{X=1}}{C_1}{\psi}
      }{
        I\vdash_m\triple{\varphi \sep \sure{X=1}}{\iftf b{C_1}{C_2}}{\psi}
      }
      \inferrule*[right=\rulereff{IfF},vdots=4.5em,leftskip=5em,rightskip=5em]{
        I\vdash_m\triple{\varphi \sep \sure{X=0}}{C_2}{\psi}
      }{
        I\vdash_m\triple{\varphi \sep \sure{X=0}}{\iftf b{C_1}{C_2}}{\psi}
      }
    }{
      I\vdash_m\triple{(\varphi \sep \sure{X=1})\vee (\varphi \sep \sure{X=0})}{\iftf b{C_1}{C_2}}{\psi\vee\psi}
    }
  }{
    I\vdash_m\triple{\textstyle\bigoplus_{X\sim\bern{p}}(\varphi \sep \sure{X=1})\vee (\varphi \sep \sure{X=0})}{\iftf b{C_1}{C_2}}{\bigoplus_{X\sim\bern{p}}\psi\vee\psi}
  }
}{
  I\vdash_m\triple{\textstyle\bigoplus_{X\sim\bern{p}}\varphi}{\iftf b{C_1}{C_2}}{\bigoplus_{X\sim\bern{p}}\psi}
}
\]
\end{proof}

\begin{lemma}
The following inference rule is derivable:
\[
\ruledef{IfPure}{
  I\vdash_m\triple{\sure{P \land b\mapsto 1}}{C_1}{\sure Q}
  \\
  I\vdash_m\triple{\sure{P \land b\mapsto 0}}{C_2}{\sure Q}
  \\
  \sure P \Rightarrow\sure{b\in \{0,1\}}
}{
  I\vdash_m\triple{\sure P}{\iftf b{C_1}{C_2}}{\sure Q}
}
\]
\end{lemma}
\begin{proof}
We only show the derivation of the weak version. The strong version can be easily derived using \ruleref{Strengthen}, since the postcondition is precise (\eg see \Cref{lem:exists-strong}). Let $X$ be some fresh logical variable such that $X \notin\free(P, I)$.
\[
\inferrule*[right=\rulereff{Consequence}]{
  \inferrule*[Right=\rulereff{Exists}]{
    \inferrule*[Right=\rulereff{NSplit2}]{
      \inferrule*[Right=\rulereff{Disj}]{
        \inferrule*[right=\rulereff{IfT}]{
          I\vdash_\wk\triple{\sure{P\land b\mapsto 1}}{C_1}{\sure Q}
        }{
          I\vdash_\wk\triple{\sure{P\land b\mapsto 1}}{\iftf b{C_1}{C_2}}{\sure Q}
        } 
        \inferrule*[right=\rulereff{IfF},vdots=4.5em,leftskip=4em,rightskip=8em]{
          I\vdash_\wk\triple{\sure{P\land b\mapsto 0}}{C_2}{\sure Q}
        }{
          I\vdash_\wk\triple{\sure{P\land b\mapsto 0}}{\iftf b{C_1}{C_2}}{\sure Q}
        } 
      }{
        I\vdash_\wk\triple{\sure{P\land b\mapsto X \land (X=0\vee X=1)}}{\iftf b{C_1}{C_2}}{\sure Q}
      } 
    }{
      I\vdash_\wk\triple{\textstyle\bignd_{X\in\{0,1\}} \sure{P\land b\mapsto X \land (X=0\vee X=1)}}{\iftf b{C_1}{C_2}}{\sure Q}
    } 
  }{
    I\vdash_\wk\triple{\sure{\exists X.\ P\land b\mapsto X \land (X=0\vee X=1)}}{\iftf b{C_1}{C_2}}{\sure Q}
  } 
}{
  I\vdash_\wk\triple{\sure P}{\iftf b{C_1}{C_2}}{\sure Q}
} 
\]
\end{proof}

\begin{lemma}\label{lem:exists-strong}
The following inference rule is derivable:
\[
\ruledef{ExistsStrong}{
  I\vdash\triple{\textstyle\bignd_{X\in E} \sure P}C{\psi}
  \\
  \sure{P} \Rightarrow \sure{e \mapsto X}
  \\
  \precise{\psi}
}{
    I\vdash\triple{\sure{\exists X\in E.\ P}}C{\psi}
}
\]
\end{lemma}
\begin{proof}
\[
\inferrule*[right=\rulereff{Strengthen}]{
  \inferrule*[right=\rulereff{Exists}]{
    \inferrule*[right=\rulereff{Weaken}]{
      I\vdash\triple{\textstyle\bignd_{X\in E} \sure P}C{\psi}
    }{
      I\vdash_\wk\triple{\textstyle\bignd_{X\in E} \sure P}C{\psi}
    }
    \\
    \sure{P} \Rightarrow \sure{e \mapsto X}
  }{
    I\vdash_\wk\triple{\sure{\exists X\in E.\ P}}C{\psi}
  }
  \\
  \precise{\psi}
}{
  I\vdash\triple{\sure{\exists X\in E.\ P}}C{\psi}
}
\]
\end{proof}

\newpage
\begin{landscape}

\section{Examples}
\label{app:examples}

In this section, we provide derivations that were omitted in the main text. The content is rotated to take advantage of the full width of the page.

\subsection{Conditional Independence Example}
\label{app:cond-ind}
In this section, we show the derivation for the program that was introduced in \Cref{sec:overview}. The program is repeated below:
\[
  z \samp \bern{\tfrac12} \fatsemi \left( x \coloneqq z \;\mathlarger\parallel\; y \coloneqq 1-z \right)
\]
The first step is to break up the sequential composition, and derive a specification for the sampling operation. This is simple using the \ruleref{Seq}, \ruleref{Frame}, and \ruleref{Samp} rules. However, the derivation for the second part of the program is more difficult, and will be filled in shortly where the $(\bigstar)$ appears.
\[
\inferrule*[right=\ruleref{Seq}]{
  \inferrule*[right=\ruleref{Frame}]{
    \inferrule*[Right=\ruleref{Samp}]{\;}{
      \triple{\own(z)}{ z \samp \bern{1/2}}{z\sim\bern{1/2}}
    }
  }{
    \triple{\own(x, y, z)}{ z \samp \bern{1/2}}{\own(x, y) \sep z\sim\bern{1/2}}
  }
  \\
  (\bigstar)
}{
  \triple{\own(x, y, z)}{ z \samp \bern{1/2} \fatsemi ( x \coloneqq z \parallel y \coloneqq 1-z )}{\textstyle\bigoplus_{\bern{1/2}} \sure{x \mapsto Z}\sep \sure{y\mapsto 1-Z}\sep \sure{z \mapsto Z}}
}
\]
We now show the $(\bigstar)$ derivation. The first step is to rearrange the precondition using the entailment laws from \Cref{fig:entailment} to bring the outcome conjunction to the outside, so that we can apply the \ruleref{Cond1} rule. Conditioning makes $z$ deterministic, and therefore trivially independent from the rest of the state, so we can allocate an invariant with the \ruleref{Share} rule. The rest of the proof is straightforward.
\[
  \inferrule*[right=\ruleref{Consequence}]{
    \inferrule*[Right=\ruleref{Cond1}]{
      \inferrule*[Right=\ruleref{Share}]{
        \inferrule*[Right=\ruleref{Par}]{
          \inferrule*[right=\ruleref{Atom}]{
            \inferrule*[Right=\ruleref{Assign}]{\;}{
              \triple{\own(x) \sep \sure{z\mapsto Z}}{x \coloneqq z}{\sure{x \mapsto Z} \sep \sure{z\mapsto Z}}
            }
          }{
            z\mapsto Z\vdash\triple{\own(x)}{x \coloneqq z}{\sure{x \mapsto Z}}
          }
          \qquad
          \inferrule*[Right=\ruleref{Atom}]{
            \inferrule*[Right=\ruleref{Assign}]{\;}{
              \triple{\own(y) \sep \sure{z\mapsto Z}}{y \coloneqq 1-z}{\sure{y \mapsto 1-Z} \sep \sure{z\mapsto Z}}
            }
          }{
            z\mapsto Z\vdash\triple{\own(y)}{y \coloneqq 1-z}{\sure{y \mapsto 1-Z}}
          }
        }{
          z\mapsto Z\vdash\triple{\own(x, y)}{x \coloneqq z \parallel y \coloneqq 1-z}{\sure{x \mapsto Z}\sep \sure{y\mapsto 1-Z}}
        }
      }{
        \triple{\own(x, y) \sep \sure{z\mapsto Z}}{x \coloneqq z \parallel y \coloneqq 1-z}{\sure{x \mapsto Z}\sep \sure{y\mapsto 1-Z}\sep \sure{z \mapsto Z}}
      }
    }{
      \triple{\textstyle\bigoplus_{\bern{1/2}}\own(x, y) \sep \sure{z\mapsto Z}}{x \coloneqq z \parallel y \coloneqq 1-z}{\textstyle\bigoplus_{\bern{1/2}} \sure{x \mapsto Z}\sep \sure{y\mapsto 1-Z}\sep \sure{z \mapsto Z}}
    }
  }{
    \triple{\own(x, y) \sep z\sim\bern{1/2}}{x \coloneqq z \parallel y \coloneqq 1-z}{\textstyle\bigoplus_{\bern{1/2}} \sure{x \mapsto Z}\sep \sure{y\mapsto 1-Z}\sep \sure{z \mapsto Z}}
  }
\]

\subsection{Almost Sure Termination of a Random Walk}
\label{app:ast-example}

Recall the following random walk program from \Cref{sec:ast-rules}.
\[
\whl{x > 0}{
 b \samp \textstyle\bern{\frac12}\fatsemi
 \iftf{b}{
  x \coloneqq x-1
}{
  x \coloneqq y
}}
\]
We now give the full derivation for this program. Below in (\ref{eq:ast1}), we derive a specification for the sampling operation.
\begin{equation}\label{eq:ast1}
\inferrule*[right=\rulereff{Weaken}]{
  \inferrule*[Right=\rulereff{Consequence}]{
    \inferrule*[Right=\rulereff{Frame}]{
      \inferrule*[Right=\rulereff{Samp}]{\;}{
        y \in \{0, \ldots, 5\} \vdash\triple{\sure{\own(b)}}{
          {b \samp \textstyle\bern{\frac12}}
        }{b\sim\bern{\tfrac12}}
      }
    }{
      y \in \{0, \ldots, 5\} \vdash\triple{\sure{\own(b)} \sep \sure{ x\mapsto R \sep 0 < R = N \le 5}}{
        {b \samp \textstyle\bern{\frac12}}
      }{b\sim\bern{\tfrac12} \sep \sure{x\mapsto R \sep 0 < R = N \le 5}}
    }
  }{
    y \in \{0, \ldots, 5\} \vdash\triple{\sure{ x\mapsto R \sep \own(b)\sep 0 < R = N \le 5}}{
      {b \samp \textstyle\bern{\frac12}}
    }{\textstyle\bigoplus_{X \sim \bern{\frac12}} \sure{ b\mapsto X \sep x\mapsto R \sep 0 < R = N \le 5}}
  }
}{
  y \in \{0, \ldots, 5\} \vdash_\wk\triple{\sure{ x\mapsto R \sep \own(b)\sep 0 < R = N \le 5}}{
    {b \samp \textstyle\bern{\frac12}}
  }{\textstyle\bigoplus_{X \sim \bern{\frac12}} \sure{ b\mapsto X \sep x\mapsto R \sep 0 < R = N \le 5}}
}
\end{equation}
Next, in (\ref{eq:ast-if2}), we show the `false' branch of the if statement.
\begin{equation}\label{eq:ast-if2}
\inferrule*[right=\rulereff{Atom}]{
  \inferrule*[Right={\rulereff{Exists}+\rulereff{Consequence}}]{
    \inferrule*[Right=\rulereff{NSplit}]{
      \inferrule*[Right=\rulereff{Assign}]{\;}{
        \vdash_\wk\triple{\sure{ b\mapsto \fls \sep x\mapsto N \sep y\mapsto R}}{
         {x \coloneqq y}
        }{\sure{x\mapsto R\sep b\mapsto\fls \sep y \mapsto R}}          
      }
    }{
      \vdash_\wk\triple{\textstyle\bignd_{R=0}^5 \sure{ b\mapsto \fls \sep x\mapsto N \sep y\mapsto R}}{
       {x \coloneqq y}
      }{\textstyle\bignd_{R = 0}^5\sure{x\mapsto R\sep b\mapsto\fls \sep y \mapsto R}}    
    }
  }{
    \vdash_\wk\triple{\sure{ b\mapsto \fls \sep x\mapsto N}\sep\sure{y \in \{0, \ldots, 5\}}}{
     {x \coloneqq y}
    }{(\textstyle\bignd_{R = 0}^5\sure{x\mapsto R\sep b\mapsto\fls})\sep\sure{y \in \{0, \ldots, 5\}}}
  }
}{
  y \in \{0, \ldots, 5\} \vdash_\wk\triple{\sure{ b\mapsto \fls \sep x\mapsto N}}{
   {x \coloneqq y}
  }{\textstyle\bignd_{R = 0}^5\sure{x\mapsto R\sep b\mapsto\fls}}
}
\end{equation}
In (\ref{eq:ast2}), we give the full if statement.
\begin{equation}\label{eq:ast2}
\inferrule*[right=\rulereff{Consequence}]{
  \inferrule*[Right=\rulereff{If}]{
    \inferrule*[right=\rulereff{Assign}]{\;}{
      y \in \{0, \ldots, 5\} \vdash_\wk\triple{\sure{ b\mapsto \tru \sep x\mapsto N}}{
        {x \coloneqq x-1}
      }{\sure{x \mapsto N-1 \sep b\mapsto \tru}}      
    }
    \\
    (\ref{eq:ast-if2})
  }{
    y \in \{0, \ldots, 5\} \vdash_\wk\triple{\textstyle\bigoplus_{X \sim \bern{\frac12}} \sure{ b\mapsto X \sep x\mapsto N}}{
      \iftf{b}{x \coloneqq x-1}{x \coloneqq y}
    }{\sure{x \mapsto N-1 \sep b\mapsto \tru} \oplus_{\frac12} \textstyle\bignd_{R = 0}^5\sure{x\mapsto R\sep b\mapsto\fls}}
  }
}{
  y \in \{0, \ldots, 5\} \vdash_\wk\triple{\bigoplus_{X \sim \bern{\frac12}} \sure{ b\mapsto X \sep x\mapsto R \sep 0 < R = N \le 5}}{
    \iftf{b}{x \coloneqq x-1}{x \coloneqq y}
  }{\bignd_{R=0}^{N-1}\sure{x \mapsto R\sep\own(b)} \oplus_{\ge\frac12} \bignd_{R=N}^{5}\sure{x \mapsto R\sep\own(b)}}
}
\end{equation}
Finally, we complete the derivation below. Recall that the loop invariant used in \ruleref{BoundedRank} is $\varphi\triangleq{\sure{x \mapsto R \sep 0\le R\le 5 \sep \own(b)}}$.
\begin{mathpar}
\inferrule*[right=\rulereff{Strengthen}]{
  \inferrule*[Right=\rulereff{Exists}]{
    \inferrule*[Right=\rulereff{Consequence}]{
      \inferrule*[Right=\rulereff{BoundedRank}]{
        \inferrule*[Right=\rulereff{Seq},rightskip=5em]{
          (\ref{eq:ast1})
          \\
          (\ref{eq:ast2})
        }{
          y \in \{0, \ldots, 5\} \vdash_\wk\triple{\sure{ x\mapsto R \sep \own(b)\sep 0 < R = N \le 5}}{
            {b \samp \textstyle\bern{\frac12}\fatsemi \iftf{b}{x \coloneqq x-1}{x \coloneqq y}}
          }{\textstyle\bignd_{R=0}^{N-1}\sure{x \mapsto R\sep\own(b)} \oplus_{\ge\frac12} \bignd_{R=N}^{5}\sure{x \mapsto R\sep\own(b)}}
        }
      }{
        y \in \{0, \ldots, 5\} \vdash_\wk\triple{\textstyle\bignd_{R=0}^5 \sure{ x\mapsto R \sep \own(b)}}{
          \whl{x > 0}{b \samp \textstyle\bern{\frac12}\fatsemi \iftf{b}{x \coloneqq x-1}{x \coloneqq y}}
        }{\sure{x \mapsto 0\sep\own(b)}}
      }
    }{
      y \in \{0, \ldots, 5\} \vdash_\wk\triple{\textstyle\bignd_{R=0}^5 \sure{ x\mapsto R \sep \own(b)}}{
        \whl{x > 0}{b \samp \textstyle\bern{\frac12}\fatsemi \iftf{b}{x \coloneqq x-1}{x \coloneqq y}}
      }{\sure{x \mapsto 0}}
    }
  }{
    y \in \{0, \ldots, 5\} \vdash_\wk\triple{\sure{x \in \{0, \ldots, 5\} \sep \own(b)}}{
      \whl{x > 0}{b \samp \textstyle\bern{\frac12}\fatsemi \iftf{b}{x \coloneqq x-1}{x \coloneqq y}}
    }{\sure{x \mapsto 0}}
  }
}{
  y \in \{0, \ldots, 5\} \vdash\triple{\sure{x \in \{0, \ldots, 5\} \sep \own(b)}}{
    \whl{x > 0}{b \samp \textstyle\bern{\frac12}\fatsemi \iftf{b}{x \coloneqq x-1}{x \coloneqq y}}
  }{\sure{x \mapsto 0}}
}
\end{mathpar}

\subsection{Entropy Mixer}
\label{app:entropy-mixer}

Recall the following entropy mixer program from \Cref{sec:entropy-mixer}.
\[
  y \coloneqq 0 \fatsemi \left(\;\; x_1 \coloneqq y \fatsemi x_2 \samp \bern{\tfrac12} \fatsemi z \coloneqq \xor(x_1, x_2)  \;\;\mathlarger{\mathlarger\parallel}\;\;
  y \coloneqq 1 \;\;\right)
\]
We will analyze this program using the invariant $y \in \{0,1\}$, and conclude in the end that $z \sim \bern{\frac12}$. We begin by analyzing the read from shared state in the first thread. The derivation is given below in (\ref{eq:em-x1}).
\begin{equation}\label{eq:em-x1}
\inferrule*[right=\rulereff{Atom}]{
  \inferrule*[Right={\rulereff{Exists}}]{
    \inferrule*[Right=\rulereff{Consequence}]{
      \inferrule*[Right=\rulereff{NSplit1}]{
        \inferrule*[Right=\rulereff{Assign}]{\;}{
          \vdash_\wk\triple{\sure{y\mapsto Y \sep \own(x_1)}}{x_1 \coloneqq y}{\sure{x_1 \mapsto Y \sep y \mapsto Y }}
        } 
      }{
        \vdash_\wk\triple{\textstyle\bignd_{Y\in\{0,1\}} \sure{y\mapsto Y \sep {\own(x_1)}}}{x_1 \coloneqq y}{\textstyle\bignd_{Y\in\{0,1\}} \sure{x_1 \mapsto Y \sep y \mapsto Y}}
      } 
    }{
      \vdash_\wk\triple{\textstyle\bignd_{Y\in\{0,1\}} \sure{y\mapsto Y \sep {\own(x_1)}}}{x_1 \coloneqq y}{(\textstyle\bignd_{Y\in\{0,1\}} \sure{x_1 \mapsto Y}) \sep \sure{y \in \{0,1\}}}
    } 
  }{
    \vdash_\wk\triple{\sure{\own(x_1)} \sep \sure{y \in \{0,1\}}}{x_1 \coloneqq y}{(\textstyle\bignd_{Y\in\{0,1\}} \sure{x_1 \mapsto Y}) \sep \sure{y \in \{0,1\}}}
  } 
}{
  y\in\{0,1\}\vdash_\wk\triple{\sure{\own(x_1)}}{x_1 \coloneqq y}{\textstyle\bignd_{Y\in\{0,1\}} \sure{x_1 \mapsto Y}}
} 
\end{equation}
First, \ruleref{Atom} is applied to open the invariant. Next, we use \ruleref{Exists} and \ruleref{NSplit1} to gain access to the value of $y$, so that we can apply the \ruleref{Assign} rule. The rule of \ruleref{Consequence} is used to weaken the information about $y$ and move it out of the scope of the $\bignd$, so that the invariant can be closed. Since the postcondition at this stage is not precise, we are forced to use a weak triple. We now move on to derive a specification for the write to $z$ below in (\ref{eq:em-z}).
\begin{equation}\label{eq:em-z}
    \inferrule*[right=\!\!\rulereff{Consequence}]{
      \inferrule*[Right=\rulereff{Split1}]{
        \inferrule*[Right=\rulereff{Assign}]{
        }{
          \triple{\sure{x_1 \mapsto X} \sep \sure{x_2 \mapsto X'} \sep \sure{\own(z)}}{z \coloneqq \xor(x_1,x_2)}{\sure{z \mapsto \xor(X,X')}}
        }
      }{
        \triple{\textstyle\bigoplus_{X'\sim\bern{1/2}} \sure{x_1 \mapsto X} \sep \sure{x_2 \mapsto X'} \sep \sure{\own(z)}}{z \coloneqq \xor(x_1,x_2)}{\textstyle\bigoplus_{X'\sim\bern{1/2}}\sure{z \mapsto \xor(X,X')}}
      }
    }{
      \triple{\sure{x_1 \mapsto X} \sep (x_2 \sim \bern{1/2}) \sep \sure{\own(z)}}{z \coloneqq \xor(x_1,x_2)}{z \sim \bern{1/2}}
    }
\end{equation}
At this point, we have that $x_1$ is deterministic and $x_2$ is distributed according to a Bernoulli distribution.
We use \ruleref{Split1} to do case analysis on the result of the sampling operation, at which point we use \ruleref{Assign} to conclude that $z \mapsto \xor(X,X')$. Since $X$ is constant, then $\xor(X,X')$ is a bijection from $\{0,1\}$ to $\{0,1\}$, and we can therefore use the rule of \ruleref{Consequence} to conclude that $z$ is uniformly distributed. We can now compose this information with the sampling operation immediately before.
\begin{equation}\label{eq:em-x2-z}
\inferrule*[right={\rulereff{Consequence}, \rulereff{NSplit2}}]{
  \inferrule*[Right=\rulereff{Seq}]{
    \inferrule*[right=\rulereff{Frame}]{
      \inferrule*[right=\rulereff{Samp}]{
      }{
        y\in\{0,1\}\vdash_\wk\triple{\sure{\own(x_2)}}{x_2 \samp \bern{1/2}}{x_2 \sim \bern{1/2})}
      }
    }{
     y\in\{0,1\}\vdash_\wk \triple{\sure{x_1 \mapsto X} \sep \sure{\own(x_2, z)}}{x_2 \samp \bern{1/2}}{\sure{x_1 \mapsto X} \sep (x_2 \sim \bern{1/2}) \sep \sure{\own(z)}}
    }
  \\
  (\ref{eq:em-z})
  }{
    y\in\{0,1\}\vdash_\wk\triple{\sure{x_1 \mapsto X} \sep \sure{\own(x_2, z)}}{x_2 \samp \bern{1/2} \fatsemi z \coloneqq \xor(x_1, x_2)}{z \sim \bern{1/2}}
  }
}{
  y\in\{0,1\}\vdash_\wk\triple{\sure{\own(x_2, z)} \sep \textstyle\bignd_{Y\in\{0,1\}} \sure{x_1 \mapsto Y}}{x_2 \samp \bern{1/2} \fatsemi z \coloneqq \xor(x_1, x_2)}{z \sim \bern{1/2}}
} 
\end{equation}
We start by applying \ruleref{NSplit2} in order to gain access to the nondeterministic value of $x_1$. It is important to do this \emph{before} the sampling operation, since we need to ensure that $x_2$ is uniformly distributed given any fixed value for $x_1$. Next, we apply \ruleref{Seq} and derive specs for the individual commands. The sampling operation is simple, we complete the proof with the \ruleref{Frame} and \ruleref{Samp} rules. Now, we can derive a specification for the entire first thread in (\ref{eq:em-t1}).
\begin{equation}\label{eq:em-t1}
\inferrule*[right=\rulereff{Strengthen}]{
  \inferrule*[Right=\rulereff{Seq}]{
    \inferrule*[right=\rulereff{Frame}]{
      (\ref{eq:em-x1})
    }{
      y\in\{0,1\}\vdash_\wk\triple{\sure{\own(x_1, x_2, z)}}{x_1\coloneqq y}{\sure{\own(x_2, z)} \sep \textstyle\bignd_{Y\in\{0, 1\}} \sure{x_1\mapsto Y}}
    } 
    \\
    (\ref{eq:em-x2-z})
  }{
    y\in\{0,1\}\vdash_\wk\triple{\sure{\own(x_1, x_2, z)}}{x_1\coloneqq y \fatsemi x_2 \samp\bern{\tfrac12} \fatsemi z\coloneqq \xor(x_1, x_2)}{z \sim\bern{\tfrac12}}
  } 
}{
  y\in\{0,1\}\vdash\triple{\sure{\own(x_1, x_2, z)}}{x_1\coloneqq y \fatsemi x_2 \samp\bern{\tfrac12} \fatsemi z\coloneqq \xor(x_1, x_2)}{z \sim\bern{\tfrac12}}
} 
\end{equation}
This step essentially just involves using \ruleref{Seq} and \ruleref{Frame} to compose the triples that we previously derived. At the end, we also use \ruleref{Strengthen}, since the postcondition is now precise.
Finally, we complete the derivation for the whole program.
\[
\hspace{-8em}
\inferrule*[right=\rulereff{Seq}]{
  \inferrule*[right=\rulereff{Assign}]{\;}{
    \vdash\triple{\sure{\own(x_1, x_2, y, z)}}{ y\coloneqq 0}{\sure{y\mapsto 0 \sep \own(x_1, x_2, z)}}
  } 
  \inferrule*[right=\rulereff{Consequence},vdots=3em,leftskip=20em,rightskip=10em]{
    \inferrule*[Right=\rulereff{Share}]{
      \inferrule*[Right=\rulereff{Par}]{
        (\ref{eq:em-t1})
        \\
        \inferrule*[Right=\rulereff{Atom}]{
          \inferrule*[Right=\rulereff{Consequence}]{
            \inferrule*[Right=\rulereff{Assign}]{
            }{
              \triple{ \sure{y\in\{0,1\}} }{ y \coloneqq 1 }{ \sure{y \mapsto 1} }
            }
          }{
            \triple{ \sure{y\in\{0,1\}} }{ y \coloneqq 1 }{ \sure{y\in\{0,1\}} }
          }
        }{
          y\in\{0,1\} \vdash \triple{ \sure\tru }{ y \coloneqq 1 }{\sure\tru }
        }
      }{
        y \in \{0, 1\}\vdash\triple{\sure{\own(x_1, x_2, z)}}{ x_1\coloneqq y \fatsemi x_2 \samp\bern{\tfrac12} \fatsemi z\coloneqq \xor(x_1, x_2) \parallel y \coloneqq 1}{z \sim \bern{\tfrac12}}
      } 
    }{
      \vdash\triple{\sure{\own(x_1, x_2, z)}\sep\sure{y \in \{0, 1\}}}{ x_1\coloneqq y \fatsemi x_2 \samp\bern{\tfrac12} \fatsemi z\coloneqq \xor(x_1, x_2) \parallel y \coloneqq 1}{(z \sim \bern{\tfrac12})\sep\sure{y \in \{0, 1\}}}
    } 
  }{
    \vdash\triple{\sure{y\mapsto 0 \sep \own(x_1, x_2, z)}}{ x_1\coloneqq y \fatsemi x_2 \samp\bern{\tfrac12} \fatsemi z\coloneqq \xor(x_1, x_2) \parallel y \coloneqq 1}{z \sim \bern{\tfrac12}}
  } 
}{
  \vdash\triple{\sure{\own(x_1, x_2, y, z)}}{ y\coloneqq 0 \fatsemi ( x_1\coloneqq y \fatsemi x_2 \samp\bern{\tfrac12} \fatsemi z\coloneqq \xor(x_1, x_2) \parallel y \coloneqq 1 )}{z \sim \bern{\tfrac12}}
} 
\]
After concluding that $y\mapsto 0$ after the first command, we can weaken this information and use \ruleref{Share} to move $y$ into the invariant. Next, we use \ruleref{Par} to compositionally analyze the two threads. We have already seen the proof of the first thread above. The second thread has a simple derivation, since $y\coloneqq 1$ is clearly atomic and obeys the invariant.

\subsection{Concurrent Shuffling}
\label{app:shuffle}

Recall the following program from \Cref{sec:shuffle} for concurrent shuffling a list using two parallel threads.
\[
\begin{array}{ll}
\begin{array}{l}
a_1 \coloneqq [] \fatsemi  a_2 \coloneqq [] \fatsemi i\coloneqq 0 \fse \\
\whl{i < \code{len}(a)}{} ( \\
\quad b \samp \bern{\frac12} \fse \\
\quad \iftf b{a_1 \coloneqq a_1 \app [a[i]]}{a_2 \coloneqq a_2 \app [a[i]]} \fse \\
\quad i\coloneqq i+1 \\
) \fse \\
\code{shuffle}_1 \parallel \code{shuffle}_2 \fse \\
a \coloneqq a_1 \app a_2
\end{array}
&
\begin{array}{l}
\code{shuffle}_k: \\
\quad i_k \coloneqq \code{len}(a_k) - 1 \fse \\
\quad \whl{i_k > 0}{} \\
\qquad j_k \samp \unif{[0, \ldots, i_k]} \fse \\
\qquad a_k \coloneqq \code{swap}(a_k, i_k, j_k) \fse \\
\qquad i_k \coloneqq i_k - 1
\end{array}
\end{array}
\]
We now give the complete derivation of the correctness proof. We begin with the $\code{shuffle}_k$ program. To define the loop invariant, we first recursively define $\code{swaps}(n, \ell)$ which gives the set of possible lists obtained when $i_k \mapsto n$. It is easy to see that $\code{swaps}(0, \ell) = \Pi(\ell)$, since once $i$ reaches $n$, all permutations of the elements above position $n$ are accounted for, so once $i$ reaches 0, all permutations of the entire list are accounted for (the final position does not explicitly need to be chosen, since only one element remains).
\[
  \code{swaps}(n, \ell) = \left\{
  \begin{array}{ll}
     \{ \ell \} & \text{if} ~ n = \code{len}(\ell) - 1
     \\
    \{ \code{swap}(\ell', n+1, j) \mid \ell' \in \code{swaps}(n+1, \ell), j \in \{0, \ldots, n+1\} \} & \text{if} ~ 0 \le n < \code{len}(\ell) - 1
  \end{array}\right.
\]
We now define the loop invariant $\varphi$ in using $\code{swaps}$. The rank $R$ is given by $i_k$, which is bounded between $0$ and $\code{len}(A) - 1$. The final postcondition $\varphi[0/R]$ can also be simplified as follows.
\[
  \varphi \triangleq \sure{i_k \mapsto R \sep 0 \le R\le \code{len}(A) - 1 \sep \own(j_k)} \sep
     (a_k \sim \unif{\code{swaps}(R, A)})
   \qquad
   \varphi[0/R] = \sure{i_k \mapsto 0 \sep \own(j_k)} \sep (a_k\sim \unif{\Pi(A)})
\]
We now show the derivation for the loop body as a decorated program:
\begin{equation}\label{eq:fy-body}
\def\arraystretch{1.2}
\begin{array}{ll}
  &\ob{\varphi \sep \sure{R = N > 0}} \\
  \ruleref{Consequence}&\ob{\sure{i_k \mapsto N \sep 0 < N \le \code{len}(A) - 1 \sep \own(j_k)} \sep (a_k \sim \unif{\code{swaps}(\code{len}(A) - 1 - N, A)})}
  \\
  \ruleref{Samp}&\quad j_k \samp \unif{[0, \ldots, i_k]} \fse \\
  &\ob{\sure{i_k \mapsto N \sep 0 < N \le \code{len}(A) - 1} \sep (a_k \sim \unif{\code{swaps}(N, A)}) \sep (j_k\sim  \unif{[0, \ldots, N]}}\\
  \ruleref{Consequence} &\ob{\bigoplus_{X \sim \unif{\code{swaps}(N, A)}} \bigoplus_{Y\sim  \unif{[0, \ldots, N]}} \sure{i_k \mapsto N \sep 0 < N \le \code{len}(A) - 1\sep a_k \mapsto X \sep j_k \mapsto Y} }\\
  \ruleref{Split1},\ruleref{Assign}&\quad a_k \coloneqq \code{swap}(a_k, i_k, j_k) \fse \\
  &\ob{\bigoplus_{X \sim \unif{\code{swaps}(N, A)}} \bigoplus_{Y\sim  \unif{[0, \ldots, N]}} \sure{i_k \mapsto N \sep 0 < N \le \code{len}(A) - 1 \sep a_k \mapsto \code{swap}(X, N, Y) \sep j_k \mapsto Y} }\\
  \ruleref{Consequence} &\ob{\sure{i_k \mapsto N \sep 0 < N \le \code{len}(A) - 1 \sep  \own(j_k)} \sep (a_k \sim \unif{\code{swaps}(N-1, A)})} \\
  \ruleref{Assign} &\quad i_k \coloneqq i_k - 1 \\
  &\ob{\sure{i_k \mapsto N-1 \sep 0 < N \le \code{len}(A) - 1 \sep  \own(j_k)} \sep (a_k \sim \unif{\code{swaps}(N-1, A)})} \\
  \ruleref{Consequence} &\ob{(\bignd_{R=0}^{N-1} \varphi) \oplus_1 (\bignd_{R=N}^{\code{len}(A)-1} \varphi)}
\end{array}
\end{equation}
When entering the loop body, we know that the tail of $a_k$ is already shuffled. After performing the swap of $i_k$ and $j_k$, we extend the shuffled tail by one position. Clearly, the shuffles remain uniformly distributed, since $X$ is uniformly distributed, and the element at position $N$ is also chosen uniformly.
Now, we give the remaining derivation of $\code{shuffle}_k$, which is quite mechanical given the derivation of the loop body above:
\begin{equation}\label{eq:fy}
\inferrule*[right=\rulereff{Seq}]{
  \inferrule*[right=\rulereff{Assign},vdots=6em,rightskip=30em]{\;}{
    \vdash\triple{\sure{a_k\mapsto A \sep \own(i_k, j_k)}}{i_k \coloneqq\code{len}(a_k)-1}{\sure{a_k\mapsto A \sep i_k \mapsto \code{len}(A) - 1 \sep \own(j_k)}}
  }
  \inferrule*[Right=\rulereff{Consequence}]{
    \inferrule*[Right=\rulereff{BoundedRank}]{(\ref{eq:fy-body})}{
      \vdash\triple{\textstyle\bignd_{R = 0}^{\code{len}(A) -1} \varphi}{\whl{i_k> 0}{(\cdots)}}{(a_k\sim \unif{\Pi(A)}) \sep \sure{i_k\mapsto 0 \sep \own(j_k)}}
    }
  }{
    \vdash\triple{\sure{a_k\mapsto A \sep i_k \mapsto \code{len}(A) - 1 \sep \own(j_k)}}{\whl{i_k> 0}{(\cdots)}}{a_k\sim \unif{\Pi(A)}}
  }
}{
  \vdash\triple{\sure{a_k\mapsto A \sep \own(i_k, j_k)}}{i_k \coloneqq\code{len}(a_k)-1 \fatsemi \whl{i_k> 0}{(\cdots)}}{a_k\sim \unif{\Pi(A)}}
}
\end{equation}
We now move to deriving the specification for the main program. To analyze the loop, we use the following loop invariant, where $M = \code{len}(A)$. Clearly, the rank $R$ is bounded between $0$ and $M$, and when $R=0$ the loop must terminate since $i= \code{len}(a)$.
\[
  \psi = \sure{i \mapsto M-R \sep 0 \le R \le M \sep a\mapsto A \sep\own(b)} \sep \bigoplus_{X \sim \unif{\{0,1\}^{M-R}}} \sure{a_1 \mapsto A[X] \sep a_2 \mapsto A[\neg X] }
\]
We now show the derivation of the loop body as a decorated program. The key idea is to merge the newly sampled value of $b$ into the $\bigoplus$ over $X$ to extend the length of the bit-string on each iteration.
\begin{equation}\label{eq:shuffle-body}
\def\arraystretch{1.1}
\begin{array}{ll}
& \ob{\psi \sep \sure{R = N > 0}} \\
\ruleref{Consequence} & \lrob{\sure{i \mapsto M-N \sep 0 < N \le M \sep a\mapsto A\sep\own(b)} \sep \displaystyle\bigoplus_{X \sim \unif{\{0,1\}^{M-N}}} \sure{a_1 \mapsto A[X] \sep a_2 \mapsto A[\neg X]} } \\
\ruleref{Samp}&\quad b \samp \bern{\frac12} \fse \\
& \lrob{\sure{i \mapsto M-N \sep 0 < N \le M \sep a\mapsto A} \sep (b\sim\bern{\frac12}) \sep \displaystyle\bigoplus_{X \sim \unif{\{0,1\}^{M-N}}} \sure{a_1 \mapsto A[X] \sep a_2 \mapsto A[\neg X]} } \\
\ruleref{If} &\quad \ift{b}{} \\
& \quad \lrob{\sure{i \mapsto M-N \sep 0 < N \le M \sep a\mapsto A \sep b \mapsto Y \sep Y=1} \sep \displaystyle\bigoplus_{X \sim \unif{\{0,1\}^{M-N}}} \sure{a_1 \mapsto A[X] \sep a_2 \mapsto A[\neg X]} } \\
\ruleref{Assign}& \qquad a_1 \coloneqq a_1 \app [a[i]] \\
& \quad \lrob{\sure{i \mapsto M-N \sep 0 < N \le M \sep a\mapsto A \sep b \mapsto Y \sep Y=1} \sep \displaystyle\bigoplus_{X \sim \unif{\{0,1\}^{M-N}}} \sure{a_1 \mapsto A[X\app[Y]] \sep a_2 \mapsto A[\neg (X\app[Y])]} } \\
& \quad \code{else} \\
& \quad \lrob{\sure{i \mapsto M-N \sep 0 < N \le M \sep a\mapsto A \sep b \mapsto Y\sep Y=0} \sep \displaystyle\bigoplus_{X \sim \unif{\{0,1\}^{M-N}}} \sure{a_1 \mapsto A[X] \sep a_2 \mapsto A[\neg X]} } \\
\ruleref{Assign}& \qquad a_2 \coloneqq a_2 \app [a[i]] \fse \\
& \quad \lrob{\sure{i \mapsto M-N \sep 0 < N \le M \sep a\mapsto A \sep b \mapsto Y\sep Y=0} \sep \displaystyle\bigoplus_{X \sim \unif{\{0,1\}^{M-N}}} \sure{a_1 \mapsto A[X\app [Y]] \sep a_2 \mapsto A[\neg (X \app [Y])]\app A[M-N]} } \\
& \lrob{\displaystyle\bigoplus_{Y\sim \bern{\frac12}}\sure{i \mapsto M-N \sep 0 < N \le M \sep a\mapsto A \sep b \mapsto Y} \sep \bigoplus_{X \sim \unif{\{0,1\}^{M-N}}} \sure{a_1 \mapsto A[X\app [Y]] \sep a_2 \mapsto A[\neg (X \app [Y])]\app A[M-N]} } \\
\ruleref{Consequence} & \lrob{\displaystyle\sure{i \mapsto M-N \sep 0 < N \le M \sep a\mapsto A \sep \own(b)} \sep \bigoplus_{X \sim \unif{\{0,1\}^{M-(N-1)}}} \sure{a_1 \mapsto A[X] \sep a_2 \mapsto A[\neg X]} } \\
\ruleref{Assign}&\quad i\coloneqq i+1 \\
& \lrob{\displaystyle\sure{i \mapsto M-(N-1) \sep 0 < N \le M \sep a\mapsto A \sep \own(b)} \sep \bigoplus_{X \sim \unif{\{0,1\}^{M-(N-1)}}} \sure{a_1 \mapsto A[X] \sep a_2 \mapsto A[\neg X]} } \\
\ruleref{Consequence} & \lrob{\psi[N-1/R] }
\end{array}
\end{equation}
Finally, we conclude by showing the derivation of the main program, which is also quite mechanical. The final consequence relies on some combinatorial reasoning, which is explained in \Cref{sec:shuffle}.
\[
\def\arraystretch{1.35}
\begin{array}{ll}
& \lrob{\sure{a\mapsto A \sep \own(b, a_1, a_2, i, i_1, i_2, j_1, j_2)}}
\\
\ruleref{Assign} & \quad a_1 \coloneqq [] \fse
\\
& \lrob{\sure{a\mapsto A \sep a_1 \mapsto [] \sep \own(b, a_2, i, i_1, i_2, j_1, j_2)}}
\\
\ruleref{Assign} & \quad a_2 \coloneqq [] \fse
\\
& \lrob{\sure{a\mapsto A \sep a_1 \mapsto [] \sep a_2 \mapsto []  \sep \own(b, i, i_1, i_2, j_1, j_2)}}
\\
\ruleref{Assign} & \quad i \coloneqq 0 \fse
\\
& \lrob{\sure{a\mapsto A \sep a_1 \mapsto [] \sep a_2 \mapsto [] \sep i \mapsto 0 \sep \own(b, i_1, i_2, j_1, j_2)}}
\\
\ruleref{Consequence} & \lrob{\psi[\code{len}(A) / R]}
\\
\ruleref{BoundedRank}, (\ref{eq:shuffle-body}) & \quad \whl{i <\code{len}(A)}{(\cdots)} \fse
\\
& \lrob{\sure{i \mapsto \code{len}(A) \sep a\mapsto A \sep\own(\cdots)} \sep{\bigoplus_{X \sim \unif{\{0,1\}^{\code{len}(A)}}}} \sure{a_1 \mapsto A[X] \sep a_2 \mapsto A[\lnot X]}}
\\
\ruleref{Consequence} & \lrob{ {\bigoplus_{X \sim \unif{\{0,1\}^{\code{len}(A)}}}} \sure{a_1 \mapsto A[X] \sep a_2 \mapsto [ A[\neg X]  \sep\own(\cdots) }}
\\
\ruleref{Split1} & \lrob{ \sure{a_1 \mapsto A[X] \sep a_2 \mapsto A[\neg X] \sep\own(\cdots) }}
\\
\ruleref{Frame}, \ruleref{Par}, (\ref{eq:fy}) & \left(
  \begin{array}{l||l}
    \lrob{\sure{a_1 \mapsto A[X]\sep \own(i_1, j_1) }}
    &
    \lrob{\sure{a_2 \mapsto A[\neg X] \sep \own(i_2, j_2) }}
    \\
    \quad \code{shuffle}_1 & \quad \code{shuffle}_2
    \\
    \lrob{a_1 \sim \unif{\Pi(A[X])}}
    &
    \lrob{a_2 \sim \unif{\Pi(A[\neg X])}  }
  \end{array}\right) \fse
  \\
  & \lrob{ a_1 \sim \unif{\Pi(A[X])} \sep a_2 \sim \unif{\Pi(A[\neg X])} \sep \own(\cdots) }
  \\
  & \lrob{ {\bigoplus_{X \sim \unif{\{0,1\}^{\code{len}(A)}}}} a_1 \sim \unif{\Pi(A[X])} \sep a_2 \sim \unif{\Pi(A[\neg X])} \sep \sure{\own(a)} }
  \\
  \ruleref{Consequence} & \lrob{ {\bigoplus_{X \sim \unif{\{0,1\}^{\code{len}(A)}}}} \bigoplus_{A_1 \sim \unif{\Pi(A[X])}} \bigoplus_{A_2 \sim \unif{\Pi(A[\neg X])}} \sure{ a_1 \mapsto A_1 \sep a_2 \mapsto A_2 \sep\own(a)  } }
  \\
  \ruleref{Split1}, \ruleref{Assign} & a \coloneqq a_1 \app a_2
  \\
  & \lrob{ {\bigoplus_{X \sim \unif{\{0,1\}^{\code{len}(A)}}}} \bigoplus_{A_1 \sim \unif{\Pi(A[X])}} \bigoplus_{A_2 \sim \unif{\Pi(A[\neg X])}} \sure{ a \mapsto A_1 \app A_2 } }
  \\
  \ruleref{Consequence} & \lrob{a \sim \unif{\Pi(A)}}
\end{array}
\]
%
%
%
%

\subsection{Private Information Retrieval}
\label{app:pir}

Recall the following private information retrieval program from \Cref{sec:pir}.
\[
\begin{array}{l}
  \code{PrivFetch}: \\
  \quad q_1 \samp \unif{ \{0, 1\}^n} \fse \\
  \quad q_2 \coloneqq \xor(q_1, x) \fse \\
  \quad \code{fetch}_1 \parallel \code{fetch}_2 \fse \\
  \quad r \coloneqq \xor(r_1, r_2)
  \\\;
\end{array}
\qquad\qquad
\begin{array}{l}
\code{fetch}_k: \\
  \quad i_k \coloneqq 0 \fatsemi
  r_k \coloneqq 0 \fse \\
  \quad \whl{i_k < \code{len}(q_k)}{\;} \\
  \qquad \ift{q_k[i_k] = 1}{\;} \\
  \qquad\quad r_k \coloneqq \xor(r_k, d[i_k]) \fse \\
  \qquad i_k\coloneqq i_k+1
\end{array}
\]
We now give the complete correctness proof, showing that $\code{PrivFetch}$ correctly fetches the data entry indicated by $x$ from the database. We begin with the $\code{fetch}_k$ procedure, whose proof is shown below in (\ref{eq:pir-fetch}).
\begin{equation}\label{eq:pir-fetch}
\inferrule*[right=\rulereff{Seq}]{
  \inferrule*[right=\rulereff{Assign}]{\;}{
    d\mapsto D \vdash \triple{\sure{q_k \mapsto Q \sep \own(i_k, r_k)}}{i_k\coloneqq 0}{\sure{q_k \mapsto Q \sep i_k \mapsto 0\sep\own(r_k)}}
  } 
  \inferrule*[Right=\rulereff{Seq},vdots=4em,leftskip=20em]{
    \inferrule*[right=\rulereff{Assign}]{\;}{
      d\mapsto D \vdash \triple{\sure{q_k \mapsto Q \sep i_k \mapsto 0\sep\own(r_k)}}{r_k\coloneqq 0}{\sure{q_k \mapsto Q \sep i_k \mapsto 0\sep r_k\mapsto 0}}
    } 
    \\
    (\ref{eq:pir-loop})
  }{
    d\mapsto D \vdash \triple{\sure{q_k \mapsto Q \sep i_k \mapsto 0\sep\own(r_k)}}{r_k\coloneqq 0\fatsemi \whl{i_k < \code{len}(q_k)}{(\cdots)}}{\sure{r_k \mapsto \xor_{0 \le i < n : Q[i] = 1}D[i]}}
  } 
}{
  d\mapsto D \vdash \triple{\sure{q_k \mapsto Q \sep \own(i_k, r_k)}}{ i_k\coloneqq 0\fatsemi r_k\coloneqq 0\fatsemi \whl{i_k < \code{len}(q_k)}{(\cdots)}}{\sure{r_k \mapsto \xor_{0 \le i < n : Q[i] = 1}D[i]}}
} 
\end{equation}
The proof above is quite mechanical, just using \ruleref{Seq} and \ruleref{Assign} to analyze the initialization. The next step is to analyze the loop. To do so, we use the following invariant, which states that after $i_k$ iterations, $r_k$ holds a dot product of the first $i_k$ elements of $D$ and $Q$.
\[
  \varphi\triangleq \sure{q_k \mapsto Q \sep i_k \mapsto n - R \sep 0\le R\le n \sep r_k \mapsto \xor_{0 \le i < n-R : Q[i] = 1}D[i]}
\]
The rank $R$ indicates the remaining number of iterations until termination, and it clearly bounded between 0 and $n$. We can also specialize $\varphi$ to obtain the pre- and postconditions of the whole loop, shown below:
\[
  \varphi[n/R] = \sure{q_k \mapsto Q \sep i_k \mapsto 0 \sep r_k \mapsto 0}
  \qquad\qquad
  \varphi[0/R] = \sure{q_k\mapsto Q \sep i_k\mapsto n-R\sep r_k\mapsto\xor_{0 \le i < n : Q[i] = 1}D[i]}
\]
Analyzing the loop is simply a matter of applying \ruleref{BoundedRank}.
\begin{equation}\label{eq:pir-loop}
\inferrule*[right=\rulereff{Consequence}]{
  \inferrule*[Right=\rulereff{BoundedRank}]{
    (\ref{eq:pir-loop-body})
  }{
    d\mapsto D \vdash \triple{\textstyle\bignd_{R=0}^{n}\varphi}{\whl{i_k < \code{len}(q_k)}{(\cdots)}}{\varphi[0/R]}
  }
}{
  d\mapsto D \vdash \triple{\sure{q_k \mapsto Q \sep i_k \mapsto 0\sep r_k\mapsto 0}}{\whl{i_k < \code{len}(q_k)}{(\cdots)}}{\sure{r_k \mapsto \xor_{0 \le i < n : Q[i] = 1}D[i]}}
}
\end{equation}
Now, we show that the loop body upholds the invariant. We show this as a decorated program. We use the derived \ruleref{IfPure} rule to analyze the if statement, which is similar to the rule for if statements in Hoare Logic. In each case of the if statement, the respective value of $Q[n-N]$ allows us to extend the dot product.
\begin{equation}\label{eq:pir-loop-body}
\def\arraystretch{1.35}
\begin{array}{ll}
  & \lrob{\varphi\sep \sure{R = N > 0}}
  \\
  \ruleref{Consequence} & \lrob{\sure{q_k \mapsto Q \sep i_k \mapsto n - N \sep r_k \mapsto \xor_{0 \le i < n-N : Q[i] = 1}D[i]}}
  \\
  \ruleref{IfPure} &\quad \ift{q_k[i_k]}{} \\
  \ruleref{Atom} & \quad\lrob{\sure{q_k \mapsto Q \sep i_k \mapsto n - N \sep r_k \mapsto \xor_{0 \le i < n-N : Q[i] = 1}D[i] \sep Q[n-N] = 1}}
  \\
  & \qquad\lrob{\sure{q_k \mapsto Q \sep i_k \mapsto n - N \sep r_k \mapsto \xor_{0 \le i < n-N : Q[i] = 1}D[i] \sep Q[n-N] = 1} \sep \sure{d\mapsto D}}
  \\
  \ruleref{Assign} &\quad\qquad r_k = \xor(r_k, d[i_k])
  \\
  & \qquad\lrob{\sure{q_k \mapsto Q \sep i_k \mapsto n - N \sep r_k \mapsto \xor\left(\xor_{0 \le i < n-N : Q[i] = 1}D[i], D[n-N]\right) \sep Q[n-N] = 1} \sep \sure{d\mapsto D}}
  \\
  & \quad\lrob{\sure{q_k \mapsto Q \sep i_k \mapsto n - N \sep r_k \mapsto \xor\left(\xor_{0 \le i < n-N : Q[i] = 1}D[i], D[n-N]\right) \sep Q[n-N] = 1}}
  \\
  \ruleref{Consequence} & \quad\lrob{\sure{q_k \mapsto Q \sep i_k \mapsto n - N \sep r_k \mapsto \xor_{0 \le i < n-(N-1) : Q[i] = 1}D[i]}}
  \\
  & \code{else}
  \\
  &\quad\lrob{\sure{q_k \mapsto Q \sep i_k \mapsto n - N \sep r_k \mapsto \xor_{0 \le i < n-N : Q[i] = 1}D[i] \sep Q[n-N] = 0}}
  \\
  \ruleref{Skip} & \qquad\skp
  \\
  &\quad\lrob{\sure{q_k \mapsto Q \sep i_k \mapsto n - N \sep r_k \mapsto \xor_{0 \le i < n-N : Q[i] = 1}D[i] \sep Q[n-N] = 0}}
  \\
  \ruleref{Consequence}&\quad\lrob{\sure{q_k \mapsto Q \sep i_k \mapsto n - N \sep r_k \mapsto \xor_{0 \le i < n-(N-1) : Q[i] = 1}D[i]}}
    \\
  &\lrob{\sure{q_k \mapsto Q \sep i_k \mapsto n - N \sep r_k \mapsto \xor_{0 \le i < n-(N-1) : Q[i] = 1}D[i]}}
  \\
  \ruleref{Assign} & i_k \coloneqq i_k+1
  \\
  &\lrob{\sure{q_k \mapsto Q \sep i_k \mapsto n - (N-1) \sep r_k \mapsto \xor_{0 \le i < n-(N-1) : Q[i] = 1}D[i]}}
  \\
  \ruleref{Consequence} & \lrob{\varphi[N-1/R]}
  \\
  \ruleref{Consequence} & \lrob{(\bignd_{R=0}^{N-1}\varphi) \oplus_1 (\bignd_{R=N}^n \varphi)}
\end{array}
\end{equation}
Now we move on to deriving the specification for the main program. Ultimately, the postcondition states that $\sure{r \mapsto D[K]}$, \ie that we selected the correct data from the database.
\[
\inferrule*[right=\rulereff{Seq}]{
  \inferrule*[right=\rulereff{Samp}]{\;}{
    \vdash\triple{\sure{ x \mapsto \code{onehot}(K) \sep d\mapsto D \sep \own(\cdots)}}{q_1 \samp \unif{\{0,1\}^n}}{\sure{ x \mapsto \code{onehot}(K) \sep d\mapsto D \sep \own(\cdots)} \sep (q_1 \sim \unif{\{0,1\}^n})}
  } 
  \\
  (\ref{eq:pir-priv})
}{
  \vdash\triple{\sure{ x \mapsto \code{onehot}(K) \sep d\mapsto D \sep \own(\cdots)}}{
    q_1 \samp \unif{\{0,1\}^n} \fatsemi q_2 \coloneqq \xor(q_1, x) \fatsemi (\code{fetch}_1 \parallel \code{fetch}_2) \fatsemi r\coloneqq \xor(r_1, r_2)
  }{\sure{r \mapsto D[K]}}
} 
\]
After dispatching the first sampling command, we move on to analyze the remainder of the program. To do so, we must do case analysis on the sampled bit-string using the \ruleref{Split2} rule, as shown below in (\ref{eq:pir-priv}).
\begin{equation}\label{eq:pir-priv}
\inferrule*[right=\rulereff{Consequence}]{
  \inferrule*[Right=\rulereff{Split2}]{
    \inferrule*[Right=\rulereff{Seq}]{
      \inferrule*[Right=\rulereff{Assign},vdots=3em,rightskip=10em]{\;}{
        \vdash\triple{\sure{ x \mapsto \code{onehot}(K) \sep d\mapsto D \sep q_1 \mapsto Q \sep\own(\cdots)}}{
    q_2 \coloneqq \xor(q_1, x)}{\sure{ x \mapsto \code{onehot}(K) \sep d\mapsto D \sep q_1 \mapsto Q \sep q_2\mapsto \xor(Q, \code{onehot}(K)) \sep\own(\cdots)}}
      } 
      \\
      (\ref{eq:pir-par})
    }{
      \vdash\triple{\sure{ x \mapsto \code{onehot}(K) \sep d\mapsto D \sep q_1 \mapsto Q \sep\own(\cdots)}}{
    q_2 \coloneqq \xor(q_1, x) \fatsemi (\code{fetch}_1 \parallel \code{fetch}_2) \fatsemi r\coloneqq \xor(r_1, r_2)
  }{\sure{r \mapsto D[X]}}
    } 
  }{
    \vdash\triple{\textstyle\bigoplus_{Q \sim \unif{\{0,1\}^n}} \sure{ x \mapsto \code{onehot}(K) \sep d\mapsto D \sep q_1 \mapsto Q \sep\own(\cdots)}}{
    q_2 \coloneqq \xor(q_1, x) \fatsemi (\code{fetch}_1 \parallel \code{fetch}_2) \fatsemi r\coloneqq \xor(r_1, r_2)
  }{\sure{r \mapsto D[K]}}
  } 
}{
  \vdash\triple{\sure{ x \mapsto \code{onehot}(K) \sep d\mapsto D \sep \own(\cdots)} \sep (q_1 \sim \unif{\{0,1\}^n})}{
    q_2 \coloneqq \xor(q_1, x) \fatsemi (\code{fetch}_1 \parallel \code{fetch}_2) \fatsemi r\coloneqq \xor(r_1, r_2)
  }{\sure{r \mapsto D[K]}}
} 
\end{equation}
Next, we analyze the concurrent fetch commands. Since the shared state $d$ is deterministic, it is easy to allocate the invariant with \ruleref{Share}. Then, we use \ruleref{Par} and the derivation for $\code{fetch}_k$ that we showed previously. In the second thread, we must use the \ruleref{Subst} rule, since the query is not $Q$, but rather a more complex logical expression.
\begin{equation}\label{eq:pir-par}
\inferrule*[right=\rulereff{Seq}]{
  \inferrule*[right=\rulereff{Share},vdots=2.5em,rightskip=10em]{
    \inferrule*[Right=\rulereff{Frame}]{
      \inferrule*[Right=\rulereff{Par}]{
        (\ref{eq:pir-fetch})
        \\
        \inferrule*[Right=\rulereff{Subst}]{
          (\ref{eq:pir-fetch})
        }{
          d\mapsto D \vdash\triple{\sure{q_2\mapsto \xor(Q, \code{onehot}(K)) \sep\own(i_2,r_2)}}{\code{fetch}_2}{\sure{r_2 \mapsto \!\!\!\!\!\!\xor_{0 \le i < n : \xor(Q, \code{onehot}(K))[i] = 1}\!\!\!\!\!\!D[i]}}
        }
      }{
        d\mapsto D \vdash\triple{\sure{ q_1 \mapsto Q\sep \own(i_1, r_1)} \sep \sure{q_2\mapsto \xor(Q, \code{onehot}(K)) \sep\own(i_2,r_2)}}{\code{fetch}_1 \parallel \code{fetch}_2}{\sure{r_1 \mapsto \!\!\!\xor_{0 \le i < n : Q[i] = 1}\!\!\!D[i]} \sep \sure{r_2 \mapsto \!\!\!\!\!\!\xor_{0 \le i < n : \xor(Q, \code{onehot}(K))[i] = 1}\!\!\!\!\!\!D[i]}}
      } 
    }{
      d\mapsto D \vdash\triple{\sure{ q_1 \mapsto Q \sep q_2\mapsto \xor(Q, \code{onehot}(K)) \sep\own(\cdots)}}{\code{fetch}_1 \parallel \code{fetch}_2}{\sure{r_1 \mapsto \!\!\!\xor_{0 \le i < n : Q[i] = 1}\!\!\!D[i] \sep r_2 \mapsto \!\!\!\!\!\!\xor_{0 \le i < n : \xor(Q, \code{onehot}(K))[i] = 1}\!\!\!\!\!\!D[i] \sep \own(r)}}
    } 
  }{
    \vdash\triple{\sure{ d\mapsto D \sep q_1 \mapsto Q \sep q_2\mapsto \xor(Q, \code{onehot}(K)) \sep\own(\cdots)}}{\code{fetch}_1 \parallel \code{fetch}_2}{\sure{r_1 \mapsto \!\!\!\xor_{0 \le i < n : Q[i] = 1}\!\!\!D[i] \sep r_2 \mapsto \!\!\!\!\!\!\xor_{0 \le i < n : \xor(Q, \code{onehot}(K))[i] = 1}\!\!\!\!\!\!D[i] \sep d\mapsto D\sep \own(r)}}
  } 
  \\
  (\ref{eq:pir-r})
}{
  \vdash\triple{\sure{ d\mapsto D \sep q_1 \mapsto Q \sep q_2\mapsto \xor(Q, \code{onehot}(K)) \sep\own(\cdots)}}{(\code{fetch}_1 \parallel \code{fetch}_2) \fatsemi r\coloneqq \xor(r_1, r_2)}{\sure{r \mapsto D[K]}}
}\hspace{2em} 
\end{equation}
Finally, we give a derivation for the final assignment.
\begin{equation}\label{eq:pir-r}
\inferrule*[right=\rulereff{Consequence}]{
  \inferrule*[Right=\rulereff{Assign}]{\;}{
    \vdash\triple{\sure{r_1 \mapsto \!\!\!\xor_{0 \le i < n : Q[i] = 1}\!\!\!D[i] \sep r_2 \mapsto \!\!\!\!\!\!\xor_{0 \le i < n : \xor(Q, \code{onehot}(K))[i] = 1}\!\!\!\!\!\!D[i] \sep\own(r)}}{r \coloneqq\xor(r_1,r_2)}{\sure{r\mapsto D[K]}}
  } 
}{
  \vdash\triple{\sure{r_1 \mapsto \!\!\!\xor_{0 \le i < n : Q[i] = 1}\!\!\!D[i] \sep r_2 \mapsto \!\!\!\!\!\!\xor_{0 \le i < n : \xor(Q, \code{onehot}(K))[i] = 1}\!\!\!\!\!\!D[i] \sep d\mapsto D\sep \own(r)}}{r \coloneqq\xor(r_1,r_2)}{\sure{r\mapsto D[K]}}
} 
\end{equation}
Below, we show why $\sure{\xor(r_1, r_2)\mapsto D[K]}$:
\begin{align*}
  \xor\left(\xor_{0\le i< n : Q[i] = 1} D[i], \xor_{0 \le i < n : \xor(Q, \code{onehot}(K))[i] = 1}D[i] \right)
  &= \xor_{i=0}^{n-1} \left\{
    \begin{array}{ll}
      \xor(D[i], D[i]) & \text{if}~ Q[i]= \xor(Q, \code{onehot}(K))[i] = 1
      \\
      D[i] & \text{if}~ Q[i]= \neg \xor(Q, \code{onehot}(K))[i]
      \\
      0 & \text{otherwise}
    \end{array}\right.
  \intertext{Since $\xor(D[i], D[i]) = 0$, we can combine the first and last cases.}
  &= \xor_{i=0}^{n-1} \left\{
    \begin{array}{ll}
      D[i] & \text{if}~ Q[i]= \neg \xor(Q, \code{onehot}(K))[i]
      \\
      0 & \text{otherwise}
    \end{array}\right.
  \intertext{$Q[i]= \neg \xor(Q, \code{onehot}(K))[i]$ occurs only when $i = K$.}
  &= \xor_{i=0}^{n-1} \left\{
    \begin{array}{ll}
      D[i] & \text{if}~ i = K
      \\
      0 & \text{otherwise}
    \end{array}\right.
  \intertext{Since $\xor(x, 0) = x$, then we can drop all the zero terms.}
  &= D[K]
\end{align*}

\subsection{The von Neumann Trick}

Recall the following program from \Cref{sec:von-neumann}, which simulates a fair coin given a coin whose bias can be altered by a parallel thread.
\[
\begin{array}{l}
x \coloneqq 0 \fatsemi \\
y \coloneqq 0 \fse \\
\whl{x=y}{} \\
\quad p' \coloneqq p \fse\\
\quad x \samp\bern{p'} \fse \\
\quad y \samp\bern{p'}
\end{array}
\]
In this section, we provide the complete derivation for the correctness of this program, showing both that it almost surely terminates, and also that $x$ is distributed like a fair coin flip at the end of the program execution. We are going to use the resource invariant $I \triangleq (p\in\ell)$, where $\ell$ is a finite list of values between $\varepsilon$ and $1-\varepsilon$, with $0 < \varepsilon \le \frac12$.
For the purposes of analyzing the while loop, we will use the loop invariant $\varphi$ below:
\[
  \varphi \triangleq \varphi_0 \vee \varphi_1
   \qquad\qquad
   \varphi_0\triangleq \smashoperator{\bigoplus_{X \sim \bern{1/2}}} \sure{x \mapsto X \sep y \mapsto \lnot X \sep R=0}
   \qquad\qquad
    \varphi_1\triangleq\sure{x=y\mapsto \tru \sep R=1 \sep \own(p')}
\]
Note that this gives us $\varphi\Rightarrow\sure{x=y \mapsto R}$, so when $R=0$ the loop will terminate. Clearly, this also implies that $R$ is bounded between 0 and 1. In addition, $\varphi[0/R] = \bigoplus_{X\in\bern{1/2}} \sure{x\mapsto X \sep y\mapsto\lnot X}$ and it is precise.
We start with the derivation of the command $p\coloneqq p'$ in (\ref{eq:vn-inv}), which reads a value from shared state. This derivation proceeds by using the \ruleref{Atom} rule to open the invariant and the \ruleref{Exists} rule to get the specific value of $p'$ at the point of the read.
\begin{equation}\label{eq:vn-inv}
\inferrule*[right=\rulereff{Atom}]{
  \inferrule*[Right={\rulereff{Exists}+\rulereff{Consequence}}]{
    \inferrule*[Right=\rulereff{NSplit1}]{
      \inferrule*[Right=\rulereff{Assign'}]{\;}{
        \vdash_\wk\triple{\sure{p\mapsto X \sep \own(x,y,p')}}{p' \coloneqq p}{\sure{p'\mapsto X\sep p\mapsto X\sep \own(x,y)}}
      }
    }{
      \vdash_\wk\triple{\textstyle\bignd_{X\in \ell} \sure{p\mapsto X \sep \own(x,y,p')}}{p' \coloneqq p}{\textstyle\bignd_{X\in \ell}\sure{p'\mapsto X\sep p\mapsto X\sep \own(x,y)}}
    } 
  }{
    \vdash_\wk\triple{\sure{x=y\mapsto \tru \sep R=1 \sep \own(p')} \sep \sure{p\in\ell}}{p' \coloneqq p}{(\textstyle\bignd_{X\in \ell}\sure{p'\mapsto X\sep \own(x,y)})\sep\sure{p\in\ell}}
  } 
}{
  p\in\ell\vdash_\wk\triple{\sure{x=y\mapsto \tru \sep R=N=1 \sep \own(p')}}{p' \coloneqq p}{\textstyle\bignd_{X\in \ell}\sure{p'\mapsto X\sep \own(x,y)}}
}
\end{equation}
Next, we derive a specification for the two sampling operations in (\ref{eq:vn-samp}). First, the \ruleref{NSplit2} rule is used to compositionally reason about the nondeterministic outcomes. In the premise of that rule, the value of $p'$ is deterministic, so that \ruleref{Samp} can be used twice for the writes to $x$ and $y$. In the end, $x$ and $y$ are independently and identically distributed according to $\bern X$. We complete the proof using the rule of \ruleref{Consequence}, with implication (\ref{eq:vn-cons}).
\begin{equation}\label{eq:vn-cons}
\begin{split}
  \sure{p'\mapsto X} \sep (x \sim \bern{X}) \sep (y \sim \bern{X})
  &\implies \textstyle\sure{p'\mapsto X} \sep \left( \bigoplus_{Y\sim\bern{X}} \sure{x \mapsto Y} \right) \sep \left( \bigoplus_{Z\sim\bern{X}} \sure{y \mapsto Z} \right)
  \\
  &\implies \textstyle \bigoplus_{Y\sim\bern{X}} \bigoplus_{Z\sim\bern{X}} \sure{p'\mapsto X\sep x \mapsto Y\sep y \mapsto Z}
  \\
  &\implies \textstyle \left(\bigoplus_{Y\sim\bern{1/2}} \sure{p'\mapsto X\sep x \mapsto Y\sep y \mapsto \lnot Y}\right) \oplus_{2X(1-X)} \sure{x=y\mapsto \tru \sep \own(p')}
  \\
  &\implies \varphi_0 \oplus_{2X(1-X)} \varphi_1
  \\
  &\implies \varphi_0 \oplus_{\ge 2\varepsilon(1-\varepsilon)} \varphi_1
\end{split}
\end{equation}
Now, the derivation is shown below.
\begin{equation}\label{eq:vn-samp}
\inferrule*[right=\rulereff{NSplit2}]{
  \inferrule*[Right={\rulereff{Consequence},(\ref{eq:vn-cons})}]{
    \inferrule*[Right=\rulereff{Seq}]{
      \inferrule*[right=\rulereff{Samp}]{\;}{
        I\vdash_\wk\triple{\sure{p'\mapsto X\sep \own(x,y)}}{x \samp\bern{p'}}{\sure{p'\mapsto X \sep \own(y)} \sep (x\sim\bern{X})}
      } 
      \inferrule*[right=\rulereff{Samp},vdots=3em,leftskip=15em,rightskip=14em]{\;}{
        I\vdash_\wk\triple{\sure{p'\mapsto X \sep \own(y)} \sep (x\sim\bern{X})}{y \samp\bern{p'}}{\sure{p'\mapsto X} \sep (x\sim\bern{X})\sep (y\sim\bern{X})}
      } 
    }{
      I\vdash_\wk\triple{\sure{p'\mapsto X\sep \own(x,y)}}{x \samp\bern{p'} \fatsemi y\samp\bern{p'}}{\sure{p'\mapsto X} \sep (x\sim\bern{X})\sep (y\sim\bern{X})}
    } 
  }{
    I\vdash_\wk\triple{\sure{p'\mapsto X\sep \own(x,y)}}{x \samp\bern{p'} \fatsemi y\samp\bern{p'}}{\varphi_0 \oplus_{\ge 2\varepsilon (1-\varepsilon)} \varphi_1}
  } 
}{
  I\vdash_\wk\triple{\textstyle\bignd_{X\in \ell}\sure{p'\mapsto X\sep \own(x,y)}}{x \samp\bern{p'} \fatsemi y\samp\bern{p'}}{\varphi_0 \oplus_{\ge 2\varepsilon (1-\varepsilon)} \varphi_1}
} 
\qquad\qquad\qquad
\end{equation}
Note that the use of \ruleref{NSplit2} results in a weak triple, and given that the postcondition of (\ref{eq:vn-samp}) is not precise, it is not possible to use a strong version of the rule. However, after the loop terminates, the postcondition \emph{is} precise, and so we can use \ruleref{Strengthen} later, as seen in (\ref{eq:vn-whl}). Note that for the derivation of the loop body, the only possibility for $R > 0$ is $R=1$, so $\varphi \sep \sure{R = N > 0}$ is equivalent to $\sure{x=y \mapsto \tru \sep\own(p') \sep R=N=1} = \varphi_1 \sep \sure{N=1}$. Now, letting $p = 2\varepsilon(1-\varepsilon)$ be the minimum probability that the rank decreases, the premise of \ruleref{BoundedRank} requires us to show that $\varphi_0$ occurs with probability at least $p$. This derivation is shown below.
\begin{equation}\label{eq:vn-whl}
\inferrule*[right=\rulereff{Consequence}]{
  \inferrule*[Right=\rulereff{Strengthen}]{
    \inferrule*[Right=\rulereff{BoundedRank}]{
      \inferrule*[Right=\rulereff{Seq}]{
        (\ref{eq:vn-inv}) \\ (\ref{eq:vn-samp})
      }{
        I \vdash_\wk\triple{\varphi_1 \sep\sure{N=1}}{p' \coloneqq p \fatsemi x \samp\bern{p'} \fatsemi y\samp\bern{p'}}{\varphi_0 \oplus_{\ge 2\varepsilon (1-\varepsilon)} \varphi_1}
      }
    }{
      I\vdash_\wk\triple{\textstyle\bignd_{R\in\{0,1\}} \varphi}{\whl{x = y}{p' \coloneqq p \fatsemi x \samp\bern{p'} \fatsemi y\samp\bern{p'}}}{\varphi_0}
    }
  }{
    I\vdash\triple{\textstyle\bignd_{R\in\{0,1\}} \varphi}{\whl{x = y}{p' \coloneqq p \fatsemi x \samp\bern{p'} \fatsemi y\samp\bern{p'}}}{\varphi_0}
  }
}{
  I\vdash\triple{
    \sure{x \mapsto 0 \sep y\mapsto 0 \sep \own(p')}
  }{
    \whl{x = y}{p' \coloneqq p \fatsemi x \samp\bern{p'} \fatsemi y\samp\bern{p'}}
  }{x \sim \bern{1/2}}
}
\end{equation}
Finally, we include the initial assignments to $x$ and $y$ to complete the proof.
\[
\inferrule*[right=\rulereff{Seq}]{
  \inferrule*[right=\rulereff{Assign}]{\;}{
    I\vdash\triple{\sure{\own(x,y,p')}}{x \coloneqq 0}{\sure{x\mapsto 0 \sep \own(y,p')}}
  } 
  \quad
  \inferrule*[Right=\rulereff{Seq},vdots=3em,leftskip=5em]{
    \inferrule*[right=\rulereff{Assign}]{\;}{
      I\vdash\triple{\sure{x\mapsto 0 \sep \own(y,p')}}{y \coloneqq 0}{\sure{x\mapsto 0 \sep y\mapsto 0 \sep \own(p')}}
    } 
    \\
    (\ref{eq:vn-whl})
  }{
    I\vdash\triple{\sure{x\mapsto 0 \sep \own(y,p')}}{y\coloneqq 0\fatsemi \whl{x = y}{p' \coloneqq p \fatsemi x \samp\bern{p'} \fatsemi y\samp\bern{p'}}}{x \sim \bern{1/2}}
  }
}{
  I\vdash\triple{\sure{\own(x,y,p')}}{x \coloneqq 0\fatsemi y\coloneqq 0\fatsemi \whl{x = y}{p' \coloneqq p \fatsemi x \samp\bern{p'} \fatsemi y\samp\bern{p'}}}{x \sim \bern{1/2}}
}
\]

\end{landscape}
\fi

\end{document}